\documentclass[12pt,english]{article}
\usepackage[a4paper, total={6in, 9in}]{geometry}
\usepackage{amsfonts,amsthm,amsmath,amssymb,latexsym,comment}
\usepackage{float}
\usepackage{amssymb}
\usepackage{amsmath}
\usepackage{eurosym}
\usepackage{hyperref}
\usepackage{hhline}
\usepackage[pdftex]{graphicx}
\usepackage[square]{natbib}
\newcommand{\commentout}[1]{}
\newtheorem{proposition}{Proposition}
\newtheorem{corollary}{Corollary}
\newtheorem{theorem}{Theorem}

\newtheorem{lemma}{Lemma}
\newtheorem{example}{Example} 
\newtheorem{definition}{Definition}

\newcommand{\ACT}{\mathit{CT}}
\newcommand{\nciteyear}[1]{[\citeyear{#1}]}
\renewcommand{\cite}{\citep*}
\newcommand{\fullv}[1]{#1}
\newcommand{\shortv}{\commentout}

\date{}

\setcitestyle{acmnumeric}

\title{Communication games, sequential equilibrium, and mediators}

\author{
Ivan Geffner \\
Cornell University\thanks{Supported in part by NSF grants IIS-1703846 and 
IIS-0911036, ARO grant W911NF-17-1-0592,  MURI grant W911NF-19-1-0217
from the ARO, and a grant from Open Philanthropy.}\ \thanks{Declarations of interest: none.}
\and
Joseph Y. Halpern\\
Cornell University\footnotemark[1]\ \footnotemark[2]
}

\begin{document}

\begin{titlepage}

\maketitle

\begin{abstract}
We consider \emph{$k$-resilient sequential equilibria},
   strategy profiles where no player in a coalition of at most $k$
    players believes that it can increase its utility by deviating,
   regardless of its local state.
      We prove  that all $k$-resilient
    sequential equilibria that can be implemented with
    a trusted
    mediator can also be implemented without the mediator  
    in a synchronous system of $n$ players if $n >3k$.  In asynchronous
systems, where there is no global
    notion of time and messages may take arbitrarily long to get to
    their recipient, we prove that a 
    $k$-resilient sequential 
    equilibrium with a mediator can be implemented
    without the mediator if $n > 4k$.  
        These results match the lower bounds given by Abraham, Dolev, and
    Halpern~\nciteyear{ADH07} and Geffner and Halpern~\nciteyear{GH21}
    for implementing a Nash equilibrium without a mediator (which
    are easily seen to apply to implementing a sequential equilibrium)
    and improve the results of Gerardi
        \nciteyear{Gerardi04}, who showed that, in the case that $k=1$, a
        sequential equilibrium can be implemented
        in synchronous systems
        if $n \ge 5$.  
\end{abstract}

\end{titlepage}

\section{Introduction}

A standard technique for dealing with complicated problems in
distributed computing, cryptography, and game theory is to show that
the problem can be solved under the assumption that there is a
\emph{trusted mediator}, a third party to whom all agents in the
system can communicate, and then showing that we can \emph{implement}
the mediator; that is, the agents just communicating among themselves
(using what economists call \emph{cheap talk}) can simulate what the
mediator does.  In the 
distributed computing and
cryptography literature, the focus has been on
\emph{secure multiparty computation}  \cite{GMW87,yao:sc}. 
\commentout{
The goal
is to compute some function $f(x_1, \ldots, x_n)$, where there are $n$
agents, and $x_i$ is agent $i$'s input.  
}
Given $n$ agents in which each agent $i$ has input $x_i$, the goal is
to compute some function $f(x_1, \ldots, x_n)$ without revealing
anything about the agents' inputs besides what can be deduced from the
output of the function itself. 
With a mediator, this is
trivial: each agent $i$ 
just tells the mediator $x_i$, and the mediator tells the agents
$f(x_1,\ldots, x_n)$.  It is known that, as long as no
more than $k$ agents are malicious, this mediator can be
implemented if $n > 3k$.

Game theory adds incentives to the mix.  While in the work done in
cryptography, it is simply assumed that malicious agents will do what
they can to bring down the system, in the game-theoretic setting, it
is assumed that agents will deviate from a strategy only if it is
in their best interests to do so.  
To make this precise, in the game
theory literature, three games were considered: an \emph{underlying}
normal-form game $\Gamma$, a game $\Gamma_d$ with a mediator, where,
after communicating with the mediator, players make a move in $\Gamma$
and get the same payoffs that they do in $\Gamma$, and  a
\emph{cheap-talk game} $\Gamma_{\ACT}$, where players just talk to
each other and then make a move in the underlying game.  Forges
\nciteyear{F90} and Barany 
\nciteyear{Barany92} showed that a Nash equilibrium (NE) in the mediator
game $\Gamma_d$ 
could be \emph{implemented} by a NE
in the cheap-talk game
if
$n \ge 4$ 
in the sense that the distribution over action profiles in the NE in
$\Gamma_d$ is the same as that in the NE in $\Gamma_{\ACT}$. 
\commentout{
This
result was later generalized  to show that a \emph{$k$-resilient
  equilibrium} (one 
that tolerates deviations by coalitions of size up to $k$) in
$\Gamma_d$ can be implemented by a $k$-resilient equilibrium in
$\Gamma_{\ACT}$ if $n > 3k$ in the synchronous case
\cite{ADGH06} and if $n > 4k$ in the asynchronous case 
\cite{ADGH19}.  
}

As is well known, Nash
equilibrium considers only deviations on the \emph{equilibrium path}
(situations that arise with positive probability if the NE
is played).  
In extensive-form games such as $\Gamma_{\ACT}$,
Nash equilibrium does not always describe what intuitively would be 
reasonable play.
For example
if $A$ tells $B$ to give her  \$1 or she will destroy the world,
$B$ giving \$1 to $A$ and $A$ not destroying the world is a NE.
However, this equilibrium is based on a non-credible
threat.  
Would 
$A$
really destroy the world if
she doesn't get \$1?  For this reason, game theorists have 
considered solution concepts such as \emph{sequential equilibrium}
\cite{KW82}, 
where players cannot increase their utility by deviating even off the
equilibrium path 
(see Section~\ref{def:seq-eq} for a formal definition).  

Not surprisingly, the question of whether a sequential equilibrium
with a mediator can be implemented has been considered before in the game
theory literature.  Ben-Porath \nciteyear{Bp03} claimed that a
sequential equilibrium could be implemented in $\Gamma_{\ACT}$ if $n \ge 3$
provided that there is a \emph{punishment strategy} (a way for players
to punish players who are caught cheating---see
Section~\ref{sec:results} for a formal definition);%
\footnote{Unfortunately, there is a serious error in Ben-Porath's
  proof; see \cite{ADH07, 
        BenPorathNote21}.  A minor modification of Ben-Porath's proof
  does deal with the case 
  that $n > 3$ \cite{BenPorathNote21}.}
  Gerardi \nciteyear{Gerardi04} showed that,
in the case that $k=1$, a sequential equilibrium can be implemented if $n \ge 5$.  

Nash and sequential equilibria incentivize individual participants to
follow the 
given strategy since 
otherwise they would get a lower expected payoff.
However, if players can form coalitions and deviate in a coordinated
way, then a coalition of players may have incentive to deviate in a
Nash or sequential equilibrium.
For instance, consider a
normal-form game for $n$ players in which players can play either $0$
or $1$. If at least $n-1$ players play $0$, they all get $0$ utility,
but if two or more players play $1$, these players get a utility of
$1$ and everyone else gets a utility of $-1$. Clearly, playing $0$ is
a Nash equilibrium of this game. However, if two players can deviate
together, they can increase their payoff by playing $1$ instead of
$0$.  
Designing mechanisms that incentivize not only individuals, but also
coalitions to act as intended has become increasingly important,
especially in 
Internet applications (e.g., blockchain), where the same person can
maintain several identities. 
Not surprisingly, coalitions have received significant attention recently in the
literature (see, e.g., \cite{ES14, 
    Heller05}). 
Abraham, Dolev, Gonen and Halpern~\nciteyear{ADGH06} (ADGH
from now on) introduced the notion of \emph{$k$-resilient Nash equilibrium},
which is a strategy profile in which no coalition of up to $k$ players
can increase their payoff by deviating. They also generalized Forges
and Barany's results by showing that any $k$-resilient Nash
equilibrium in $\Gamma_d$ can be implemented in $\Gamma_{\ACT}$ if the
number of players $n$ satisfies $n > 3k$, and proved that the $n > 3k$
bound is tight~\cite{ADH07}. More precisely, if $n \le 3k$, then there
exists a game $\Gamma_d$ for $n$ players and a mediator and a
$k$-resilient Nash equilibrium $\vec{\sigma}$ in $\Gamma_d$ that can't
be implemented in $\Gamma_{\ACT}$. 

Two natural questions follow from these results. First, whether we can
generalize Gerardi's result and give necessary and sufficient conditions on
$n$ and $k$ to implement a $k$-resilient sequential equilibrium (which is
defined analogously to $k$-resilient Nash equilibrium), and second, whether
the $n \ge 5$ upper bound given by Gerardi is actually tight. It is
easy to check that the $n > 3k$ lower bound given in \cite{ADH07} for
Nash equilibrium applies without change to $k$-resilient sequential
equilibria. In particular, if $k = 1$, it shows that at least $4$
players are required to implement a sequential equilibrium, but it was
not known whether $n \ge 5$ is necessary.
Here, we show that we can
always implement 
a $k$-resilient sequential equilibrium in a game with a mediator and $n$ players if $n
> 3k$, matching the lower bound given in \cite{ADH07} and improving
Gerardi's result for $k = 1$. 

Following
Gerardi \nciteyear{Gerardi04} and Gerardi and Myerson \nciteyear{GM05}, we
also relate our results to two other solution concepts in normal-form
games: \emph{correlated equilibrium} \cite{Aumann87} and
\emph{communication equilibrium} \cite{F86,Myerson86}.  Both of these concepts
can be understood in terms of games with mediators.  A correlated
equilibrium is an equilibrium in a game with a mediator where the
players do not talk to the mediator; the mediator simply tells 
the agent which strategy to play.   A \emph{communication
equilibrium} in a Bayesian game is an equilibrium  in a game with
a mediator where the players can tell the mediator their types, and
the mediator then tells the players what strategy to 
play (see
Appendix~\ref{sec:definitions} 
for formal definitions).
\commentout{
As shown by Gerardi \nciteyear{Gerardi04}, if $\Gamma$ is a
normal-form game, the set of sequential equilibria in $\Gamma_d$ is
$\Gamma$ is a Bayesian game, the set of communication equilibria in
$\Gamma_d$ is the set of correlated equilibria of $\Gamma$. These
results can be easily generalized to show that the set of
$k$-resilient sequential equilibria in $\Gamma_d$ are the sets of
$k$-resilient correlated equilibria or the sets of $k$-resilient
communication equilibria of $\Gamma$ depending if $\Gamma$ is a
normal-form or a Bayesian game respectively. Thus, in 
}
Our results give a
characterization of the set $SE_k(\Gamma_{\ACT})$ of outcomes that can be
implemented with $k$-resilient sequential equilibria in
the cheap-talk game
$\Gamma_{\ACT}$; if $\Gamma$ is a normal-form game, then
$SE_k(\Gamma_{\ACT})$ is the set of $k$-resilient correlated equilibria
of $\Gamma$, while if $\Gamma$ is a Bayesian game,
$SE_k(\Gamma_{\ACT})$ is the set of $k$-resilient communication
equilibria of $\Gamma$. 

Sequential equilibrium involves describing not only what the players
do at each point in an extensive-form game, but describing what their
beliefs are, even off the equilibrium path.  It must be shown that
players are always best responding to their beliefs.  The key
difficulty in our proof involves finding appropriate beliefs
that are consistent with the equilibrium strategy.
The proof given by Gerardi involves an extensive case analysis for
both consistency and sequential rationality, which makes it hard to
generalize to other settings. Our proof introduces an interesting new
technique. 
We show that all strategies in $\Gamma_{\ACT}$ (even those that are not
an equilibrium) admit a consistent \emph{$k$-paranoid belief system},
where all coalitions $K$ 
of at most $k$ players always believe that, if
other players deviated, they did so by sending inappropriate messages
only to players in
$K$.  That is, all coalitions of size at most $k$ believe that the
remaining players are being truthful among themselves.
We show that, given a $k$-resilient Nash equilibrium $\vec{\sigma}$,
we can extend 
it to a $k$-resilient sequential equilibrium by constructing a
$k$-paranoid belief 
system.  The idea is that, given these $k$-paranoid beliefs, players
in a coalition 
$K$ will not believe that there is anything that they can do to
prevent the remaining players from playing their part of the equilibrium.
We then apply these results to the $k$-resilient Nash equilibrium of \cite{ADGH06} to get our result.

\fullv{
All of the results presented so far assume a \emph{synchronous}
setting: communication 
proceeds in atomic rounds, and all messages sent during round $r$
are received by round $r+1$.  In an \emph{asynchronous} setting,
there are no rounds and messages sent by the players may take
arbitrarily long to get to their 
recipients. 
Asynchrony is a standard assumption in 
the distributed computing
and cryptography
literature, precisely because many systems
that practitioners care about,
such as markets or the internet,
are asynchronous in practice.
Considering asynchronous systems can have significant implications for how
players will play a game.
For instance, in an online second-price auction, the seller can
benefit from inserting fake transactions whose value is between that
of the
highest and second-highest bid immediately after a new highest bid
is received, thus increasing the second-highest price at no 
cost~\cite{Roughgarden2020}.
This type of
attack can be carried out only in asynchronous or \emph{partially synchronous}
systems (where there is an upper bound on how long messages take to
arrive), since in a synchronous system, all bids are received at the
same time. (In a synchronous system, the seller would have to guess
what the players will bid in order to benefit from a fake transaction.) 
\commentout{
For instance, all blockchain implementations assume \emph{partial
    synchrony}, where there is an upper 
bound on how long messages take to arrive.
}

This example shows that asynchrony
and partial synchrony allow players
to influence the communication pattern in ways that they
can't in a synchronous setting.
Other examples include the fact that,
in a fully asynchronous
system, a player that didn't send a message is indistinguishable from
a player whose message is delayed (which means that players can
safely deviate by not sending messages for a while, which they cannot
do in a synchronous system), 
and that,
in a partially synchronous
system, players may read other players' messages before sending their own and
adapt their strategy accordingly. As a result, 
the number of honest players required to implement a given
functionality is greater in an asynchronous setting. For instance, the
upper bound for asynchronous secure computation is $n > 4k$
\cite{BCG93}, as opposed to $n > 3k$ bound in synchronous
systems proved by Ben-Or, Goldwasser, and Wignerson \nciteyear{BGW88}.  

The simplicity of our approach allows us to generalize our main result
to the asynchronous setting with very little work. 
Given a $k$-resilient sequential equilibria $\vec{\sigma}$ in an
asynchronous game with $n$ players and a mediator, we extend 
implementation given by Abraham et al. \nciteyear{ADGH19} of asynchronous
$k$-resilient Nash equilibria in $\Gamma_{\ACT}$ for $n > 4k$
by constructing a consistent $k$-paranoid belief
system. Reasoning analogous to that used in the synchronous case
shows that the resulting strategy and belief system form a
$k$-resilient sequential equilibria. In the asynchronous setting, this
result is also optimal, since it matches the $n > 4k$ lower
bound given by Geffner and Halpern~\nciteyear{GH21}. We believe that
the techniques presented in this paper can be applied to design
$k$-resilient sequential equilibrium in many other scenarios as well. 
}

The rest of the paper is organized as follows.
In Section~\ref{sec:resilience}, we provide the concepts and
definitions required to understand the main results and its proofs
(some of the most basic definitions can be found in
Appendix~\ref{sec:definitions}). In Section~\ref{sec:results}, we state
the main result and give a short outline of its proof. The
details are given in Section~\ref{sec:proof}.
In Section~\ref{sec:descriptions}, we provide our characterization of
$SE_k(\Gamma_{CT})$.  
\fullv{
In Section~\ref{sec:sync-async}, we carefully define asynchronous
systems and extend the results of
Section~\ref{sec:sync-async} to the asynchronous case; these
results are further extended in Section~\ref{sec:extension} assuming
that players can send arbitrarily precise real numbers in their
messages (rather than just rational numbers, which is what we assume
up to this point).
}
\shortv{These results are extended in Section~\ref{sec:extension} assuming
that players can send arbitrarily precise real numbers in their
messages (rather than just rational numbers, which is what we assume
up to this point). }

\section{$k$-resilient equilibrium}\label{sec:resilience}

In this section, we extend the different types of equilibria that
appear in the literature (which are discussed in
Appendix~\ref{sec:definitions}) to account for deviations of
coalitions of at most $k$ players.
For future reference, in normal-form games, we use $n$ to denote the number of players, $P$ to denote the set of players, $A$ to denote the set of possible action profiles, and $U$ to denote the tuple of utility functions. In Bayesian games, we use $T$ to denote the possible type profiles and $q$ to denote the prior distribution over $T$. Moreover, in extensive-form games, we use $G$ to denote the game tree, $M$ to denote the function that outputs which player moves at each node, and $R$ to denote the tuple of equivalence relations between the nodes in which two nodes $v$ and $w$ are equivalent according to the $i$th relation if $i$ cannot distinguish these between $v$ and $w$ (i.e., if $v$ and $w$ are in the same partition). The full definitions can be found in Appendix~\ref{sec:definitions}.

Intuitively, traditional notions of
equilibria guarantee that no individual player can increase her own
payoff by deviating from the proposed strategy. For each of these
equilibria we consider also its $k$-resilient variant, which states
that no coalition of at most $k$ players can jointly increase their
payoff by deviating, even if they do so in a coordinated way. We also
consider the notion of \emph{strong $k$-resilience}, which guarantees
that no individual player inside the coalition can increase its own
payoff, even at the expense of other coalition members (as opposed to the
original notion, in which no member of the coalition can be worse than
by following the proposed strategy).  

\subsection{$k$-resilient Nash, sequential, Bayesian Nash, and
  communication equilibrium} 

We begin by extending the definition of Nash equilibrium
(Definition~\ref{def:nash-eq}) to $k$-resilient Nash
equilibrium:

\begin{definition}
In a normal-form game $\Gamma$, a (mixed) strategy profile
$\vec{\sigma} = (\sigma_1, \ldots, \sigma_n)$ is a \emph{$k$-resilient
Nash equilibrium} (resp., \emph{strongly $k$-resilient Nash equilibrium})  
if, for all coalitions $K$ 
such that $|K| \le k$,
and all strategies $\vec{\sigma}'_K \in \Delta(A_K)$
$$u_i(\vec{\sigma}_K, \vec{\sigma}_{-K}) \ge u_i(\vec{\sigma}'_K,
\vec{\sigma}_{-K})$$
for some (resp., for all) $i \in K$.
\end{definition}

We can similarly extend the notion of correlated equilibrium
(see Definition~\ref{def:correlated-eq} of
Appendix~\ref{sec:definitions}) to $k$-resilient correlated 
equilibrium: 

\begin{definition}\label{def:k-correlated}
Given a normal-form game $\Gamma = (n, A, U)$, a distribution 
$p \in \Delta(A)$ is a \emph{$k$-resilient correlated equilibrium} (resp., \emph{strongly $k$-resilient correlated equilibrium}) if, for
all subsets 
$K \subseteq P$ such that $|K| \le k$,
and all 
action profiles
  $\vec{a}'_K, \vec{a}''_K$ for
players 
in $K$ such that $p(\vec{a}'_K) > 0$,  $$\sum_{\vec{a} \in \Delta(A) :
    \vec{a}_K = \vec{a}'_K} u_i(\vec{a}'_K, \vec{a}_{-K})p(\vec{a}\mid
\vec{a}_K = \vec{a}'_K) \ge \sum_{\vec{a} \in \Delta(A) : 
  \vec{a}_K = \vec{a}'_K} u_i(a''_K, a_{-K})p(\vec{a}\mid \vec{a}_K = \vec{a}'_K)$$
  for some (resp., for all) $i \in K$.
\end{definition}

Simply put, a distribution 
$p \in \Delta(A)$ 
is a $k$-resilient correlated
equilibrium if no 
subset of at most $k$ players
would be better off by
deviating if they knew that the 
action profile
used was sampled according
to $p$ (and they knew their components of the profile), even if they could
coordinate their deviations. 

The definition of $k$-resilient Nash equilibrium in extended-form
games is equivalent. However, the extension of sequential equilibrium
(Definition~\ref{def:seq-eq}) to $k$-resilient sequential equilibrium
requires additional definitions \cite{ADH13}.  

\begin{definition}
 Let $\Gamma = (P,G,M,U,R)$ be an extensive-form game and $K \subseteq
P$ be a subset of players. We define the equivalence relation $\sim_K$
by $v \sim_K v'$ iff $v \sim_i v'$ for all $i \in K$.%
\footnote{The $\sim_i$ relation defines player $i$'as information
sets in an extensive-form game.  Se Appendix~\ref{sec:definitions} for
the formal definition.} 
The equivalence
classes of $\sim_K$ are called \emph{$K$-information sets}. 
\end{definition}

Intuitively, two nodes are related according to $\sim_K$ if they are indistinguishable for all players $i \in K$. With this, we can extend the notion of belief systems to account for coalitions of players that share information between themselves.

\begin{definition}
A \emph{$k$-belief system} $b$ in a game $\Gamma$ is a function that
maps each $K$-information set $I$ such that $K \subseteq P$ and $|K|
\le k$ to a distribution over the nodes in $I$.
\end{definition}

We can define what it 
means
for a $k$-belief system to be consistent
with $\vec{\sigma}$ just as in the case of 
standard
belief systems.
We say that a $k$-belief system $b$ is consistent with strategy
$\vec{\sigma}$ if there exists a sequence $\vec{\sigma}^1,
\vec{\sigma}^2, \ldots$ of completely-mixed strategies that converges
to $\vec{\sigma}$ such that the beliefs induced by Bayes' rule
converge to $b$. 
With these definitions in hand, we can find the desired generalization
of sequential equilibrium.

\begin{definition}\label{def:k-sequential}
  A pair $(\vec{\sigma},b)$ consisting of a strategy profile
$\vec{\sigma}$ and a $k$-belief system $b$ consistent with $\vec{\sigma}$
is a 
\emph{$k$-resilient sequential equilibrium}
(resp., \emph{strongly $k$-resilient sequential equilibrium}) 
if, for all 
$K \subseteq P$ with $|K| \le k$,
all $K$-information sets $I$, 
and all strategies
$\tau_K$ for players in $K$ 
(w.r.t. $\sim_K$),
$$u_i(\vec{\sigma}, I, b) \ge u_i((\tau_K, \vec{\sigma}_{-K}), I, b)$$
for some (resp., for all) $i \in K$.
\end{definition}

Note that the notions of $1$-resilient correlated equilibrium and $1$-resilient
sequential equilibrium are equivalent to the standard
notions of correlated equilibrium and sequential equilibrium,
respectively. 
Note that, in games with a mediator $d$, since
the mediator never deviates from its strategy, a belief system $b$ is
consistent with strategy $\vec{\sigma} + \sigma_d$ (where
$\vec{\sigma} + \sigma_d$ means that players play strategy profile
$\vec{\sigma}$ and the mediator plays strategy $\sigma_d$) if there
exists a sequence $\vec{\sigma}^1 + \sigma_d, 
\vec{\sigma}^2 + \sigma_d, \ldots$ of completely-mixed strategies that converges
to $\vec{\sigma} + \sigma_d$ such that the beliefs induced by Bayes' rule
converge to $b$. The definition of $k$-resilient sequential
equilibrium in mediator games is identical to the one given in
Definition~\ref{def:k-sequential}, using this definition of consistent
beliefs. Intuitively, this means that the players believe that
the mediator never deviates from its protocol. 

In a Bayesian game with type space $T$ (see
Appendix~\ref{sec:bayesian} for the formal definitions), a
\emph{correlated strategy profile} is a map $\mu : T 
\rightarrow \Delta(A)$. Note that all distributions over strategy
profiles can be viewed as correlated 
strategy profiles, but the converse is not true. For instance, in a game
with two players in which $A_i = T_i = \{0,1\}$ for all $i$, a
correlated strategy profile may consist of both players playing action $t_1 +
t_2 \bmod 2$; this cannot be represented by a strategy profile, since
players must independently 
choose
their action given their type.

The expected utility of player $i$ in a coalition $K$
when playing (a possibly correlated) strategy profile $\vec{\mu}$,
where $T_K$ is the set of possible 
type profiles
of the players in $K$, is
$$u_i^K(\vec{\mu}) = \sum_{\vec{t}_K \in T_K} q(\vec{t}_K) \sum_{\vec{t} \in T}
q(\vec{t}\mid \vec{t}_K ) u_i(\vec{\mu}(\vec{t})),$$ 
where 
$u_i(\vec{\mu}(\vec{t}))$ denotes the expected utility of player
$i$ when an action profile 
is chosen
according to 
$\mu(\vec{t})$.
Intuitively, we are assuming that players in $K$ can share their
types, which is why we condition on $\vec{t}_K$.  
With this, we can extend Bayesian Nash equilibrium
(see Definition~\ref{def:bayesian-nash} in Appendix~\ref{sec:definitions}) to $k$-resilient Bayesian Nash
equilibrium:  

\begin{definition}
In a Bayesian game $\Gamma = (P,T, q, A, U)$, a strategy profile
  $\vec{\mu} := (\mu_1, \ldots, \mu_n)$ 
is a \emph{$k$-resilient Bayesian Nash equilibrium} (resp., \emph{strongly
$k$-resilient Bayeisan Nash equilibrium}) if,  
for
all coalitions $K$ with $|K| \le k$ and
all maps $\mu'_K : T_K \rightarrow
\Delta(A_K)$, 
$$u_i^K(\vec{\mu}) \ge u_i^K(\vec{\mu}_{-K}, \mu'_K)$$
for some (resp., for all) $i \in K$.
\end{definition}

We can similarly generalize communication equilibrium
(see Definition~\ref{def:comm-eq} of
Appendix~\ref{sec:definitions})) to $k$-resilient communication
equilibrium: 

\begin{definition}\label{def:k-comm}
    A
     correlated strategy profile
    $\mu : T \rightarrow \Delta(A)$ is a \emph{$k$-resilient communication
equilibrium} (resp., \emph{strongly $k$-resilient communication equilibrium} of Bayesian Game $\Gamma = (P,T, q, A, U)$ if, for all $K
\subseteq P$
such that $|K| \le k$,
   all $\vec{t}_K \in T_K$,  all $\psi: T_K \rightarrow T_K$, and all
$\varphi : A_K \rightarrow A_K$, we have that $$\sum_{\vec{t}_{-K} \in T_{-K}}
\sum_{\vec{a} \in A} q(t_{-K}, \vec{t}_K)\mu(\vec{a} \mid \vec{t}_{-K},
\vec{t}_K)u_i(\vec{t}_{-K}, \vec{t}_K,\vec{a}) \ge $$ $$\sum_{\vec{t}_{-K}
  \in T_{-K}} \sum_{\vec{a} \in A} q(t_{-K}, \vec{t}_K)\mu(\vec{a} \mid
\vec{t}_{-K}, \psi(\vec{t}_K))u_i(\vec{t}_{-K},
\vec{t}_K,\vec{a}_{-K}, \varphi(\vec{a}_K))$$
for some (resp., for all) $i \in K$.
\end{definition}

As in the case of the standard communication equilibrium, we can
identify a correlated strategy profile $\mu$ with a canonical game
with a mediator in which each player $i$ must send its type $t_i$ to
the mediator, then the mediator samples an action profile $\vec{a}$
from $\mu(\vec{t})$ and sends $a_i$ to each player $i$, and then each
player $i$ plays whatever it received from the mediator.  
In Definition~\ref{def:k-comm}, the functions $\psi$ and
$\varphi$ represent possible deviations for players in $K$ in this canonical game; $\psi$
describes how they might lie to the mediator about their type (by sending type
profile
$\psi(\vec{t}_K)$ 
instead of their actual type profile
$\vec{t}_K$), while $\varphi$ represents how 
they
might deviate in what
they do 
(playing strategy profile $\varphi(\vec{a}_K)$ instead of the
strategy profile $\vec{a}_K$ suggested by the
mediator). Definition~\ref{def:k-comm} says that $\mu$ is a
$k$-resilient (resp., strongly $k$-resilient) communication
equilibrium if none of these deviations is 
profitable for all players (resp., for some player) in $K$. 

\commentout{
In this section, we review the basic concepts and definitions that
are needed for this paper.
\commentout{
We begin by reviewing basic
game-theoretic definitions, and then review some distributed computing 
primitives that are used  in our protocols.
}
We begin by presenting
several types of games that are common in
the literature.  For each of these
games, we also 
introduce some of 
the \emph{solution concepts} that we are interested in. A solution
concept describes what rational agents should do in a game, according
to some theory  of rationality.  Typically, a solution concept
determines for each game a set of \emph{strategy profiles},  
where a strategy profile $\vec{s} = (s_1, \ldots, s_n)$ is a tuple
consisting of one strategy for each agent, and a strategy is 
a description of how each player should
move/act in all ``situations'' in the given game.  (As we shall see,
what counts as a situation depends 
on
the type of game we consider.) 



A normal-form game $\Gamma$ is a tuple $(P, A, U)$, where
$P =
\{1, \ldots, n\}$ is the set of \emph{players}, $A = A_1 \times \cdots 
A_n$, where  $A_i$ is the set of possible 
actions
for
player $i \in P$, and $U = (u,
\ldots, u_n)$ is an $n$-tuple of \emph{utility}
functions $u_i: A \rightarrow \mathbb{R}$, again, one for each player
$i \in P$.  
A
\emph{pure strategy
  profile} 
  $\vec{a}$
  is
  an $n$-tuple of 
  actions
  $(a_1, \ldots,
  a_n)$, with $a_i \in A_i$.
  A \emph{mixed strategy} for player $i$ is
an element of $\Delta(A_i)$, the set of probability distributions
on $A_i$.  We extend $u_i$ to mixed strategy profiles
$\sigma = (\sigma_1, \ldots,
\sigma_n)$ by defining
$u_i(\sigma) = 
\sum_{(a_1,\ldots,a_n) \in A} \sigma_1(a_1)\ldots\sigma_n(a_n) u_i(a_1,
\ldots, a_n)$: that is, the sum over all pure strategy profiles $\vec{a}$ of the
probability of playing $\vec{a}$ according to
$\sigma$ times
the utility 
of $\vec{a}$ to $i$.

\begin{definition}
In a normal-form game $\Gamma$, a (mixed) strategy profile
$\vec{\sigma} = (\sigma_1, \ldots, \sigma_n)$ is a \emph{Nash equilibrium}
if, for all $i \in P$ and strategies 
$\sigma_i' \in \Delta(A_i)$,
$$u_i(\sigma_i, \vec{\sigma}_{-i}) \ge u_i(\sigma'_i,
\vec{\sigma}_{-i}).$$ 
\end{definition}

\begin{definition}
  In a normal-form game $\Gamma = (P, A, U)$, a distribution
    $p \in \Delta(A)$ is a \emph{correlated equilibrium} (CE) if,
  for all  
  $i \in P$
and all $a'_i, a''_i \in A_i$ such that $p(a'_i) >
0$,  $$\sum_{\vec{a} \in \Delta(A) : a_i = a'_i} u_i(a'_i,
\vec{a}_{-i})p(\vec{a} \mid a_i = a'_i) \ge \sum_{\vec{a} \in \Delta(A)
    : a_i = a'_i} u_i(a_i'', 
\vec{a}_{-i})p(\vec{a} \mid a_i = a'_i).$$ 
\end{definition}

Intuitively, for a distribution $p$ over 
action profiles
to be a
correlated equilibrium, it cannot be worthwhile
for player $i$ to deviate from 
$a_i$ to $a_i'$
if $i$ knows only that 
$\vec{a}$ is sampled from distribution $p$ (and $a_i$).
Note that all Nash equilibria are correlated equilibria as well.
We can extend 
these definitions
to 
a setting in which coalitions of $k$ players may deviate in a
coordinated way:

\begin{definition}
In a normal-form game $\Gamma$, a (mixed) strategy profile
$\vec{\sigma} = (\sigma_1, \ldots, \sigma_n)$ is a \emph{$k$-resilient Nash equilibrium}
if, for all coalitions $K$ of size at most $k$, 
all $i \in K$, and all strategies $\vec{\sigma}'_K \in \Delta(A_K)$
$$u_i(\vec{\sigma}_K, \vec{\sigma}_{-K}) \ge u_i(\vec{\sigma}'_K,
\vec{\sigma}_{-K}).$$ 
\end{definition}

\begin{definition}\label{def:k-correlated}
Given a normal-form game $\Gamma = (n, A, U)$, a distribution 
$p \in \Delta(A)$ is a \emph{$k$-resilient correlated equilibrium} if, for
all subsets 
$K \subseteq P$ such that $|K| \le k$,
all $i \in K$,
and all 
action profiles
  $\vec{a}'_K, \vec{a}''_K$ for
players 
in $K$ such that $p(\vec{a}'_K) > 0$,  $$\sum_{\vec{a} \in \Delta(A) :
    \vec{a}_K = \vec{a}'_K} u_i(\vec{a}'_K, \vec{a}_{-K})p(\vec{a}\mid
\vec{a}_K = \vec{a}'_K) \ge \sum_{\vec{a} \in \Delta(A) : 
  \vec{a}_K = \vec{a}'_K} u_i(a''_K, a_{-K})p(\vec{a}\mid \vec{a}_K = \vec{a}'_K).$$ 
\end{definition}
Simply put, a distribution 
$p \in \Delta(A)$ 
is a $k$-resilient correlated
equilibrium if no 
subset of at most $k$ players
would be better off by
deviating if they knew that the 
action profile
used was sampled according
to $p$ (and they knew their components of the profile), even if they could
coordinate their deviations. 

An \emph{extensive-form game} $\Gamma$ is a tuple $(P, G, M, U, R)$, where
\begin{itemize}
  \item
    $P = \{1, \ldots, n\}$ is the set of players.
\item $G$ is a
game 
tree
in which each node represents a possible state of the game and
each of the edges going out of a node an action that a player
can perform in the state corresponding to that node. (In the sequel,
we refer to the edges as actions.)
\item
  $M$ is a 
function that associates with each 
non-terminal
node of $G$ (we let 
$G^\circ$
denote the 
non-terminal
nodes) a player in $P$
that describes which player moves at each of these nodes; if $M(v) =
i$, then all the edges going out of node $v$ must correspond to
actions that player $i$ can play.
\item
$U = (u_1, \ldots, u_n)$ 
is an $n$-tuple of 
utility
functions from the leaves of $G$ to $\mathbb{R}$.
\item
$R = (\sim_1, \ldots, \sim_n)$ is an $n$-tuple of equivalence
relations on the nodes of  
$G^\circ$, one for each player.
Intuitively, if $v \sim_i v'$, then nodes $v$
and $v'$ are indistinguishable to player $i$.
Clearly, if $v
\sim_i v'$ for some $i \in [n]$, 
then the set of possible actions that can be performed at $v$ and $v'$
must be identical and $M(v) = M(v')$
(for otherwise, $i$ would have a basis for
distinguishing $v$ and $v'$).  
\end{itemize}
The equivalence relation $\sim_i$ induces a partition of $G^\circ$ into
equivalence classes called \emph{information sets} (of player $i$).
It is more standard in game theory to define $\sim_i$ as an
equivalence relation only on the nodes where $i$ moves.  We define it
on all 
non-terminal
nodes because if $K$ is a subset of players, we are
then able to define $\sim_K$ as the the intersection of $\sim_i$ for 
$i \in K$.  
Intuitively, $v \sim_K v'$ if the agents in 
$K$ cannot distinguish $v$ from $v'$, 
even if they pool all their information together.
We will need the equivalence relation $\sim_K$ for some of our
definitions below.

A \emph{(pure) strategy}
$s_i$
for player $i$ in an extensive-form game
$\Gamma$ is a function
that maps each 
node $v$ where $i$ moves
to an action in 
$v$,
with the constraint that if two nodes are indistinguishable to $i$, then $s_i$ must choose the same action on both. More precisely, if 
$v \sim_i v'$ then
$s_i(v) = s_i(v')$.
Each pure strategy profile $\vec{s}$
induces a unique path from the root of the tree to some
terminal node
$\ell \in G \setminus G^\circ$ where, given node 
$v$
on the path, the
next node in the path is 
defined by $s_{M(v)}(v)$.
The payoff of player $i$
given strategy $\vec{\sigma}$ is given by $u_i(\ell)$, although we
also write $u_i(\vec{\sigma})$ for simplicity. 
A
\emph{behavioral strategy} for player $i$ 
allows randomization. More precisely, a behavioral strategy
is a function $s_i$ that maps each node 
$v$
such that 
$M(v) = i$
to a  distribution over
$i$'s possible actions at 
$v$ (again, with the requirement that $s_i(v) = s_i(v')$ if $v \sim_i v'$).
We denote by $u_i(\vec{\sigma})$ the
expected payoff of player $i$  if players play (behavioral) strategy 
profile $\vec{\sigma}$. 

Nash equilibrium is defined for extensive-form games just as it is in
strategic-form games.
Given a NE $\vec{\sigma}$, the \emph{equilibrium path} consists of
the nodes of $G$ that can be reached with positive probability
using $\vec{\sigma}$. 
It is well known that in extensive-form games, Nash equilibrium does
not always 
describe what intuitively would be a reasonable play.
For instance,
consider the following game for two players in which $(p_1,p_2)$
describes the payoff of players 1 and 2 respectively: 
\begin{center}
    \includegraphics[]{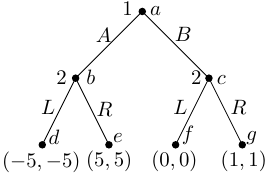}
\end{center}
In this case, the strategy profile in which
player 1 plays $B$ in $a$ and player 2 plays $L$ in $b$ and $R$ in $c$
is a Nash equilibrium. 
However, the reason that it is better for player 1 to play
$B$ is that if it plays $A$ then player 2 would play $L$ and they
would both get a utility of $-5$. This means that player 1 is being
influenced by an irrational threat, since if 1 plays $A$ it is in
player 2's best interest to play $R$ instead of $L$. In order 
to avoid these situations, we can extend the notion of Nash
equilibrium to \emph{subgame-perfect Nash equilibrium}, where
$\vec{\sigma}$ is a subgame-perfect Nash equilibrium in game $\Gamma$
if $\vec{\sigma}$ is a Nash equilibrium in all of subgames of $\Gamma$. In
the example above, the only subgame-perfect equilibrium is the one
given by $\sigma_1(a) = A$, $\sigma_2(b) = R$ and $\sigma_2(c) = R$.

Unfortunately, subgame-perfect equilibrium may not be well-defined
if  players have nontrivial information sets.
Consider the following game, where $b$ and $b'$ are in the same
information set for player 2, as are $c$ and $c'$.
\begin{center}
    \includegraphics[]{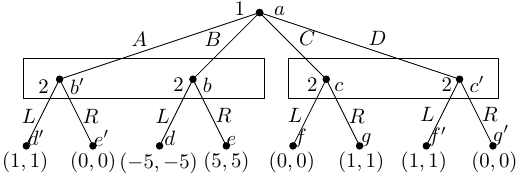}
\end{center}
In $b'$, player 2 would play $L$; 
in $b$, it would play $R$.
For player 2 to decide its
move, it must have a belief about the probability of being at each
node of the information set, and this belief must be consistent with
the strategy $\vec{\sigma}$ being played. For example, if
$\sigma_1(a) = B$
if player 2 is in information set $\{b, b'\}$,
then 2 would be sure that the true node is $b'$. 
We capture the players' beliefs by using a \emph{belief system} $b$, 
which is
a function from information sets $I$ to a probability
distribution over nodes in $I$. Intuitively, if $I$ is an information
set for player $i$, then $b$ represents
how likely it is for $i$ to be in each of the nodes of $I$.
Note that if $\vec{\sigma}$
is \emph{completely mixed}, which means that all actions are taken with positive probability, then
$b$ is uniquely determined by Bayes' rule.
More generally, we 
 say that a belief 
 system $b$ is \emph{consistent}
 with a strategy $\vec{\sigma}$ if there exists a
 sequence  $\vec{\sigma}^1, \vec{\sigma}^2,  \ldots$ of
 completely-mixed strategies that converges to $\vec{\sigma}$ such that the
 beliefs 
induced by Bayes' rule converge to $b$.
Given an extensive-form game $\Gamma = (P,G, M, U, R)$, a strategy
profile $\vec{\sigma}$, a belief system $b$, and an information set
$I$ for player $i$, let $u_i(\vec{\sigma}, I, b)$  denote $i$'s
expected utility conditional on being in information set $I$ and having
belief system $b$, if $\vec{\sigma}$ is played.

\begin{definition}[\cite{KW82}]\label{def:seq-eq}
A pair $(\vec{\sigma},b)$ consisting of a strategy profile
$\vec{\sigma}$ and a belief system $b$ consistent with $\vec{\sigma}$
is a \emph{sequential equilibrium}
if, for all players
$i$, all information sets $I$ of player $i$, and all strategies
$\tau_i$ for player $i$, $$u_i(\vec{\sigma}, I, b) \ge u_i((\tau_i,
\vec{\sigma}_{-i}), I, b).$$
\end{definition}

Even though it is standard to define a sequential equilibrium as a
pair $(\vec{\sigma}, b)$ consisting of strategy profile and a belief
system, for convenience we also say that a strategy profile
$\vec{\sigma}$ is a sequential equilibrium if there exists a belief
system $b$ such that $(\vec{\sigma}, b)$ is a sequential
equilibrium. 

Abraham, Dolev, and Halpern~\nciteyear{ADH13} (ADH) generalized
sequential equilibrium to allow for deviations by up to $k$ players.
To make this precise, we first need to define a generalization of a
belief system.

\begin{definition}
  Let $\Gamma = (P,G,M,U,R)$ be an extensive-form game and $K \subseteq
P$ be a subset of players. We define the equivalence relation $\sim_K$
by $v \sim_K v'$ iff $v \sim_i v'$ for all $i \in K$. The equivalence
classes of $\sim_K$ are called \emph{$K$-information sets}. 
\end{definition}

\begin{definition}
A \emph{k-belief system} $b$ in a game $\Gamma$ is a function that
maps each $K$-information set $I$ such that $K \subseteq P$ and $|K|
\le k$ to a distribution over the nodes in $I$.
\end{definition}
We can define what it 
means
for a $k$-belief system to be consistent
with $\vec{\sigma}$ just as we did in the case of belief systems.
With these definitions in hand, we can find the desired generalization
of sequential equilibrium.

\begin{definition}\label{def:k-sequential}
  A pair $(\vec{\sigma},b)$ consisting of a strategy profile
$\vec{\sigma}$ and a $k$-belief system $b$ consistent with $\vec{\sigma}$
is a 
\emph{$k$-resilient sequential equilibrium}
if, for all 
$K \subseteq P$ with $|K| \le k$,
all $K$-information sets $I$, 
and all strategies
$\tau_K$ for players in $K$ 
(w.r.t. $\sim_K$),
$$u_i(\vec{\sigma}, I, b) \ge u_i((\tau_K, \vec{\sigma}_{-K}), I, b)$$
for all $i \in K$.
\end{definition}

Note that the notions of $1$-resilient correlated equilibrium and $1$-resilient
sequential equilibrium are equivalent to the standard
notions of correlated equilibrium and sequential equilibrium,
respectively. 

\subsubsection{Bayesian Games}

\commentout{
\emph{Bayesian games} extend normal-form games by
assuming that each player $i \in P$ has a \emph{type} 
$t_i \in T_i$.  A player's type can be thought of as private
information that the player has, such as whether he is lazy or
industrious.  In our applications, an agent's type will be its input.
Types are assumed to 
be
sampled from a
distribution $q \in \Delta (T)$, where $T = T_1 \times \cdots \times T_n$.
The utility $u_i$ of a player $i$ is a function of not only the action
profile played, but also of of the type profile $(t_1, \ldots, t_n)$.
Formally, a Bayesian game is a tuple $(P, T, q, A, U)$, where,
as in  normal-form games, 
$P$, $A$, and $U$ 
are  the set of players, their 
actions, 
and their utility
functions, respectively;
$T$ is the set  
of possible type profiles, and $q$ is a distribution in $\Delta(T)$.
}

In all of the previous definitions, the utility of each player is
assumed to be common knowledge (\emph{perfect information}). However,
this is not always the case. Bayesian games capture this idea by
assuming that each player $i$ has a type 
$t_i \in T_i$ sampled from a
distribution $q \in \Delta (T)$, where $T = (T_1, \ldots, T_n)$, and
that the utility $u_i$ of $i$ is not only a function of the action
profile being played, but also of its type $t_i$.
Formally, a \emph{Bayesian games} is a tuple $(P, T, q, A, U)$, where,
as in  normal-form games, 
$P$, $A$, and $U$ 
are  the set of players, their 
actions, 
and their utility
functions, respectively;
$T$ is the set  
of possible type profiles, and $q$ is a distribution in $\Delta(T)$.

A \emph{strategy} in a Bayesian game for player $i$ is a map $\mu_i :
T_i \rightarrow \Delta(A_i)$. Intuitively, 
a strategy in a Bayesian game tells player $i$ how to choose its
action given its type. A \emph{correlated strategy profile} is a map $\mu : T
\rightarrow \Delta(A)$. Note that all distributions over strategy
profiles can be viewed as correlated 
strategy profiles, but the converse is not true. For instance, in a game
for two players in which $A_i = T_i = \{0,1\}$ for all $i$, a
correlated strategy profile may consist of both players playing action $t_1 +
t_2 \bmod 2$; this cannot be represented by a strategy profile, since
players must independently 
choose
their action given their type.  

Since the distribution $q$ is common knowledge, given a (possibly
correlated) strategy profile $\vec{\mu}$ in $\Gamma$, the expected
utility of a member $i$ of a coalition $K$
is
$$u_i^K(\vec{\mu}) = \sum_{\vec{t}_K \in T_K} q(\vec{t}_K) \sum_{\vec{t}}
q(\vec{t}\mid \vec{t}_K ) u_i(\vec{\mu}(\vec{t})),$$ 
where 
$u_i(\vec{\mu}(\vec{t}))$ denotes the expected utility of player
$i$ when an action profile 
is chosen
according to 
$(\mu(\vec{t}))$.
Intuitively, we are assuming that players in $K$ can share their
types, which is why we condition on $\vec{t}_K$.  
This allows us to define $k$-resilient Nash equilibrium in Bayesian
 games: 

\begin{definition}
In a Bayesian Game $\Gamma = (P,T, q, A, U)$, a strategy profile
$\vec{\mu} := (\mu_1, \ldots, \mu_n)$ 
is a $k$-resilient Nash equilibrium if, 
for
all coalitions $K$ with $|K| \le k$,
all maps $\mu'_K : T_K \rightarrow
\Delta(A_K)$, 
and all $i \in K$, 
$$u_i^K(\vec{\mu}) \ge u_i^K(\vec{\mu}_{-K}, \mu'_K)$$
\end{definition}

We can generalize correlated equilibrium to
Bayesian games as
follows. We can view 
a correlated equilibrium as a distribution $p \in \Delta(A)$ such that
if a trusted mediator samples an action profile $\vec{a} \in p$ and
sends action $a_i$ to each player $i$, it is always better for $i$ to
play $a_i$ rather than something else. In a Bayesian game,
the mediator instead samples the action
profile from a distribution that depends on the type profile.
More precisely, suppose that players send their types to a
trusted mediator, the mediator samples an action profile
$\vec{a}$ from a distribution 
$\mu(\vec{t})$
that depends on the type
profile $\vec{t}$ received,
and then sends action $a_i$ to each player $i$. We
say that 
$\mu$ 
is a
\emph{communication equilibrium} if it is optimal 
for the players to (a) tell their true type to the mediator, and (b)
play the action sent by the mediator.
The following definitions make this precise.

\begin{definition}
[\cite{F86,Myerson86}]
  A 
  correlated strategy profile
  $\mu : T \rightarrow \Delta(A)$ is a \emph{communication equilibrium} of
    Bayesian Game $\Gamma = (P,T, q, A, U)$ if, for all $i \in P$, all
    $t_i \in T_i$, all
    $\psi: T_i \rightarrow T_i$
    and all $\varphi : A_i \rightarrow
        A_i$, we have that $$\sum_{\vec{t}_{-i} \in T_{-i}} \sum_{\vec{a} \in A} q(t_{-i},
t_i)\mu(\vec{a} \mid \vec{t}_{-i}, t_i)u_i(\vec{t}_{-i}, t_i,\vec{a})
\ge $$ $$\sum_{\vec{t}_{-i} \in T_{-i}} \sum_{\vec{a} \in A} q(t_{-i},
t_i)\mu(\vec{a} \mid \vec{t}_{-i}, \psi(t_i))u_i(\vec{t}_{-i},
t_i,\vec{a}_{-i}, \varphi(a_i)).$$ 
\end{definition}


\begin{definition}\label{def:k-comm}
    A
     correlated strategy profile
    $\mu : T \rightarrow \Delta(A)$ is a \emph{$k$-resilient communication
equilibrium} of Bayesian Game $\Gamma = (P,T, q, A, U)$ if, for all $K
\subseteq P$, 
   all $\vec{t}_K \in T_K$, 
all $i \in K$, all $\psi: T_K \rightarrow T_K$, and all
$\varphi : A_K \rightarrow A_K$, we have that $$\sum_{\vec{t}_{-K} \in T_{-K}}
\sum_{\vec{a} \in A} q(t_{-K}, \vec{t}_K)\mu(\vec{a} \mid \vec{t}_{-K},
\vec{t}_K)u_i(\vec{t}_{-K}, \vec{t}_K,\vec{a}) \ge $$ $$\sum_{\vec{t}_{-K}
  \in T_{-K}} \sum_{\vec{a} \in A} q(t_{-K}, \vec{t}_K)\mu(\vec{a} \mid
\vec{t}_{-K}, \psi(\vec{t}_K))u_i(\vec{t}_{-K},
\vec{t}_K,\vec{a}_{-K}, \varphi(\vec{a}_K)).$$ 
\end{definition}

In Definition~\ref{def:k-comm}, the functions $\psi$ and
$\varphi$ represent possible deviations for players in $K$; $\psi$
describes how they might lie to the mediator about their type (by sending type
profile
$\varphi(\vec{t}_K)$ 
instead of their actual type profile
$\vec{t}_K$), while $\varphi$ represents how the might deviate in what
they do 
(playing strategy profile $\varphi(\vec{a}_K)$ instead of the
strategy profile $\vec{a}_K$ suggested by the
mediator). Definition~\ref{def:k-comm} says that $\mu$ is a
$k$-resilient communication equilibrium if none of these deviations is
profitable for the players in $K$. 

We can combine the notions of extensive-form game and Bayesian game in
the obvious way to get \emph{extensive-form Bayesian games}: we start
with an extensive form game, add a type space $T$ and a commonly known
distribution $q$ on $T$, and then 
have the utility function  depend on the type profile as well the 
terminal node
reached. We leave formal details to the reader.

\commentout{
\textbf{Actions vs. Strategies}

In the literature, sometimes the elements of the sets $S_i$ in normal-form and Bayesian games are referred to as \emph{strategies} instead of \emph{actions}. We reserve the term \emph{strategy} for a more generic notion of player behavior in the different types of games presented so far. For instance, a pure strategy in a normal-form game is essentially equivalent to an action, a strategy in an extensive-form game  
}
}

\subsection{Cheap-talk games}

Given a normal-form game $\Gamma = (n, A, U)$ 
or a Bayesian game $\Gamma = (P, T, q, A, U)$, 
consider an extension
$\Gamma_{\ACT}$ in which players can communicate
among themselves before
playing an action in $\Gamma$. We call this extension a
\emph{cheap-talk game}. More precisely, in a cheap-talk game
$\Gamma_{\ACT}$, players are able to exchange messages 
and play an action in
$\Gamma$ (although they may play an action in $\Gamma$ at most once).
The payoff of each player is defined by the utilities $U$ of $\Gamma$, given the action 
profile played in $\Gamma$.
In this paper we always assume that messages are \emph{authenticated}, so that each recipient knows who sent the message, and
messages are never corrupted or modified.

We can view cheap-talk games as possibly infinite extensive-form games
in which the nodes
are the
\emph{global histories}
$\vec{h} = (h_1, \ldots,
h_n)$ of the 
players, and the 
actions associated with an edge involve either performing an 
internal computation
(e.g., tossing a random coin or updating the local state with the
messages received), 
sending messages, or playing an
action in the underlying game.  A global history $\vec{h}$ is just a tuple
of local histories, where the local history $h_i$ of player $i$
contains 
all internal computations,
all messages sent by $i$ and received by $i$, along with their time stamps
(in the synchronous case) or the
number of
times that $i$ was scheduled
since the beginning of the game
(in the asynchronous case). Two 
global histories $\vec{h} = (h_1, \ldots, h_n)$ and $\vec{h}' = (h_1', \ldots, h_n')$
are
indistinguishable by player $i$ if 
$h_i = h'_i$.
Analogously, $\vec{h}$ and $\vec{h}'$ are indistinguishable by a coalition $K$ of
players if $\vec{h}_K = \vec{h}'_K$.

A strategy profile $\vec{\sigma}$ in a cheap-talk game
$\Gamma_{\ACT}$ \emph{induces} a strategy profile $\vec{\tau}$ in $\Gamma$ 
if the outcome that results from
playing $\vec{\sigma}$  
in 
$\Gamma_{\ACT}$ 
is
identical
to that that results from playing $\vec{\tau}$
in $\Gamma$. If $\Gamma$ is a normal-form game, the outcome that
results from playing 
$\vec{\sigma}$ is the 
distribution over 
action profiles in $\Gamma$ that results
when $\vec{\sigma}$ is played in $\Gamma_{\ACT}$. If $\Gamma$ is 
a Bayesian game, the outcome is a map from type profiles to
distributions over action profiles 
(i.e., a correlated strategy profile).
%
A strategy profile  $\vec{\sigma}$ \emph{implements} a strategy
profile $\vec{\sigma}'$ 
if $\vec{\sigma}$ and $\vec{\sigma}'$
induce the same strategy profiles in $\Gamma$.

\commentout{
A strategy profile $\vec{\sigma}$ in a cheap-talk game
$\Gamma_{\ACT}$ \emph{implements} a
strategy $\vec{\tau}$ in $\Gamma$ if the outcome that results from
playing $\vec{\sigma}$  
in 
$\Gamma_{\ACT}$ 
is
identical
to that that results from playing $\vec{\tau}$
in $\Gamma$.
%
If $\Gamma$ is a normal-form game, the outcome that results from playing
$\vec{\sigma}$ is the 
distribution over 
action profiles in $\Gamma$ that results
when $\vec{\sigma}$ is played in $\Gamma_{\ACT}$. If $\Gamma$ is 
a Bayesian game, the outcome is a map from type profiles to
distributions over action profiles 
(i.e. a correlated strategy profile).
}
Unless specified otherwise, we assume that all communication in
cheap-talk games is \emph{synchronous}, which means that it proceeds
in rounds and that all messages sent in a given round by each player are
guaranteed to be received by their recipient at the beginning of the
next round.
\fullv{
The analysis of \emph{asynchronous} communication is done
in section~\ref{sec:sync-async}. 
}

These definitions are illustrated by the following example.

\begin{example}\label{ex:comm-games}
  Consider a Bayesian game $\Gamma$ with two players. Each player $i \in
  \{1,2\}$ has a type $t_i \in \{0,1\}$ and can play actions in
$\{0,1\}$. If the action $a_i$ played by player $i$ satisfies $a_i =
  t_1t_2$, then that player gets a utility of 1; otherwise, she gets
  0. Consider a strategy profile $\vec{\sigma}$ for the cheap-talk
extension $\Gamma_{\ACT}$ of $\Gamma$ in which each player sends its type to the
other player in the first round, then the players each locally compute
$t_1t_2$ and play it in the underlying game. It is easy to check that
$\vec{\sigma}$ is a Bayesian Nash equilibrium (both players get the
maximum possible payoff), and that it implements the correlated
strategy profile 
that maps a given type profile $(t_1, t_2)$ to the distribution that assigns probability $1$ to $(t_1t_2, t_1t_2)$ and $0$ to everything else.
\end{example}

\commentout{
\section{Synchrony vs. asynchrony}\label{sec:sync-async}
In the game-theory literature, communication between the players is
assumed to proceed in synchronous rounds.  That is, at each round $t$,
all the players send one message to each other player (we identify not
sending a message with sending a special message $\bot$)
and these messages are received by
their intended recipients before the beginning of round $t+1$.
The assumption that we identify not sending a message with sending
$\bot$ is made without loss of generality---player $j$ can tell if
player $i$ has not sent her a message.
Moreover, it is typically assumed that it is common knowledge when the
communication phase ends and that, after it ends, all the players
simultaneously move in the underlying game.  We  call this the
\emph{synchronous} setting.

Here we also consider an \emph{asynchronous} setting that is quite
standard in the distributed computing literature.  In the asynchronous
setting, 
there is no global notion of time, and messages may take arbitrary
amounts of time to be delivered (although we do assume that all
messages are eventually delivered).
Thus, we can no longer identify not sending a message with sending
$\bot$; if $j$ has not received a message from player $i$, $j$ is not
sure if this is because $i$ did not send $j$ a message or if $i$ sent
a message that has not yet been delivered.
For ease of exposition, we assume
that message delivery is under the control of a
\emph{scheduler}, who 
also decides when
each player gets to move, with the guarantee that all players
eventually get to move if they want to move. In both the synchronous
and the asynchronous setting, we assume that the messages are
\emph{authenticated}, so that each recipient knows who sent the message, and
messages are never corrupted or modified.
If $\Gamma_{\ACT}$ is asynchronous, the payoff of the players may
depend on the strategy $\sigma_e$ of the scheduler.
This means that the definition of implementation and
of the solution concepts must take the scheduler into account.
We extend the definitions of Nash equilibrium, correlated equilibrium,
sequential equilibrium, and communication equilibrium by requiring that the
relevant inequality 
holds for all choices of $\sigma_e$. For example
a strategy profile $\vec{\sigma}$ is a Nash
equilibrium in an asynchronous setting if, for all $i \in P$, all
strategies $\tau_i$ for 
player $i$, and all schedulers $\sigma_e$, $u_i(\vec{\sigma},
\sigma_e) \ge u_i(\tau_i, \vec{\sigma}_{-i}, \sigma_e)$.
Since
the action profile played might depend on the
scheduler, 
a strategy $\vec{\sigma}$ in $\Gamma_{\ACT}$ might induce
more than one strategy in $\Gamma$ (see Example~\ref{ex:race-game}
below). 
A strategy profile $\vec{\sigma}$ for
$\Gamma_{\ACT}$ \emph{implements a set $S$ of strategy profiles in
$\Gamma$} in an asynchronous setting 
if (a)
for every scheduler $\sigma_e$ in $\Gamma_{\ACT}$,
the outcome obtained when playing $\vec{\sigma}$ with scheduler
$\sigma_e$ is the same as that obtained when playing some strategy profile
in $S$, 
and (b) for every strategy profile $\vec{\tau}$ in $S$, there exists a
scheduler $\sigma_e$ such that
the outcome obtained when playing $\vec{\sigma}$ with $\sigma_e$
is the same as that obtained when playing $\vec{\tau}$. Intuitively, 
$\vec{\sigma}$ implements $S$ if the set of
outcomes that result from playing $\vec{\sigma}$ with different
choices of scheduler are precisely those that result from playing
strategy profiles in
$S$. As in the synchronous setting, if $\Gamma$ is a
normal-form game, we take the strategy profiles in $S$ to consist of
distributions over action profiles in $\Gamma$, while if $\Gamma$ is 
a Bayesian game, the strategy profiles in $S$ are 
correlated strategy profiles.

In cheap-talk games in synchronous systems, we assumed (as is
standard in the literature) that the game tree is 
finite and at the last move on a path, the players play an action in the
underlying game; we take the utility of a 
terminal node
to be the
utility of the corresponding action profile in the normal-form game.
We cannot assume
this in asynchronous systems. Indeed, in asynchronous systems, to
allow for players who deviate from an equilibrium, we must consider
game trees with infinite paths. For example, a player $i$ may wait forever
for a message that never arrives, because the sender deviated and
never sent it, so $i$ may never play an action in the underlying game.   We
must thus define what the payoffs are on infinite paths and at the 
terminal node
of
a finite path where some players do not play an action in the
underlying game.    This amounts to treating players who have not
played an action as, in fact, having played some action.
The two main approaches for doing this are the \emph{default-move
approach}
\cite{ADGH19}, 
where if player $i$ never plays an action in the
underlying game,  its action is sampled from a default distribution, and
the \emph{Aumann and Hart} (AH) approach \cite{AH03}, where the
action played by $i$ is some function of $i$'s local history.
The AH approach essentially 
assumes
that players 
can leave a ``will'' defining
the action in the underlying game that they should play if they never
actually play an action in the underlying game while playing the
cheap-talk game.  
In the asynchronous setting
(with both  the AH and the default move approach),
we say that a strategy profile $\vec{\sigma}$ 
in $\Gamma_{\ACT}$
\emph{implements} a strategy profile
$\vec{\tau}$ 
in $\Gamma$
if
the set of strategy profiles implemented by $\vec{\sigma}$ (for
different choices of scheduler)
$\{\vec{\tau}\}$ (thus, playing $\vec{\sigma}$ results in the 
same
outcome no matter what the scheduler does).

The following example illustrates some of the new subtleties that asynchrony introduces.

\begin{example}\label{ex:race-game}
Let $\Gamma$ be a normal-form game for $n$ players where the set
of actions is $\{1,2,\ldots, n\}$, and let $\Gamma_{d,asyn}$ be a cheap-talk
extension of $\Gamma$ in the asynchronous setting, where there is a trusted
mediator $d$ and the AH approach is used.   Suppose that
$\vec{\sigma}$ is the strategy profile that proceeds as follows:
each player sends an empty message to the mediator when
they are first scheduled. The mediator waits until
it receives  a message, and then sends a message to all players
with the index of the player whose message was received
first. When a player $j$ receives a message from the mediator with a
number $i \in \{1,2,\ldots, n\}$, player $j$ plays action $i$. If 
player $j$ never receives a message from the mediator, then according
to its will, $j$ plays action $j$.
The player whose index appears the most often in the resulting action
profile receives a payoff of 1 (if there are ties, they all get 1),
while the remaining players get 0. 

This game can be viewed as a
race. The player whose message gets to the mediator first receives a
payoff of 1 while the rest receive 0. Note that, regardless of the
scheduler, all players play the same action in $\Gamma$.
However, the scheduler decides which action is played. Thus, the
strategy $\vec{\sigma}$ in the example implements the set of
strategies $\{(1,1 \ldots, 1), (2,2,\ldots, 2), \ldots, (n, n, \ldots,
n)\}$. This strategy
is not a Nash equilibrium. To see this,
consider the scheduler $\sigma_e$ that schedules the players
sequentially (first player 1, then player 2, and so on).  
If player $i$ sends two messages to the
mediator and is the only player to do so, then the scheduler delivers
$i$'s first message before it delivers 
any other player's message. If more than one player sends two messages to the
mediator, the schedulers chooses one of these players at
random and delivers her message first.  The remaining messages 
are delivered in some random order.
With this scheduler, the players benefit by deviating and
sending two messages to the mediator instead of just
one. Thus, $\vec{\sigma}$ is not a Nash equilibrium. A similar
argument can be used to show that, for all games $\Gamma$, if
$\vec{\sigma}$ is a $k$-resilient Nash equilibrium in $\Gamma_{\ACT}$,
then the payoffs of the players when playing $\vec{\sigma}$ cannot
depend on the scheduler (see~\cite{ADGH19}). 
\end{example}

This example shows that in asynchronous systems, the set of possible
deviations is much larger than in synchronous
systems. However, there are cases where controlling the
scheduling of the messages and players does not give that much power
to an adversary. For instance, if we consider an asynchronous
cheap-talk extension of Example~\ref{ex:comm-games}, the strategy
profile in which both agents send their type to each other the first
time they are scheduled and play action $t_1t_2$ is a Nash
equilibrium (note that this is almost equivalent to the strategy
proposed in Example~\ref{ex:comm-games}, however there is no notion of
``rounds'' in an asynchronous setting). Many of the
results that hold in the synchronous case also hold in the
asynchronous case, but may require a larger proportion of the players
to be honest. For instance, asynchronous multiparty secure
computation tolerates deviations by at most a quarter of the agents,
while synchronous multiparty secure computation tolerates deviations
by at most a third \cite{BGW88,BCG93}. Our results hold for similar
bounds.  (An
intuitive explanation of why there is a difference between the
thresholds for the synchronous and asynchronous settings can be
found in Appendix~\ref{sec:tools2}.)
}

\section{The main result}\label{sec:results}
\commentout{
Given a normal-form or Bayesian game $\Gamma$, let $\Gamma_{\ACT, syn}$
be the synchronous cheap-talk extension of $\Gamma$ and let
$\Gamma_{\ACT, asyn}$ be its asynchronous extension. Let $\Gamma_{d, syn}$
and $\Gamma_{d, asyn}$ be the synchronous and
asynchronous cheap-talk extensions of $\Gamma$ in which players can
also communicate with a trusted mediator $d$. 
If
$\Gamma_{\ACT}$ is a 
synchronous
cheap-talk game,
let $SE_k(\Gamma_{\ACT, syn})$ denote the set of possible strategies 
in $\Gamma$ induced by $k$-resilient sequential equilibria in
$\Gamma_{\ACT, syn}$ (note that if $\Gamma$ is a normal-form game,
$SE_k(\Gamma_{\ACT, syn})$ is a set of distributions over action profiles,
while if $\Gamma$ is a Bayesian game, $SE_k(\Gamma_{\ACT, syn})$ is a set of
correlated strategy profiles).
If $\Gamma_{\ACT, asyn}$ is asynchronous, we define
$SE_k(\Gamma_{\ACT, asyn})$ analogously, except that in this case,
since the outcome may depend on the 
scheduler
(see the discussion in
Section~\ref{sec:sync-async}), if $\Gamma$ is a normal-form game, then
$SE_k(\Gamma_{\ACT, asyn})$ is a set of sets of strategy profiles, and if
$\Gamma$ is a Bayesian game, then $SE_k(\Gamma_{\ACT, asyn})$ is a set
of sets of correlated strategy profiles. 
\commentout{
We also denote by $CE_k^S(\Gamma)$ the set of possible distributions
over action profiles induced by a strongly $k$-resilient correlated 
equilibrium in $\Gamma$. $Com_k^S(\Gamma)$ and $SE_k^S(\Gamma)$ are
defined analogously. 
}
}

Given a normal-form or Bayesian game $\Gamma$, let $\Gamma_{\ACT}$
be its cheap-talk extension and $\Gamma_{d}$ be its cheap-talk
extension where players can 
also communicate with a trusted mediator $d$. 
Let $SE_k(\Gamma_{\ACT})$ (resp., $SE_k(\Gamma_d)$) denote the set of
possible strategies  
in $\Gamma$ induced by $k$-resilient sequential equilibria in 
$\Gamma_{\ACT}$ (resp., $\Gamma_d)$.  Note that if $\Gamma$ is a normal-form game,
$SE_k(\Gamma_{\ACT})$ is a set of distributions over action profiles,
while if $\Gamma$ is a Bayesian game, $SE_k(\Gamma_{\ACT})$ is a set of
correlated strategy profiles. 
We define $SE_k^S$ equivalently but for strong $k$-resilient sequential equilibria instead.

In the rest of this paper,  we consider only equilibria in
which each action profile is sampled with rational probability. The
reason for this restriction is that it is a standard assumption in distributed
systems that messages have finite length and that players cannot
perform arbitrarily large operations. This means that agents cannot
encode real numbers into messages and they cannot perform operations
on real numbers, in general. If we relax these constraints and allow
players to send messages 
and operate with
arbitrary
real numbers, our results can be
extended to all equilibria inside the convex hull of rational
equilibria, using techniques due to Gerardi
\nciteyear{Gerardi04}.
These techniques are described in Section~\ref{sec:extension}.

\commentout{
Our two main results show that we can implement 
a  (strongly) $k$-resilient equilibrium with a mediator using cheap talk
if $n > 3k$ in synchronous systems
and if $n > 4k$ in asynchronous systems.
}
Our main result shows that we can implement 
a  (strongly) $k$-resilient equilibrium with a mediator and $n$ players using cheap talk
if $n > 3k$.

\begin{theorem}\label{thm:main}
  If $\Gamma = (P, T, q, A,U)$ is a Bayesian game for $n$ players and
  $n > 3k$, then
  $SE_k(\Gamma_{\ACT}) = SE_k(\Gamma_{d})$ 
    and $SE_k^S(\Gamma_{\ACT}) = SE_k^S(\Gamma_{d})$.
\end{theorem}

\commentout{
\begin{theorem}\label{thm:main2}
  If $\Gamma = (P,T, q, A,U)$ is a Bayesian game for $n$ players and
  $n > 4k$, then
  $SE_k(\Gamma_{\ACT, asyn}) = SE_k(\Gamma_{d, asyn})$
and 
$SE_k^S(\Gamma_{\ACT, asyn}) = SE_k^S(\Gamma_{d, asyn})$
with both the default move and AH approaches. 
\end{theorem}
}

\commentout{
As shown in Section~\ref{sec:descriptions}, the sets $SE_k(\Gamma_{d,
syn})$, $SE_k^S(\Gamma_{d, syn})$, $SE_k(\Gamma_{d, asyn})$, and (in
some cases) $SE_k^S(\Gamma_{d, asyn})$ have a relatively simple
characterization. For instance, it can be shown that if $\Gamma$ is a
normal-form game, then $SE_k(\Gamma_{d, syn})$ is the set of
$k$-resilient correlated equilibria of $\Gamma$; and if $\Gamma$ is a
Bayesian game, then $SE_k(\Gamma_{d, syn})$ is the set of
$k$-resilient communication equilibria of $\Gamma$. (Gerardi
\nciteyear{Gerardi04} proved these results for $k = 1$.)
}
As shown in Section~\ref{sec:descriptions}, the sets
$SE_k(\Gamma_{d})$ and $SE_k^S(\Gamma_{d})$ have a relatively simple 
characterization. For instance, it can be shown that if $\Gamma$ is a
normal-form game, then $SE_k(\Gamma_{d})$ is the set of
$k$-resilient correlated equilibria of $\Gamma$; and if $\Gamma$ is a
Bayesian game, then $SE_k(\Gamma_{d})$ is the set of
$k$-resilient communication equilibria of $\Gamma$. (Gerardi
\nciteyear{Gerardi04} proved these results for $k = 1$.)

\commentout{

\begin{proposition}\label{prop:desc1}
If $\Gamma = (P,A,U)$ is a normal-form game for $n$ players, then
$SE_k(\Gamma_{d, syn}) = CE_k(\Gamma)$
and 
$SE_k^S(\Gamma_{d, syn}) = CE_k^S(\Gamma)$
for all $k \le n$. 
\end{proposition}

\begin{proof}
  Clearly, the distribution over action profiles that results from
  playing a $k$-resilient 
  sequential equilibrium
(resp., strongly $k$-resilient sequential equilibrium)
$\vec{\sigma}$ in $\Gamma_{d, syn}$ 
is also a $k$-resilient
correlated equilibrium 
(resp., strongly $k$-resilient correlated equilibrium)
of $\Gamma$, for otherwise there exists a coalition
$K$ of at most $k$ players 
such that all (resp., some)
members of the coalition
can increase their 
utility
by deviating
from
$\vec{\sigma}$ by playing a different action in the underlying game.
For the opposite inclusion, observe that all $k$-resilient
(resp., strongly $k$-resilient) 
correlated
equilibria $p$ 
of 
$\Gamma$ can be easily implemented with a mediator
the mediator samples an action profile
$\vec{a}$ following distribution $p$ and gives $a_i$ to each player
$i$. Then each player $i$ plays whatever is sent by the mediator.
\commentout{
and
if it never receives an action during the communication phase, it
plays the best response assuming that all other players play their
part of $\vec{a}$ and $\vec{a}$ is sampled from $p$.
}
\end{proof}

Theorem~\ref{thm:main} and Proposition~\ref{prop:desc1} together
imply the following corollary: 

\begin{corollary}
If $\Gamma = (P,A,U)$ is a normal-form game for $n$ players and $n > 3k$, then
  $SE_k(\Gamma_{\ACT, syn}) = CE_k(\Gamma)$ 
and 
$SE_k^S(\Gamma_{\ACT, syn}) = CE_k^S(\Gamma)$. 
\end{corollary}

\commentout{
\begin{corollary}
If $\Gamma = (P,A,U)$ is a normal-form game for $n$ players, then
$SE_k(\Gamma_{\ACT, asyn}) = CE_k(\Gamma)$
and 
$SE_k^S(\Gamma_{\ACT, asyn}) = CE_k^S(\Gamma)$
for all $k$ such that $n >
4k$, with both the DM and AH approaches. 
\end{corollary}
}

In the asynchronous setting, the description of $SE_k(\Gamma_{d,
  asyn})$ is a bit more convoluted. Reasoning similar to that used in
the proof of
  Proposition~\ref{prop:desc1} shows that for a fixed scheduler, the
  a $k$-resilient (resp., strongly $k$-resilient) sequential
  equilibrium in $SE_k(\Gamma_{d, asyn})$ induces a $k$-resilient (resp.,
strongly $k$-resilient) correlated equilibrium in $\Gamma$. Thus,
$SE_k(\Gamma_{d, asyn}) \subseteq \mathcal{P}(CE_k(\Gamma))$ and
$SE_k^S(\Gamma_{d, asyn}) \subseteq \mathcal{P}(CE_k^S(\Gamma))$,
where $\mathcal{P(S)}$ denotes the power set of $S$ (i.e., the set of
all subsets of $S$). The next proposition gives a precise description
of $SE_k(\Gamma_{d, asyn})$ and $SE_k^S(\Gamma_{d, asyn})$.

\begin{proposition}\label{prop:desc-asyn}
Given a set $S$ of strategies, let $\mathcal{P}_=(S)$ the set of 
  nonempty subsets $S'$ of $S$ such that every element of $S'$ gives the
same expected utility to all players. If $\Gamma = (P,A,U)$ is a
normal-form game for $n$ players and $n \ge k$, then 
$SE_k(\Gamma_{d, asyn}) = \mathcal{P}_=(CE_k(\Gamma))$
and 
$SE_k^S(\Gamma_{d, asyn}) = \mathcal{P}_=(CE_k^S(\Gamma))$.
\end{proposition}
\begin{proof}
In games with a mediator, we write $\sigma_d$ to denote a generic
strategy for the mediator, and $\vec{\sigma} + \sigma_d$ to denote a
strategy profile for the players and the mediator.
To show that $SE_k(\Gamma_{d, asyn}) \subseteq
\mathcal{P}_=(CE_k(\Gamma))$ and $SE_k^S(\Gamma_{d, asyn}) \subseteq
\mathcal{P}_=(CE_k^S(\Gamma))$, 
suppose by way of contradiction that some strategy $\vec{\sigma} +
\sigma_d$ in $SE_k(\Gamma_{d, asyn})$ (resp., $SE_k^S(\Gamma_{d,
  asyn})$) induces a set $S$ of strategies such that, for some player $i$,
there exist two strategies $\vec{\tau}, \vec{\tau}' \in S$ induced
by schedulers $\sigma_e$ and $\sigma_e'$, respectively, such that
$u_i(\vec{\tau}) < 
u_i(\vec{\tau}')$.  Consider a scheduler $\sigma_e''$ that does the
following: it first schedules player $i$. If $i$ sends a message to
itself when it is first scheduled, then $\sigma_e$ acts like
$\sigma'_e$, and otherwise it acts like $\sigma''_e$. By construction,
when the scheduler is $\sigma''_e$, 
$i$ gains by deviating from $\sigma_i$
and sending a message to itself when it is first scheduled.  (This
construction will not work if $\sigma_i$ already requires $i$ to send
a message to itself; in this case, we construct $\sigma''_e$
so that the signal from $i$ to the
scheduler is encoded differently, for example, by $i$ sending two
messages to itself.)  This
shows that $SE_k(\Gamma_{d, asyn}) \subseteq
\mathcal{P}_=(CE_k(\Gamma))$ and $SE_k^S(\Gamma_{d, asyn}) \subseteq
\mathcal{P}_=(CE_k^S(\Gamma))$. 

To prove the opposite inclusions, 
since we consider only equilibria with rational
probabilities, the set $CE_k(\Gamma)$ (resp., $CE_k^S(\Gamma)$) is
countable, as is any subset $S \subseteq CE_k(\Gamma)$ (resp., any
subset $S \subseteq CE_k^S(\Gamma)$). Given $S = \{\vec{\tau}^1,
\vec{\tau}^2, \ldots \}$ such that for all $k$ and players $i$ and
$j$, $\vec{\tau}^k_i = \vec{\tau}^k_j$, consider
the following strategy $\vec{\sigma} + \sigma_d$ in $SE_k(\Gamma_{d,
  asyn})$ (resp., $SE_k^S(\Gamma_{d, asyn})$). Each player sends an
empty message to the mediator when it is first scheduled. Let $N$
be the number of times the mediator that has been scheduled before
receiving a message from some player. The mediator samples
an action profile $\vec{a}$ from $\vec{\tau}^N$ (or from $\vec{\tau}^1$
if $S$ is finite and has fewer than $N$ elements), and sends $a_i$ to
each player $i$. Players play whatever they receive from the
mediator. It is straightforward to check that for each strategy
$\vec{\tau}$ in $S$, there is a scheduler $\sigma_e$ that induces
$\vec{\tau}$ when the players and the mediator play $\vec{\sigma} +
\sigma_d$. To check that $\vec{\sigma} + \sigma_d$ is indeed a
$k$-resilient (resp., strongly $k$-resilient) correlated equilibrium,
note that all strategies in $S$ are $k$-resilient (resp., strongly
$k$-resilient) correlated equilibria of $\Gamma$ that give the same
utility to each player, and thus (a) no coalition of $k$ or less
players can benefit from playing an action profile  different from the
one suggested by the mediator, and (b) regardless of what the
coalition does, the expected 
payoff at the end of the protocol is constant. 
\end{proof}

Theorem~\ref{thm:main2} and Proposition~\ref{prop:desc-asyn} together imply:
\begin{corollary}
  If $\Gamma = (P,A,U)$ is a normal-form game for $n$ players and $n >
  4k$, then
  $SE_k(\Gamma_{\ACT, asyn}) = \mathcal{P}_=(CE_k(\Gamma))$
and 
$SE_k^S(\Gamma_{\ACT, asyn}) = \mathcal{P}_=(CE_k^S(\Gamma))$
with both the DM and AH approaches. 
\end{corollary}

For Bayesian games, we have the following result.

\begin{proposition}\label{prop:desc2}
  If $\Gamma = (P, T, q, A,U)$ is a Bayesian game for $n$ players and
  $n \ge k$, then
  $SE_k(\Gamma_{d, syn}) = Com_k(\Gamma)$ 
and 
$SE_k^S(\Gamma_{d, syn}) = Com_k^S(\Gamma)$ .
\end{proposition}

\begin{proof}
Suppose that $\vec{\sigma}$ is a $k$-resilient (resp., strongly
$k$-resilient) sequential equilibrium of $\Gamma_{d, syn}$. Then the
correlated strategy profile $\mu$ induced by $\vec{\sigma}$ in
$\Gamma$ must be a $k$-resilient (resp., strongly $k$-resilient)
communication equilibrium. If $\mu$ is not a $k$-resilient (resp.,
strongly $k$-resilient) communication equilibrium, then 
there exists a coalition $K$ of players with $|K| \le k$ and two
functions $\psi: T_K \rightarrow T_K$ and $\varphi : A_K \rightarrow
A_K$ such that the inequality of Definition~\ref{def:k-comm} does not
hold for some (resp., for all) $i \in K$. It follows that if
agents in $\Gamma_{d, syn}$ $K$ play $\vec{\sigma}$ as if they had
type profile $\psi(\vec{t}_K)$ instead of their true types, and then
play action $\varphi(\vec{a}_K)$ instead of the action profile
$\vec{a}_K$, the utility of all (resp., some) agents would strictly
increase, which contradicts the assumption that $\vec{\sigma}$ is a
$k$-resilient 
(resp., strongly $k$-resilient) sequential equilibrium of $\Gamma_{d,
  syn}$. For the opposite inclusion, recall from the discussion after
  Definition~\ref{def:k-comm} that we can identify a correlated
strategy profile $\mu$ in $\Gamma$ with a canonical strategy
$\vec{\sigma}$ in $\Gamma_{d, syn}$ in which players tell the mediator
their type and the mediator computes which action profile they should
play according to $\mu$. If $\mu$ is a $k$-resilient (resp., strongly
$k$-resilient) communication equilibrium, then the canonical strategy
$\vec{\sigma}$ is a $k$-resilient (resp., strongly $k$-resilient)
sequential equilibrium of $\Gamma_{d, syn}$ that induces $\mu$ (see
the discussion after Definition~\ref{def:k-comm} for details).
\end{proof}

Theorem~\ref{thm:main} and Proposition~\ref{prop:desc2} together imply:

\begin{corollary}
If $\Gamma = (P, T, q, A,U)$ be a Bayesian game for $n$ players and $n
> 3k$, then 
  $SE_k(\Gamma_{\ACT, syn}) = Com_k(\Gamma)$ 
and 
$SE_k^S(\Gamma_{\ACT, syn}) = Com_k^S(\Gamma)$ .
\end{corollary}

Unfortunately, we do not believe that there is a simple
description of $SE_k(\Gamma_{d, asyn})$.
\commentout{
Asynchrony seems to make things much more complicated.
Reasoning analogous
to that of the proof of Proposition~\ref{prop:desc1} shows that
$SE_k(\Gamma_{d, asyn}) \subseteq Com_k(\Gamma)$, 
but it is still an open problem if all $k$-resilient communication
equilibria of $\Gamma$ can be implemented with a mediator in
asynchronous systems.
}
To understand why, note that a key step in the proofs of
Theorems~\ref{thm:main} and \ref{them:main2} is for players
to send their types to the mediator.
In asynchronous systems, the
mediator cannot  
distinguish between players that don't send their type to the mediator
and players
that send it 
but the receipt of the message is delayed by the scheduler.
Nevertheless, if players can somehow punish those that never send
their type to the mediator, 
\commentout{
the same construction used in the synchronous case would suffice for the asynchronous
case as well. More precisely, 
with the AH approach, all $\mu \in
Com_k(\Gamma)$ can be implemented in $\Gamma_{d, asyn}$ if
$\Gamma$ has a $k$-\emph{punishment equilibrium} with respect to
$\mu$. 
}
we can guarantee that it is optimal for all players to send a type to the mediator, and thus
we can extend the constructions used in the proof of
Propositions~\ref{prop:desc-asyn} and \ref{prop:desc2} to
describe $SE_k(\Gamma_{d, asyn})$ and $SE_k^S(\Gamma_{d, asyn})$.  

\begin{definition}[\cite{Bp03,ADGH06}]
If $\Gamma = (P, T, q, S,U)$ is a Bayesian game for $n$ players and
$\mu \in Com_k(\Gamma)$ 
or $\mu \in Com_k^S(\Gamma)$,
then a  \emph{$k$-punishment equilibrium}
  with respect 
  to $\mu$ is a 
  $k$-resilient Bayesian Nash equilibrium $\mu'$ such that $u_i(\mu) >
  u_i(\mu')$ for all players $i$. 
  \commentout{
  strategy profile $\vec{\tau}$ such that
  \begin{itemize}
      \item [(a)] $\vec{\tau}$ is a $k$-resilient Nash equilibrium of $\Gamma$.
      \item [(b)] If at least $n-k$ players play
        their part of 
        $\vec{\tau}$, all players are guaranteed to get a lower expected payoff than the expected utility of $\mu$. 
  \end{itemize}
}
\end{definition}

Intuitively, a $k$-punishment equilibrium
with respect to $\mu$
is a strategy that is a 
$k$-resilient
equilibrium which, if played by 
at least $n-k$ 
players, results in all
players being worse off than they would be with $\mu$.
If a $k$-punishment equilibrium exists, then 
we have the following.

\begin{proposition}\label{prop:desc3}
\commentout{
  If $\Gamma = (P, T, q, S,U)$ is a Bayesian game,
 $\mu \in Com_k(\Gamma)$
 (resp., $\mu \in Com_k^S(\Gamma)$),
 and 
$\Gamma$ has a $k$-punishment equilibrium $\vec{\tau}$
with respect to $\mu$, then $\mu \in
SE_k(\Gamma_{d, asyn})$ 
(resp., $\mu \in SE_k^S(\Gamma_{d, asyn})$) 
with the AH approach or with the default
move approach if the default move is $\vec{\tau}$. 
}
If $\Gamma = (P, T, q, S,U)$ is a Bayesian game,
$S$ is a subset of $k$-resilient (resp., strongly $k$-resilient)
communication equilibria of $\Gamma$ such that, for all players $i$
and all types $t_i$ of $i$, the expected utility of $i$ given $t_i$ is
the same for all $\mu$ in $S$, and there exists a $k$-punishment
equilibrium  $\vec{\tau}$ with respect to some $\mu$ in $S$, then $S \in
SE_k(\Gamma_{d, asyn})$ (resp., $S \in SE_k^S(\Gamma_{d, asyn})$) with
the AH approach or with the default 
move approach if the default move for player $i$ is $\tau_i$.
\end{proposition}

Note that if $\vec{\tau}$ is a $k$-punishment equilibrium
$\vec{\tau}$ with respect to some $\mu$ in $S$, then it is a
$k$-punishment equilibrium with respect to all $\mu \in S$, since all $\mu$
give the same utility to the players. 

\begin{proof}
Since we are considering only distributions with rational
probabilities, the set $S$ is countable. Let $S = \{\mu^1, \mu^2,
\ldots\}$. Consider a strategy $\vec{\sigma} +  \sigma_d$ in
which each player sends its type to the mediator when it is first 
scheduled and the mediator waits until it receives the type profile
of all players. Let $N$ be the number of times that the
mediator is scheduled before receiving all the types. The
mediator samples an action profile $\vec{a}$ from $\mu_N$ (or from
$\mu_1$ if $S$ is finite and has fewer than $N$ elements), and sends
$a_i$ to each player $i$. Players play whatever they receive from the
mediator. If player $i$ never receives an action, they play $\tau_i$
(either because it is the default action or because it leaves it in
its will). As in the proof of Proposition~\ref{prop:desc-asyn}, it is
easy to check that $\vec{\sigma}$ induces $S$. To see that
$\vec{\sigma}$ is a $k$-resilient (resp., strongly $k$-resilient)
sequential equilibrium, note that a player $i$ cannot gain by not
sending its type to the mediator, since then the 
remaining players will play $\vec{\tau}$, which is strictly worse for
$i$ than
the equilibrium payoff. Moreover, given that all players send an input
to the mediator, since all correlated strategies $\mu \in S$ are
$k$-resilient (resp., strongly $k$-resilient) communication
equilibria, no coalition of $k$ or fewer players can gain
by lying to the mediator or by playing an action other than the
one the mediator sends. 

It remains to show that playing $\tau_i$ is optimal if player $i$
never gets a message from the mediator. If a player $i$
does not
receive an action to play,
then this must be because the mediator did not receive all inputs from
the players, because the mediator does not deviate from its strategy.
Thus, $i$ will believe that no player received a message from the
mediator,  
ans so all players will will play their part of $\vec{\tau}$. Since 
$\vec{\tau}$ is a $k$-resilient Nash equilibrium,
playing $\tau_i$ is optimal. 
\end{proof}

\commentout{
\begin{proof}
  Given 
  $\mu \in SE_k(\Gamma_{d, asyn})$ (resp., $\mu \in SE_k^S(\Gamma_{d, asyn})$),
  consider the strategy $\vec{\sigma}_{\ACT}$ in
$\Gamma_{d, asyn}$ where each player sends its type to the
  mediator, the mediator waits until it receives the type of all
players (note that it may never receive the types of all players),
then computes 
$\vec{a} := \mu(\vec{t})$ and sends $a_i$ to 
each player $i$.
If a player receives an action from the
mediator, it plays that action. If player $i$ never receives an action
from a mediator, it plays $\tau_i$
(either because it is the default action or because it leaves it in its will).
It is easy to check that
$\vec{\sigma}_{\ACT}$ is a $k$-resilient (resp., strongly $k$-resilient) sequential equilibrium: since
$\vec{\tau}$ 
is a $k$-punishment equilibrium with respect to $\mu$, it
is never in the interest of any player not to send its input. Since 
$\mu \in Com_k(\Gamma)$, it is also not in the interest 
of the players to lie about their inputs
or to play an action different from what the mediator
sends. Moreover, if a player $i$ does not receive an action to play,
it will believe that no one else did
(note that if a player receives an action from the mediator, it is guaranteed that all players will also receive one eventually since the mediator cannot deviate from the protocol), and thus that everyone else will play their part of $\vec{\tau}$. Since
$\vec{\tau}$ is a $k$-resilient Nash equilibrium,
playing $\tau_i$ is optimal. 
\end{proof}
}

Proposition~\ref{prop:desc3} implies the following: 

\begin{corollary}
    If $\Gamma = (P, T, q, S,U)$ is a Bayesian game 
and $\mu$ is a $k$-resilient (resp., strongly $k$-resilient)
correlated equilibrium of $\Gamma$ such that  there exists a
$k$-punishment equilibrium  $\vec{\tau}$ w.r.t. $\mu$, then $\{\mu\}
\in SE_k(\Gamma_{d, asyn})$ (resp., $\{\mu\}  \in SE_k^S(\Gamma_{d,
  asyn})$) with the AH approach or with the default 
move approach if the default move is $\vec{\tau}$.
\end{corollary}
}

\subsection{Outline of the proof}

The proof of 
Theorem~\ref{thm:main}
is
divided into two parts. We first 
describe a relatively simple strategy profile $\vec{\sigma}$ in $\Gamma_{\ACT}$
that implements the desired 
strategy profile in $\Gamma_d$
and
tolerates deviations of up to $k$ players.
Moreover, with $\vec{\sigma}$,
no subset
of at most $k$ players can sabotage the joint computation or learn
anything other than what they were originally supposed to learn. These
security properties of $\vec{\sigma}$ guarantee that it 
is a $k$-resilient Nash equilibrium.

The next step in the proof is
extending $\vec{\sigma}$ to a $k$-resilient sequential equilibrium
$\vec{\sigma}^{seq}$. The key idea for doing this is
showing that there exists a belief system $b$ such that all coalitions
$K$ of at most $k$ players believe that all players that deviated from
the protocol did so by sending inappropriate messages only to players in $K$;
that is, each coalition $K$ of size at most $k$ believes
that if $i, j \not \in K$, then $i$ and $j$ always send each other the messages
that they are supposed to send according to $\vec{\sigma}$. 
If players have belief system $b$, then all
coalitions $K$ with $|K| \le k$ believe, regardless of their
local history, that the remaining players will always compute their
part of the sequential equilibrium correctly, since $\vec{\sigma}$
tolerates deviations by up to $k$ players 
and all the players not in $K$
get the  messages that they are supposed to get according to
$\vec{\sigma}$ from all players except possibly those in $K$ (and thus
do not observe deviations by more than $k$ players). 
Moreover, in the construction
of $\vec{\sigma}$, if fewer than $k$ players deviate, all players
eventually receive enough information to compute their own action,
even if they don't take part in the computation. These two
properties are enough to guarantee that players in $K$ cannot gain by
deviating during the communication phase, since (they
believe that) nothing they can do will affect
what the remaining players play in the underlying game, and they
can't learn any additional information. 

Constructing a strategy profile 
that implements a given 
sequential equilibrium
with a mediator is in general nontrivial, since players must not only jointly
compute the action that the mediator sends to each player, but
this must be done in such a way that 
$i$ does not learn anything from the communication phase other than
its own action $a_i$. 
For if $i$ was able to get such information, $i$ might be tempted to
deviate (since it might have more information about what the utility
of 
deviating
would be).

There are well-known techniques in distributed computing for
constructing protocols with the appropriate security properties.
The standard solution (which we also use) involves distributing the
information about the necessary computations between the players in
such a way that each player knows only a part or \emph{share}
of that information.
Players can thus jointly compute the mediator's local history and
simulate everything that the mediator would do,
without leaking any information.
More precisely, if a player $i$ would send its type $t_i$ to the
mediator in the mediator game, in the cheap-talk game
it instead shares $t_i$ among all players in such a way that no
player (or coalition of players) 
can learn anything about it. When all the types have been distributed,
players can then compute the actions 
that the mediator would have sent to each player in the game with the
mediator. However, each player $i$ 
does not directly compute its own action $a_i$, but rather its share of every
action $a_j$. Then each player $i$ sends the share of each action
$a_j$ to player $j$. This guarantees that only $j$ learns $a_j$. 

There are two standard primitives that we use in this construction,
\emph{verifiable secret sharing} (VSS)
and \emph{circuit computation} (CC).  Roughly speaking, VSS takes as
input a value $v$, and distributes to each player $i$ a share $s^i_v$
of $v$.  If we are interested in $k$-resilience, we can arrange for
these shares to have the property that knowing a subset of at most $k$
shares gives no information about $v$, while knowing 
strictly more than $k$
suffices to reconstruct $v$. CC is used to distribute the
shares of the output of a function for which each of the players has
shares of all of its inputs.  More precisely, given a function $f: D^t
\rightarrow D$ and $t$ values $v_1, \ldots, v_t$ such that each player
$i$ knows only its own share $s_{v_j}^i$ of $v_j$, for $j=1, \ldots, t$, the
CC procedure takes as inputs the shares $s_{v_1}^i, \ldots,
s_{v_t}^i$ of each player $i$, and allows player $i$ to compute
its share $s^i$ of $f(v_1, \ldots, v_t)$ in such a way that no player
learns any information about $v_1, \ldots, v_t$. Therefore, VSS
allows players to share private values in a secret way, and CC allows
players to generate (shares of) new values without leaking any
information about all values previously shared.
We define VSS and CC formally in 
Appendix~\ref{sec:tools2},
and show that,
if $n > 3k$,
then
it is possible to implement VSS and CC
in a $k$-resilient way, so that they both tolerate deviations by up to $k$
players.
In the next section, we prove Theorem~\ref{thm:main} assuming that we
can implement VSS and CC.

\section{Proof of Theorem~\ref{thm:main}}\label{sec:proof}

In this section, we prove Theorem~\ref{thm:main} for the case of
$k$-resilient sequential equilibrium. The case of strongly
$k$-resilient sequential equilibrium is similar.

\subsection{Constructing a $k$-resilient Nash equilibrium}\label{sec:nash}

By definition, $SE_k(\Gamma_{\ACT}) \subseteq SE_k(\Gamma_{d})$. It thus remains to show that a
    $k$-resilient
    sequential
equilibrium
with a mediator can be implemented by a 
$k$-resilient sequential equilibrium in the cheap-talk game.
Let $Com_k(\Gamma)$ be the set of $k$-resilient communication equilibria of $\Gamma$.
As we noted in
Section~\ref{sec:results}, $SE_k(\Gamma_{d}) =
Com_k(\Gamma)$, and thus $SE_k(\Gamma_{\ACT}) \supseteq
SE_k(\Gamma_{d})$ reduces to showing that all $k$-resilient
communication equilibria $\mu$ of $\Gamma$ can be implemented with
cheap talk without a mediator.
We begin by constructing a strategy
profile $\vec{\sigma}$ in $\Gamma_{\ACT}$ that is a $k$-resilient
Nash equilibrium and induces the same
correlated strategy profile as
$\mu$.

Given $\mu$, consider the canonical protocol $\vec{\sigma}_{can}$ for $n$ players and a mediator 
in which each player $i$ sends its type $t_i$ to the mediator, and the
mediator then computes the type profile $\vec{t}$ of the players using
the messages received, replacing the types of players that didn't
send theirs by a default type $\bot$.
(Since we are working in synchronous systems for this part of the
proof, the mediator can tell if a player did not send his type.)
The mediator then
samples $\vec{a} \in \Delta(A)$
according to the distribution
$\mu(\vec{t})$ and sends $a_i$ to each player 
$i$.
If $i$ sent its true type at the beginning of the game
and receives an action $a_i$ from the mediator, it plays $a_i$.
Otherwise,
it plays the best response 
assuming that the remaining
players are playing their part of $\mu(\vec{t})$ (note that the best
response is well-defined since both $\mu$ and the distribution $q$ of
type profiles are common knowledge). It is straightforward to check
that $\vec{\sigma}_{can}$ induces the same
correlated strategy profile 
as $\mu$, and that $\vec{\sigma}_{can}$ is a
$k$-resilient sequential equilibrium. The details are left to the
reader.

The idea for constructing $\vec{\sigma}$ is that instead of having a
mediator that receives the types of the players and then computes
$\mu$, the players perform these tasks by themselves using VSS and CC
to help them simulate the 
mediator:
each player $i$, instead of sending its type $t_i$ to the mediator, 
shares it 
among
all the players using a $k$-resilient VSS
implementation. After the number of rounds needed to perform VSS,
players compute $(a_1, \ldots, a_n) := \mu(t_1, \ldots, t_n)$ with CC,
using as input the shares of each of the types $t_i$. Each player $i$
then sends to each player $j$ $i$'s share of $a_j$. Thus, each
player $i$ can compute  $a_i$ without learning anything
about the other players' actions. 

Note that, since $\vec{a}$ is sampled from $\mu(\vec{t})$, players are
required to randomize in a coordinated way.
The problem of jointly sampling an element from a
distribution with rational probabilities over a finite set can be
reduced to the problem of sampling an integer from $[N] := \{1,
\ldots, N\}$ for some $N$, as the following proposition shows.  

\begin{proposition}\label{prop:sample}
Given a 
correlated strategy profile $\mu$
with rational parameters,
there exists an integer $N$ and a function
$\mu^* : T \times [N]
\rightarrow A$
such that, if $X_{\vec{t}}$ is the distribution
over action profiles 
obtained from playing $\mu^*(\vec{t}, r)$
when $r$ is uniformly sampled from $[N]$,
then $X_{\vec{t}} = \mu(\vec{t})$
for  all $\vec{t} \in T$. 
\end{proposition}

\begin{proof}
  Fix an integer $N$ such that, 
  for all $\vec{t}$ and all action profiles $\vec{a}$, the 
probability $\mu(\vec{t})(\vec{a})$ is an integer 
  multiple of $1/N$. For each $\vec{t}$, 
partition $N$ into subsets $C_1^{\vec{t}}, C_2^{\vec{t}}, \ldots,
C_m^{\vec{t}}$ such that $|C_i^{\vec{t}}|$ 
  is proportional to the probability of $\mu(\vec{t})(\vec{a}^i)$ for
some indexing $\vec{a}^1, \ldots, \vec{a}^m$ of the action profiles. 
We take $\mu^*(\vec{t}, r)$ to be  $\vec{a}^i$ 
if $r \in C_i^{\vec{t}}$ for some $i$.  It is
straightforward to check that $\mu^*$ satisfies the desired
properties.  
\end{proof}

Players can thus compute the desired action profile using a
deterministic function $\mu^*$, given that they are able to jointly
sample an integer $r \in N$ uniformly at random.

We next present a
formal construction of the protocol $\vec{\sigma}$ described above.
This protocol is an adaptation of the  multiparty computation
protocol of Ben-Or, Goldwasser, and Wigderson
\nciteyear{BGW88} to a setting in which the output of the function
must be sampled from a common distribution. 

\begin{enumerate}
    \item [Phase 1:] Each player $i$ generates a random number $r_i$
            between $1$ and $N$ uniformly at random. Then $i$ shares $r_i$
      and its type $t_i$ using a $k$-resilient VSS procedure. If a
         player $i$ does not share $t_i$ or $r_i$, then $t_i$ and
      $r_i$ are taken to be a default type $\bot \in T_i$ and $0$,
      respectively. 
    \item [Phase 2:] For each player $j$, let $r_j^i$ be $i$'s share
      of $r_j$ and let $t_j^i$ be $i$'s share of $t_j$. Each player $i$
      initiates a $k$-resilient CC of the functions $\mu^*_1, \ldots,
      \mu^*_n$, where $\mu^*_j(\vec{t}, \vec{r})$  is
      the $j$th component of $\mu^*(\vec{t}, r)$ and $r := r_1 +
      \ldots + r_n \bmod N$. Player $i$'s inputs for each of these CC
      instances are the shares $r_1^i, \ldots, r_n^i$ and
      $t_1^i, \ldots, t_n^i$ computed in Phase 1. 
    \item [Phase 3:] For each player $i$, let $o_j^i$ be $i$'s output
            of the CC instance of $\mu^*_j$ (this output is $i$'s share of
            $\mu^*_j(\vec{t}, \vec{r})$). For each player $j$, player
            $i$ sends $o_j^i$ to $j$.             
    \commentout{
    \item [Phase 4:] Whenever $i$ receives at least $2k+1$ shares
      $o_i^j$ from different players such that any subset of size
      $k+1$ of those shares defines the same secret $s_i$, $i$ plays
      $s_i$ and terminates. If $i$ never gets to compute $s_i$, it
      plays a default action $\bot \in S_i$. 
      }
        \item [Phase 4:] If $i$ shared its true type in Phase 1 and
          $i$ receives at least $2k+1$ shares 
          $o_i^j$ from different players such that every subset  of
          those shares of size $k+1$ defines the same secret $a_i$,
          $i$ plays $a_i$ and terminates. Otherwise, $i$ plays a best
          response given 
          its local history,
      assuming that all the remaining players
      shared their true type and
      play their part of the
      communication equilibrium. 
      The optimal response depends on which was the value $t_i$ that $i$ shared in Phase 1 (note that $t_i$ is 0 if $i$ didn't successfully share a type), and which shares of $o_i$ $i$ received from other players. Note that these values give a posterior distribution over the possible types of the other players and the actions that they played since $i$ is assuming that they share their true type $\vec{t}_{-i}$ and
      play their part of the
      communication equilibrium with input $\vec{t}$. 
\end{enumerate}

Since the VSS and CC primitives are $k$ resilient, the protocol above
is explicitly defined in all scenarios in which at most $k$ players
deviate. However, there might be histories in which more than $k$
players deviate and other players are not able to continue because
either they 
didn't receive enough messages or the messages they received were
inconsistent. (For example, players will not able to jointly compute
the functions $u_i^*$ in Phase 2 if insufficiently many players send
messages. This means they 
won't
be able to start Phase 3.) We extend
the protocol to cover all these cases using the following convention:
if a player is not able to continue because the protocol does not
explicitly state what to do in a certain situation, the player does
not take part in subsequent phases. For instance, if a player $i$ is
not able to compute the functions $u_i^*$ in Phase 2, then $i$ does
not take part in 
Phase 3. Note, however, that the action played at the end of Phase 4
is always specified even if a player doesn't complete some of the
previous phases. 

\begin{proposition}\label{prop:nash-eq}
If $n > 3k$, $\vec{\sigma}$ is a $k$-resilient Nash equilibrium and
induces the same 
correlated strategy profile 
as $\mu$. 
\end{proposition}

\begin{proof}
  Clearly $\vec{\sigma}$ induces the same 
correlated strategy profile 
as $\mu$ by construction. The fact that $\vec{\sigma}$ is a
$k$-resilient NE follows directly from the properties of $k$-resilient
VSS and $k$-resilient CC: no matter what a coalition of at most $k$
players does, the remaining $n-k$ players are guaranteed to terminate
all their VSS and CC instances correctly. Since honest players choose
$r_i$ uniformly at random, the number
$r := r_1 + \ldots + r_n \bmod N$ 
is
also randomly distributed in $[N]$ regardless of what a coalition of
at most $k$ players decides to share. This guarantees that the actions
$a_1, \ldots, a_n$ encoded by the shares $(o_1^1, \ldots, o_1^n),
\ldots, (o_n^1, \ldots, o_n^n)$ of honest players are distributed 
according to $\mu(\vec{t})$. The correctness of the output of each
player follows from the fact that if a player $i$ waits until
receiving at least $2k+1$ shares that define the same secret $a_i$, at
least $k+1$ of those shares were sent by honest players and thus
define the correct value $a_i$ (note that if $n > 3k$, all players are
guaranteed to be able to reconstruct their action $a_i$ in Phase 4
even if a coalition of $k$ players 
deviates,
since at least $2k+1$
players are honest). To complete the proof, note that the secrecy
properties of the $k$-resilient VSS and CC guarantee that no coalition
of at most $k$ players can learn anything besides their own actions,
and thus that playing the action $a_i$ computed in Phase 4 is
optimal. 
\end{proof}

\subsection{Constructing a $k$-resilient sequential equilibrium}\label{sec:sequential} 

Proposition~\ref{prop:nash-eq} shows that, for all $\mu \in
Com_k(\Gamma)$, there exists a strategy $\vec{\sigma}$ in
$\Gamma_{\ACT}$ such that $\vec{\sigma}$ is a $k$-resilient Nash
equilibrium that implements $\mu$. Our aim is to extend $\vec{\sigma}$
to a $k$-resilient sequential equilibrium. The difficulty in
constructing such an extension is due to the fact that the $k$-resilience of
VSS and CC guarantees that no coalition of at most $k$ players can
prevent the remaining players from carrying out the VSS and CC
computations 
if
no other player deviates (which is all that is needed
to show that the construction 
gives a $k$-resilient Nash equilibrium). However, it is not clear what
the 
coalition can accomplish starting at a point off the equilibrium path
where they detect that other players are 
deviating as well.  

For instance, consider a normal-form game $\Gamma$ for $n$ players
such that the set of actions for each player is $[n]$. 
The payoffs are as follows:
$\vec{a}$,
if at least $n-k$   coordinates of an action profile $\vec{a}$ are 
different, then 
all players get a payoff of $0$; otherwise, all players get
a payoff of $1$.
Consider the $k$-resilient correlated equilibrium $p$
in which all possible permutations of $[n]$ are played with equal probability,
giving all players a payoff of 0.
Suppose that players play the 
protocol $\vec{\sigma}$ described in Section~\ref{sec:nash}. If at
least $n-k$ of the players are honest, it is guaranteed by
construction that all of them will play different actions by the end
of the protocol, and thus that all players will get a payoff of 0
regardless of what the remaining $k$ players do. However, consider the
case that a coalition $K$ of at most $k$ players follows the protocol but
at some point detects an inconsistency among their local
histories. Since players in $K$ were following the protocol
faithfully, it must be the case that one of the remaining $n-k$
players deviated. If this leads the players in $K$ to believe that
there is a chance that at least two of the remaining players played the
same action $a \in [n]$, it is worthwhile for players in $K$ to
deviate and play $a$ at the end of the protocol. 
\commentout{
which means that 
the payoff of each
player if an action profile in the support of $p$ is 0.
Consider the following situation, where
players play the 
protocol $\vec{\sigma}$ described in Section~\ref{sec:nash}: if a
coalition $K$ of at most $k$ players is being honest but they detect
an inconsistency between their local histories, it must be the case that one of
the remaining $n-k$ players deviated. Since $k$-resilient
implementations of VSS and CC can guarantee the correctness of
the outputs only when at most $k$ players deviate, it is certainly
possible that there exists a strategy for players in $K$ such that,
from this point on,
they get a payoff of 1.
}
If this is the case, then we do not have a sequential equilibrium.

Our approach to showing 
that
$\mu$ can be implemented with a $k$-resilient
sequential equilibrium is 
to show that 
there exists a belief system such that
all 
coalitions $K$ of size at most
$k$ believe, regardless of their local history, that 
\begin{itemize}
\item [(a)] the remaining players will successfully complete their
  part of $\vec{\sigma}$, no matter what the players in $K$ do;
\item [(b)] players in $K$ cannot infer anything about the
  action profile $\vec{a}$ being played other than their own part
  $\vec{a}_K$ (and $\mu$). 
\end{itemize}
It is important to stress that property (b) does in fact depend on the
beliefs of players in $K$: If some player $i$ in $K$ receives a
message $\bot$ from a player $j$ that was not supposed to be sent on
the equilibrium path, if $i$ believes that $j$ sends that message only
when $j$ has type $t_j$ (whether or not this is actually the case),
then $i$ will believe that it has acquired important information when
it receives $\bot$, and adjust its play accordingly.
However, if $i$ believes that $j$ chose that
message uniformly at random, then $i$ will not adjust its play. The
following proposition shows that if properties (a) and (b) are
satisfied, then $\vec{\sigma}$ is $k$-sequentially rational 
\commentout{
(i.e., it
  is never in the interest of coalitions of size at most $k$ to deviate from
$\vec{\sigma}$)  during the whole communication phase. 
}
 (i.e., it
is never in the interest of coalitions of size at most $k$ to deviate from
$\vec{\sigma}$). 

\begin{proposition}\label{prop:seq-rational}
If $b$ is a belief system compatible with $\vec{\sigma}$ that 
  satisfies properties (a) and (b), then $(\vec{\sigma}, b)$ is
a $k$-resilient sequential equilibrium.
\end{proposition}


\begin{proof}
  If $b$ satisfies property (a), then at all nodes in the game tree
  (including nodes off the equilibrium path), no matter what the
  players in $K$ do, 
  the remaining players will eventually terminate their part of
$\vec{\sigma}$ and thus sample an action profile $\vec{a} \in
\mu(\vec{t})$ and play their part $\vec{a}_{-K}$. Property (a) also
guarantees that by the end of Phase 4 (and the end of the
communication phase), players not in $K$ will send the shares of
$\vec{a}_K$ to players in $K$, which means that players in $K$ believe
that they will always learn $\vec{a}_K$. Property (b)
guarantees that players in $K$ won't learn anything besides their own
actions. 
Properties (a) and (b) together imply that players in $K$ believe
that they cannot prevent the remaining 
players from playing their part of an action profile $\vec{a}$ sampled
according to
$\mu$, and that players in $K$ will eventually learn $\vec{a}_K$ and nothing else.
Since $\mu$ is a $k$-resilient
communication equilibrium,
it is always a best response for the players in $K$ to
share their own type in Phase 1, as long as they can continue the
basic protocol.
Similar 
arguments
show that following the basic protocol if
they can (i.e., they receive messages enough consistent with the
protocol) is a best response in phases 2 and 3.  Moreover, it also
follows from properties (a) and (b) that, since the players in $K$
cannot influence the outcome, doing nothing is a best response if they
cannot continue the basic protocol (in all phases). 
It remains to show that the action played in Phase 4 is optimal. By
construction, this
reduces to showing that if a coalition $K$ of players with $|K| \le k$
shared their true type in Phase 1 and were able to compute their
actions $a_i$ in Phase 4, 
it is optimal for them to play $a_i$
(note that otherwise $i$ is already best responding given that the
belief system $b$
satisfies properties (a) and (b)). 
This follows from the fact that, by properties (a) and (b), players in
$K$ believe that all remaining players computed $\vec{a}$ and thus
played $\vec{a}_{-K}$. Since $\mu$ is a
$k$-resilient communication equilibrium, it is then always in the
interest of players in $K$ to play $\vec{a}_K$. 
\end{proof}

\commentout{

It follows from Proposition~\ref{prop:seq-rational} if a belief system
$b$ satisfies properties (a) and (b), strategy $\vec{\sigma}$ can be
extended in a simple way to a $k$-resilient communication equilibrium
$\vec{\sigma}^{seq}$ that implements $\mu$:
each coalition of players
$K$ with $|K| \le k$ plays $\vec{\sigma}_K$ during the communication
and then plays the action in the underlying game  that maximizes their
utility, given their histories and their beliefs.
Thus, provided that we can find a belief system consistent with
$\vec{\sigma}^{seq}$ that has the properties (a) and (b), we are done.

Given $b$, let $(\vec{\sigma})_1, (\vec{\sigma})_2, \ldots$ be a
sequence of completely-mixed strategies in $\Gamma_{\ACT, syn}$ such
that $\lim_{m \rightarrow \infty} (\vec{\sigma})_m  = \vec{\sigma}$
and such that the beliefs induced by $(\vec{\sigma})_m$ converge to
$b$. Let $(\vec{\sigma})_1^{seq}, (\vec{\sigma})_2^{seq}, \ldots$ be a
sequence of protocols such that $(\vec{\sigma})_m^{seq}$ is identical
to $(\vec{\sigma})_m$ except that, whenever players follow the
protocol in Phase 4, they play the best response instead of whatever
they would have played in $\vec{\sigma}$. Let $b^{seq}$ be the limit
of the beliefs systems induced by $(\vec{\sigma})_m$ when $m
\rightarrow \infty$. Then we have the following: 

\begin{proposition}\label{prop:k-seq-eq}
If $b$ is a belief system compatible with $\vec{\sigma}$ that
  satisfies properties (a) and (b), then $(\vec{\sigma}^{seq}, b^{seq})$
is a $k$-resilient sequential equilibrium. 
\end{proposition}

\begin{proof}
The proof reduces to show that the best response whenever a coalition $K$ of players of size at most $k$ follows the protocol and is able to terminate all phases is to play the actions $\vec{s}_K$ computed in Phase 4. To see this, note that, by construction, properties (a) and (b) are also satisfied by $b^{seq}$, which means that players in $K$ always believe that the remaining players were able to sample an action profile $\vec{s} \leftarrow \mu(\vec{t})$ and compute $\vec{s}_{-K}$. If playing $\vec{s}_{-K}$ is their best response, properties (a) and (b) imply that players in $K$ always believe, according to $b^{seq}$, that the remaining players are playing $\vec{s}_{-K}$, which exactly the same that players in $K$ would have believed if they had belief system $b$. Thus, a best response with $b$ is also a best response with $b^{seq}$.
It remains to show that playing $\vec{s}_K$ is a best response with belief system $b$ whenever a coalition $K$ of players of size at most $k$ follows the protocol and is able to terminate all phases successfully. If players in $K$ compute action profile $\vec{s}_K$, by properties (a) and (b), they believe that the remaining players computed $\vec{s}_K$ and thus played $\vec{s}_{-K}$ as described in $\vec{\sigma}$. Since $\mu$ is a $k$-resilient communication equilibrium, it is then always in the interest of players in $K$ to play $\vec{s}_{K}$.
\end{proof}
}

Proposition~\ref{prop:seq-rational} shows that Theorem~\ref{thm:main}
reduces to finding a belief system $b$ compatible with $\vec{\sigma}$
that satisfies properties (a) and (b). 
We
start by showing how such a belief system can be constructed for a
fixed subset $K$ with $|K| \le k$.
Consider the sequence 
$\vec{\sigma}^1, \vec{\sigma}^2, \ldots$ of strategy profiles in which player $i$
proceeds as follows with $\vec{\sigma}^m_i$: Given local history $h_i$,
$i$ 
does exactly what it would have done with $\vec{\sigma}$,
except that if it is supposed to send a message $msg$ to another
player $j$, if 
$j \in K$, with probability $1/m$ it replaces $msg$ by a random
message $msg'$ sampled according to some fixed distribution that gives
all messages positive probability, and sends $msg'$
instead.
If $j \not \in K$, $i$ replaces $msg$ with $msg'$ with probability $m^{-m}$.
That is, players play as in $\vec{\sigma}$, except that they
may lie to players in $K$ with a low probability,
and to players outside of $K$ with a much lower probability. 
By construction, each $\vec{\sigma}^m$ is completely mixed and the belief
$b$ induced by the limits of the beliefs in  
$\vec{\sigma}^1, \vec{\sigma}^2, \ldots$
satisfies the properties (a) and (b)
described above for the set $K$.
\commentout{
\footnote{The strategies $(\vec{\sigma})_1^K,
  (\vec{\sigma})_2^K, \ldots$ are not fully mixed for ease of
  exposition. However, the same properties are satisfied if we assume
  that players  lie also to players not in $K$ but with a much smaller
  probability (for instance $m^{-m}$).}. 
  }
  To see this, note that
property (a) is satisfied because players in $K$ believe that the
remaining $n-k$ players are being truthful between themselves
(note that the probability of anyone lying to a player not in $K$ is
negligible)  
and
thus, since the VSS and CC implementations are $k$-resilient, that
they can carry out all the necessary computations without the help of
players in $K$. Moreover, since players in $K$ believe that if they get a
message incompatible with $\vec{\sigma}$ (i.e., a message
that could not have been according to $\vec{\sigma}$, given their joint
histories), then they are getting a message
drawn from some fixed 
distribution,
they learn no useful
information about the action profile being played by getting this
message. Property (b) follows.

This example shows how to construct a $k$-belief system compatible
with $\vec{\sigma}$ that satisfies properties (a) and (b) for a
particular subset $K$ with $|K| \le k$, but it unfortunately does
not satisfy these properties for all such sets simultaneously. Note
that it is critical that players in $K$ believe that they are the only
ones being lied to, and also that the lies sent to players in $K$
convey no information since they do not depend on the senders' local
histories (in  the example, they are sampled
from a fixed distribution). The next definitions generalizes
these ideas: 

\commentout{
This example
shows that if a subset $K$ of players of size at most $k$ believe that
they are the only ones being lied to, 
then properties (a) and (b) hold for that subset $K$. This
motivates the following definitions:
}

\begin{definition}
\commentout{
  Player $i$ \emph{lies relative to $\vec{\sigma}$ in global history
    $\vec{h}$} if $i$ sends a message to a player $j \ne i$ in $\vec{h}$ that
    is incompatible with $\vec{\sigma}$.  
A subset $K$ of players is \emph{truthful relative to
  $\vec{\sigma}$ in $\vec{h}$} 
  if, in $\vec{h}$, no player in $K$ ever 
lies to another player in $K$.
}
Given history $h_i$ of player $i$, a message $msg$ in $h_i$ from $i$
to $j$ is a \emph{lie}  
if it is incompatible with $\vec{\sigma}$. A subset $K$ of players is \emph{truthful relative to
  $\vec{\sigma}$ in global history $\vec{h}$} 
  if, in $\vec{h}$, no player in $K$ ever 
  \emph{lies} (i.e., sends a message that is a lie) to another player in $K$.
\end{definition}

\commentout{
Note that the message that $i$ was supposed to send to $j$ is uniquely
determined by $i$'s local history, since it includes all of $i$'s
randomization.
}
\commentout{
Note that, given a global history $\vec{h}$, we can determine whether
a message sent by $i$ to $j$ is compatible with $\vec{\sigma}$ even if
$\sigma_i(h_i)$ gives a distribution over messages.  This is because
$h_i$ includes all of $i$'s randomization up to the point where it is
supposed to send a message, so we can tell which message actually
should be sent.
}
Note that since the local history $h_i$ of player $i$ includes all of
$i$'s randomization up to the point where it is 
supposed to send a message, given 
local history $h_i$,
there is a unique
message sent by $i$ to $j$ that is compatible with $\vec{\sigma}$,
even if $\vec{\sigma}$ is a mixed strategy (intuitively, when $i$
sends a message to $j$, all of $i$'s coins have already been tossed
and their outcomes appear in $i$'s history). 
A message from player $i$ to player $j$ in local
history $h_i$ 
is \emph{correct} if it is not a lie (relative to the protocol $\vec{\sigma}$).

\begin{definition}[$k$-paranoid belief system]\label{def:k-paranoid}
Given a $k$-belief system $b$ and the local histories $\vec{h}_K$ of a
  subset $K$ of players, a global history $\vec{h}$ is
\emph{$b$-consistent} with $\vec{h}_K$ if the probability according to $b$ that
players have global history 
$\vec{h}$
conditional on players in $K$ having
history $\vec{h}_K$ is non-zero. 
Given a protocol $\vec{\sigma}$,
a $k$-belief system $b$ is \emph{$k$-paranoid} if
\begin{itemize}
  \item[(1)] for all subsets $K \subseteq [n]$ with $|K| \le k$ and all
histories $\vec{h}_K$ of players in $K$, in all histories $\vec{h}$ that are
$b$-consistent with $\vec{h}_K$, the set $\overline{K}$ of players is
truthful relative to $\vec{\sigma}$ (where $\overline{K}$ is the
complement of $K$);  
\item[(2)] there exists a sequence $\vec{\sigma}^1, \vec{\sigma}^2,
  \ldots$ of protocols that induces $b$ such that the lies relative 
  to $\vec{\sigma}$ sent in $\vec{\sigma}^1, \vec{\sigma}^2, \ldots$
  do not depend on the players' local histories. 
\end{itemize}
\end{definition}

Intuitively, a $k$-belief system is $k$-paranoid if all subsets of
at most $K$ players believe, regardless of their local history, that
they are the only ones being lied to, 
and also that the lies received convey no information.
Our earlier discussion provides a proof of the following proposition.

\begin{proposition}\label{prop:paranoid-implies-properties}
Let $\vec{\sigma}$ be the protocol presented in Section~\ref{sec:nash}. If $b$ is a $k$-paranoid belief system consistent with $\vec{\sigma}$, then $b$ satisfies properties (a) and (b).
\end{proposition}

Thus, by 
Proposition~\ref{prop:seq-rational},
Theorem~\ref{thm:main}
reduces to find a $k$-paranoid belief system $b$ that is
consistent with $\vec{\sigma}$. The first candidate for generating
such a belief system is the sequence  $\vec{\sigma}^1,
\vec{\sigma}^2, \ldots$ of strategy profiles  
mentioned earlier, except that with
strategy $\sigma_i^m$, rather than player $i$ lying to player $j$ with 
probability $1/m$
if $j$ is in some fixed set $K$ and with probability $m^{-m}$ if $j$
is not in $K$,
now $i$ lies to all players with probability
$1/m$.
\commentout{
(as opposed 
to our previous example, in which they lied with probability $1/m$
only to players in a fixed subset $K$).  
}
We might
hope that if players detect that they received inappropriate messages, they
would assume that they are the only ones being lied to, since lies are
highly unlikely. For instance, if a player $i$ receives four
inconsistent messages, we would hope that $i$ would believe that these
are the only  lies.   Unfortunately, this is not necessarily the case.
If an earlier lie by some other player could result in players
sending these messages to $i$ according to $\vec{\sigma}$, 
then $i$ would instead believe that there was only 
one
lie.
To prevent this, we want players to believe that earlier lies are much
less likely than 
later
lies.  It turns out that if we do this, then
we can get the result that we want.

\begin{proposition}\label{prop:paranoid}
For all strategies $\vec{\sigma}$ in $\Gamma_{\ACT, syn}$, there
exists a $k$-paranoid belief system $b$ that is consistent with
$\vec{\sigma}$. 
\end{proposition}

\begin{proof}
  Consider the sequence $\vec{\sigma}^1, \vec{\sigma}^2,
    \vec{\sigma}^3, \ldots$ of strategy profiles where
in $\vec{\sigma}^n$ players 
act
exactly as they do in 
$\vec{\sigma}$, except that whenever $i$ would send a message $msg$ to
another player $j$, 
with probability $m^{-\left((2n)^{-r}\right)}$ (where $r$ is the
current round of communication), $i$ replaces this message by a 
message $msg'$ chosen according to some fixed distribution and sends
$msg'$ instead.
Intuitively, these are protocols where each player lies to
other players with a probability that (greatly) increases with the round
number and that decreases with $m$.
It is straightforward to check that 
\commentout{
the belief assessments induced by
this sequence of strategy profiles
}
$\vec{\sigma}^1, \vec{\sigma}^2,
    \vec{\sigma}^3, \ldots$
    converges to $\vec{\sigma}$.
  Let $b$ be the belief system induced by the sequence.  We want to
    show that $b$ is $k$-paranoid.
\commentout{
We need some
preliminary lemmas.
}
By construction, $b$ satisfies property (2) of Definition~\ref{def:k-paranoid} since lies in $\vec{\sigma}^1, \vec{\sigma}^2,
    \vec{\sigma}^3, \ldots$ are sampled from a fixed distribution. To
        show that it also satisfies property (1), we need some
        preliminary lemmas.  

Given a global history $\vec{h}$ of
$\vec{\sigma}$,
let  
$L(\vec{h})$
be the sequence $(\ell_1, \ell_2, \ldots, \ell_R)$, where $R$ is the
total number of communication rounds in $\vec{h}$ and $\ell_r$ is the number
of lies sent at round $r$ in $\vec{h}$.

\begin{lemma}\label{lemma:aux}
If $P^m$ is the probability on global histories induced by
$\vec{\sigma}^m$, and $\vec{h}$ and $\vec{h}'$ are global history
profiles 
such that  
$L(\vec{h}) < L(\vec{h}')$
in lexicographical order (the smaller sequence is the one that has the
smaller value in the first position in which they differ), then
$\lim_{m \rightarrow \infty} \frac{P^m(\vec{h}')}{P^{m}(\vec{h})} = 0$.
\end{lemma}

\begin{proof}
If
$L(\vec{h}) = (\ell_1, \ell_2, \ldots, \ell_R)$,
then $$\log(P^m(\vec{h})) = -\log(m)\left(\sum_{r = 1}^n \ell_r
(2n)^{-r}\right) + \Theta(1),$$ 
%
where $\Theta(1)$ denotes an expression that is bounded by some constant.
Thus, if $L(\vec{h}) = (\ell_1, \ell_2, \ldots, \ell_R)$ and $L(\vec{h}') =
(\ell'_1, \ldots, \ell'_{R'})$, we must show that if
$L(\vec{h}) < L(\vec{h}')$, then $\sum_{r = 1}^{R'} \ell_r (2n)^{-r} < \sum_{r =
    1}^{R'} \ell'_r (2n)^{-r}$, which is equivalent to showing that $\sum_{r =
  1}^{\max(R,R')} (\ell'_r - \ell_r) (2n)^{-r} > 0$ 
(taking the elements beyond the end of a sequence to be 0). 

Since each player sends exactly one message to each other player, we have
that $0 \le \ell_i \le n$. Therefore, if $r^*$ is the first round such
that $\ell'_r \not = \ell_r$, we have that $$\sum_{r = 1}^{\max(R,R')}
(\ell'_r - \ell_r) (2n)^{-r} \ge (2n)^{-r^*} - n \sum_{r = r^* +
  1}^{\max(R, R')}  
(2n)^{-r}
= (2n)^{-r^*}\left(1 - n\sum_{r = 
    1}^{\max(R, R')-r^*} (2n)^{-r} \right),$$
  and
   $$1 - n\sum_{r = 1}^{\max(R, R')-r^*} (2n)^{-r} > 1 - n\sum_{r > 0}
  (2n)^{-r} > 1 - n\sum_{r > 0} 2^{-r} n^{-1} = 0 $$ 
\end{proof}

\begin{lemma}\label{lemma:aux2}
  If $|K| \le k$ and $\vec{h}_K^*$
is a local history of the players in $K$, then 
for every global history profile $\vec{h}$ such that
$\overline{K}$ is not truthful in $\vec{h}$, $$\lim_{m \rightarrow \infty}
P^m(\vec{h} \mid \vec{h}_K^*) = 0.$$ 
\end{lemma}

\begin{proof}
  Suppose that for some global history $\vec{h}$, $P^m(\vec{h} \mid
  \vec{h}_K^*) > 0$ and $\overline{K}$ is not truthful in $\vec{h}$. Since  
$$P^m(\vec{h} \mid \vec{h}_K^*) = \frac{P^m(\vec{h})}{\sum_{\vec{h}' : \vec{h}'_K = \vec{h}_K^*} P^m(\vec{h}')},$$ by 
Lemma~\ref{lemma:aux}, this result reduces to finding a global history
 $\vec{h}'$ such that  
 $\vec{h}'_K = \vec{h}_K$ and
 $L(\vec{h}') < L(\vec{h})$.
 This history can be constructed as follows: suppose that $r$ is the
 first round in $\vec{h}$ where a player $i \in \overline{K}$ lies to
 another player $j \in \overline{K}$. Consider a history $\vec{h}'$ in
 which  
 players in $K$ perform the same internal computations and send
 exactly the same messages as in $\vec{h}$, but players in
 $\overline{K}$ only do so until round $r-1$. In round $r$, players in
 $\overline{K}$ perform exactly the same internal computations as in
 $\vec{h}$ and send the same messages to players in $K$, but they send
 \commentout{
 the correct messages to other players in  
 $\overline{K}$.
 }
 to other players in $\overline{K}$ the \emph{correct} messages (i.e.,
 the messages 
  players in $\overline{K}$ were supposed to send according to the
  protocol).  Notice that since a player $i$'s local history includes
  the random bits that the players needs, there is a unique correct
  message, even if player $i$ is uses a randomized protocol.
 From round $r+1$ on, players in $\overline{K}$ 
 send players in $K$ exactly the same messages as in $\vec{h}$.
 Clearly, 
 by construction,
  $\vec{h}'_K = \vec{h}_K$. Let $L(\vec{h}) = (\ell_1, \ell_2, \ldots,
 \ell_R)$ and $L(\vec{h}') = (\ell'_1, \ell'_2, \ldots,
 \ell'_R)$. Then $\ell_i = \ell'_i$ 
  for all $i < r$,
 since players 
 have the same local histories in $\vec{h}$ and $\vec{h}'$ for the
 first $r-1$ rounds of communication, 
 and they also send exactly the same messages.
 A similar argument shows that, for all $r' \ge r$, the players in
 $K$ tell identical lies in round $r'$ in both $h$ and $\vec{h}'$.  
    Finally, note that strictly fewer lies are told by players in
    $\overline{K}$ in  round $r$ of $\vec{h}'$ than in round $r$ of
    $\vec{h}$, since the players in $\overline{K}$ have  
    exactly the same local history in both $\vec{h}$ and $\vec{h}'$,
    send exactly the same 
 messages to players in $K$ 
 in both,
 but in $\vec{h}'$ they tell no lies to each other whereas, by assumption,
 there is at least one such lie told in $\vec{h}$. 
 Thus, $\ell'_r < \ell_r$, and $L(\vec{h}') < L(\vec{h})$ as desired.
\end{proof}

This shows that $b$ 
satisfies property (1) of Definition~\ref{def:k-paranoid}, and thus that $b$
is a $k$-paranoid belief system consistent with
$\vec{\sigma}$, completing the proof of Proposition~\ref{prop:paranoid}.
\end{proof}

\commentout{
\section{Proof of Theorem~\ref{thm:main2}}

The proof of Theorem~\ref{thm:main2} is quite similar to the proof of
Theorem~\ref{thm:main}. Given a strategy profile $\vec{\sigma}' +
\sigma'_d \in SE_k(\Gamma_{d, asyn})$,
we first construct a $k$-resilient Nash equilibrium
$\vec{\sigma}$ that implements $\vec{\sigma}' + \sigma'_d$ and
guarantees that, if less than $k$ players deviate, all players will 
eventually be able to compute the action that they should play,
and then we construct a $k$-paranoid belief system $b$ that is
consistent with $\vec{\sigma}$ just as in the proof of
Theorem~\ref{thm:main}. An argument similar to the one used in the
synchronous case shows that $(\vec{\sigma}, b)$ is indeed a
$k$-resilient sequential equilibrium. The only major difference
between the synchronous and the asynchronous case is that, in an
asynchronous system, constructing a $k$-resilient Nash equilibrium
$\vec{\sigma}$ that implements $\vec{\sigma}' + \sigma'_d$ is far more
complicated than in a synchronous system (the reason for this can be
found below), and is beyond the scope of this work. In this paper, we
only provide a very high-level overview of the construction used in
\cite{GH19} and \cite{ADGH19} that satisfies all the properties we
need. 
As in the case of Theorem~\ref{thm:main}, we prove Theorem~\ref{thm:main2} for the case of $k$-resilient sequential equilibrium. The case of strongly $k$-resilient equilibrium is analogous.

As discussed in Section~\ref{sec:results}, in general, there
is no simple description of $SE_k(\Gamma_{d, asyn})$; thus,
the proof of Theorem~\ref{thm:main2} consists of showing that
arbitrary interactions 
$\vec{\sigma}' + \sigma'_d$
with a trusted mediator
in the asynchronous setting can be emulated by a strategy
$\vec{\sigma}$ without the mediator that tolerates deviations of
coalitions of size at most $k$. We proceed as in the proof of
Theorem~\ref{thm:main}: we first construct a $k$-resilient Nash
equilibrium $\vec{\sigma}$ that implements $\vec{\sigma}' +
\sigma'_d$, and then we construct a $k$-paranoid belief system to
extend the Nash equilibrium to a $k$-resilient sequential equilibrium
$\vec{\sigma}_{seq}$. The main problem is that the
construction of a $k$-resilient Nash equilibrium is not as simple as
in the synchronous setting. The communication between the
players and the mediator can be arbitrarily convoluted (as opposed to
the players just sending their types and receiving their suggested
action), and it may depend on the scheduler. This means that, to
be sure that all possible cases are considered, players in
$\vec{\sigma}$ must somehow be able to simulate all possible
scheduling orders for the players and the messages in $\vec{\sigma}' +
\sigma_d$. (In a given run, the particular order that they choose must
depend on the actual scheduler in $\vec{\sigma}$.) The simulation is
quite nontrivial, but the high-level idea is that each player $i$ in
$\vec{\sigma}$ simulates its counterpart in $\vec{\sigma}' +
\sigma'_d$, except that whenever it would send a message $msg$ in
$\vec{\sigma}' + \sigma'_d$, $i$ shares $msg$ using VSS. Moreover,
each player simultaneously plays its part simulating the mediator,
which means using a \emph{consensus protocol} 
(see Appendix~\ref{sec:consensus})
to guarantee that the
nonfaulty
players agree on when the
mediator is scheduled (there exist $t$-resilient consensus protocols
in asynchronous systems if $n > 3t$ \cite{ADH08})
and using CC to compute which messages it receives, which
messages it sends, and what the contents of those messages are. With
this procedure, each player is able to compute its share of every
message sent by the mediator. If that message was supposed to be sent
to player $j$, then each player $i$ sends its share of the message to
$j$ in order for $j$ to be able to reconstruct it. The full
construction can be found in \cite{GH19}. It actually satisfies a
property called \emph{$k$-bisimulation}, which states the following: 
\begin{itemize}
      \item [(a)] For all coalitions $K$ of size at most $k$, all
      strategies $\vec{\tau}_K$ for players in $K$, and all schedulers
      $\sigma_e$, there exists a strategy $\vec{\tau}'_K$ and a
      scheduler $\sigma_e'$ such that $(\vec{\sigma}_{-K},
      \tau_K)(\sigma_e)$ and 
      $((\vec{\sigma}'_{-K}, \tau'_K) + \sigma'_d)(\sigma'_e)$
      are identically distributed.
          \item [(b)] For all coalitions $K$ of size at most $k$, all
      strategies $\vec{\tau}'_K$ for players in $K$, and all
      schedulers $\sigma'_e$, there exists a strategy $\vec{\tau}_K$
      and a scheduler $\sigma_e$ such that $(\vec{\sigma}_{-K},
      \tau_K)(\sigma_e)$ and 
      $((\vec{\sigma}'_{-K}, \tau'_K) + \sigma'_d)(\sigma'_e)$
      are identically distributed. 
\end{itemize}

Intuitively, if $\vec{\sigma}$ $k$-bisimulates $\vec{\sigma}' + \sigma'_d$,
then all possible distributions over action profiles in $\vec{\sigma}$
when a coalition of at most $k$ players deviates also occur in
$\vec{\sigma}' + \sigma'_d$ when the same coalition deviates (possibly in a
different way). In particular, if a coalition $K$ of players with $|K|
\le k$ can increase their payoff in $\vec{\sigma}$ by deviating, they
could already do so in $\vec{\sigma}' + \sigma'_d$, which contradicts the
assumption that
$\vec{\sigma}' + \sigma'_d$ is a $k$-resilient sequential equilibrium. This
implies that $\vec{\sigma}$ is a $k$-resilient Nash 
equilibrium that implements $\vec{\sigma}' + \sigma'_d$ 
(a result shown by Abraham,
Dolev, Geffner, and Halpern~\nciteyear{ADGH19}).  

In the proof of Theorem~\ref{thm:main2} we use a slightly modified
variant of Geffner and Halpern's construction. In this variant, each
player $i$ terminates only if all 
messages
sent in the past are
correct
according to $\vec{\sigma}$ and $i$'s local history.
If $i$ ever 
lied
to some player $j$, $i$
sends the correct 
message
before terminating. The purpose of this is
that even if a player $j$ is not able to proceed with
$\vec{\sigma}$ (for instance, because it received incorrect messages),
$j$ would still believe that in the future it will still be able to
reconstruct the (simulated) mediator's messages as long as the
remaining players are able to 
continue.
Additionally, as
in the synchonous case, if players cannot
continue with the protocol (for instance, because they received
incorrect messages), they do nothing. 
It is easy to check that this variant still $k$-bisimulates $\vec{\sigma}' + \sigma'_d$ and, thus, $\vec{\sigma}$ is also a $k$-resilient Nash equilibrium that implements $\vec{\sigma}' + \sigma'_d$.

The rest of the proof of Theorem~\ref{thm:main2} consists of finding
an appropriate $k$-paranoid belief system consistent with
$\vec{\sigma}$. If these beliefs exist then, as in the synchronous
case, no coalition $K$ of up to $k$ players would have an incentive
to deviate from the proposed protocol since, according to their
beliefs, the remaining players would still be able to compute their
part of the action profile $\vec{a}$ (which is sampled from a
communication equilibrium), and thus it is optimal for players in $K$
to play their part of $\vec{a}$ as well. Note that, since 
the protocol states that players should
send the correct 
messages at the very end of
the protocol, with any belief system consistent with $\vec{\sigma}$,
players in $K$ always believe that they will be able to compute
their part of $\vec{a}$ if they wait long enough, regardless of their
local history. This means that proving Theorem~\ref{thm:main2} reduces
to proving the following analogue of Proposition~\ref{prop:paranoid} for
asynchronous systems: 

\begin{proposition}\label{prop:paranoid-asyn}
For all strategies $\vec{\sigma}$ in $\Gamma_{\ACT, asyn}$, there
exists a $k$-paranoid belief system $b$ that is consistent with
$\vec{\sigma}$. 
\end{proposition}

\begin{proof}
Consider a sequence $\vec{\sigma}^1, \vec{\sigma}^2, \ldots$ of
strategy profiles defined just as in the proof of
Proposition~\ref{prop:paranoid}, and consider a randomized scheduler
$\sigma_e$ such that all valid \emph{schedule prefixes} (i.e., ordered
list of players scheduled and messages delivered so far) occur with
positive probability. For instance, $\sigma_e$ could be a scheduler
that at each point in time either delivers a message or chooses a
player to be scheduled uniformly at random. Consider the belief system
$b$ induced by the previous sequence of strategies with scheduler
$\sigma_e$. An argument similar to that given in
Lemmas~\ref{lemma:aux} and~\ref{lemma:aux2} shows that $b$ is
a $k$-paranoid belief system, as desired. 
\end{proof}
}

\section{Characterization of $SE_k(\Gamma_{d})$ and $SE_k^S(\Gamma_{d})$} \label{sec:descriptions}

In this section we characterize the sets $SE_k(\Gamma_{d})$,
and $SE_k^S(\Gamma_{d})$
(note that, by Theorem~\ref{thm:main}, the characterization of $SE_k(\Gamma_{\ACT})$ and $SE_k^S(\Gamma_{\ACT})$ is identical if $n > 3$.)
If $\Gamma$ is a 
normal-form game, let $CE_k(\Gamma)$ denote the set of possible
distributions over action profiles 
induced by a $k$-resilient correlated
equilibrium
in $\Gamma$; if $\Gamma$ is a Bayesian game, let
$Com_k(\Gamma)$ denote the set of possible maps from type profiles to
distributions over action profiles induced by $k$-resilient
communication equilibria in $\Gamma$. We also denote by $CE_k^S(\Gamma)$ the set of possible distributions
over action profiles induced by a strongly $k$-resilient correlated
equilibrium in $\Gamma$. $Com_k^S(\Gamma)$ and $SE_k^S(\Gamma)$ are
defined analogously. 

The first characterization shows that, if $\Gamma$ is a normal-form game, then $SE_k(\Gamma_{d})$ is the set of $k$-resilient correlated equilibria of $\Gamma$ and $SE_k^S(\Gamma_{d})$ is the set of strongly $k$-resilient correlated equilibria of $\Gamma$.

\begin{proposition}\label{prop:desc1}
If $\Gamma = (P,A,U)$ is a normal-form game for $n$ players, then
$SE_k(\Gamma_{d}) = CE_k(\Gamma)$
and 
$SE_k^S(\Gamma_{d}) = CE_k^S(\Gamma)$
for all $k \le n$. 
\end{proposition}

\begin{proof}
  Clearly, the distribution over actions induced by a $k$-resilient
sequential equilibrium
(resp., strongly $k$-resilient sequential equilibrium)
$\vec{\sigma}$ in $\Gamma_{d}$ 
is also a $k$-resilient
correlated equilibrium 
(resp., strongly $k$-resilient correlated equilibrium)
of $\Gamma$, for otherwise there exists a coalition
$K$ of at most $k$ players 
such that all (resp., some)
members of the coalition
can increase their 
utility
by deviating
from
$\vec{\sigma}$ by playing a different action in the underlying game.
For the opposite inclusion, observe that all $k$-resilient
(resp., strongly $k$-resilient) 
correlated
equilibria $p$ 
of 
$\Gamma$ can be easily implemented with a mediator
the mediator samples an action profile
$\vec{a}$ following distribution $p$ and gives $a_i$ to each player
$i$. Then each player $i$ plays whatever is sent by the mediator.
\end{proof}

Theorem~\ref{thm:main} and Proposition~\ref{prop:desc1} together
imply the following corollary: 

\begin{corollary}
If $\Gamma = (P,A,U)$ is a normal-form game for $n$ players and $n > 3k$, then
$SE_k(\Gamma_{\ACT}) = CE_k(\Gamma)$ 
and 
$SE_k^S(\Gamma_{\ACT}) = CE_k^S(\Gamma)$. 
\end{corollary}

\commentout{
In the asynchronous setting, the description of $SE_k(\Gamma_{d,
  asyn})$ is a bit more convoluted. Reasoning similar to that used in
the proof of
  Proposition~\ref{prop:desc1} shows that for a fixed scheduler, 
  a $k$-resilient (resp., strongly $k$-resilient) sequential
  equilibrium in $SE_k(\Gamma_{d, asyn})$ induces a $k$-resilient (resp.,
strongly $k$-resilient) correlated equilibrium in $\Gamma$. Thus,
$SE_k(\Gamma_{d, asyn}) \subseteq \mathcal{P}(CE_k(\Gamma))$ and
$SE_k^S(\Gamma_{d, asyn}) \subseteq \mathcal{P}(CE_k^S(\Gamma))$,
where $\mathcal{P(S)}$ denotes the power set of $S$ (i.e., the set of
all subsets of $S$). The next proposition gives a precise description
of $SE_k(\Gamma_{d, asyn})$ and $SE_k^S(\Gamma_{d, asyn})$.

\begin{proposition}\label{prop:desc-asyn}
Given a set $S$ of strategies, let $\mathcal{P}_=(S)$ the set of 
  nonempty subsets $S'$ of $S$ such that every element of $S'$ gives the
same expected utility to all players. If $\Gamma = (P,A,U)$ is a
normal-form game for $n$ players and $n \ge k$, then 
$SE_k(\Gamma_{d, asyn}) = \mathcal{P}_=(CE_k(\Gamma))$
and 
$SE_k^S(\Gamma_{d, asyn}) = \mathcal{P}_=(CE_k^S(\Gamma))$.
\end{proposition}
\begin{proof}
In games with a mediator, we write $\sigma_d$ to denote a generic
strategy for the mediator, and $\vec{\sigma} + \sigma_d$ to denote a
strategy profile for the players and the mediator.
To show that $SE_k(\Gamma_{d, asyn}) \subseteq
\mathcal{P}_=(CE_k(\Gamma))$ and $SE_k^S(\Gamma_{d, asyn}) \subseteq
\mathcal{P}_=(CE_k^S(\Gamma))$, 
suppose by way of contradiction that some strategy $\vec{\sigma} +
\sigma_d$ in $SE_k(\Gamma_{d, asyn})$ (resp., $SE_k^S(\Gamma_{d,
  asyn})$) induces a set $S$ of strategies such that, for some player $i$,
there exist two strategies $\vec{\tau}, \vec{\tau}' \in S$ induced
by schedulers $\sigma_e$ and $\sigma_e'$, respectively, such that
$u_i(\vec{\tau}) < 
u_i(\vec{\tau}')$.  Consider a scheduler $\sigma_e''$ that does the
following: it first schedules player $i$. If $i$ sends a message to
itself when it is first scheduled, then $\sigma_e$ acts like
$\sigma'_e$, and otherwise it acts like $\sigma''_e$. By construction,
when the scheduler is $\sigma''_e$, 
$i$ gains by deviating from $\sigma_i$
and sending a message to itself when it is first scheduled.  (This
construction will not work if $\sigma_i$ already requires $i$ to send
a message to itself; in this case, we construct $\sigma''_e$
so that the signal from $i$ to the
scheduler is encoded differently, for example, by $i$ sending two
messages to itself.)  This
shows that $SE_k(\Gamma_{d, asyn}) \subseteq
\mathcal{P}_=(CE_k(\Gamma))$ and $SE_k^S(\Gamma_{d, asyn}) \subseteq
\mathcal{P}_=(CE_k^S(\Gamma))$. 

To prove the opposite inclusions, 
since we consider only equilibria with rational
probabilities, the set $CE_k(\Gamma)$ (resp., $CE_k^S(\Gamma)$) is
countable, as is any subset $S \subseteq CE_k(\Gamma)$ (resp., any
subset $S \subseteq CE_k^S(\Gamma)$). Given $S = \{\vec{\tau}^1,
\vec{\tau}^2, \ldots \}$ such that for all $k$ and players $i$ and
$j$, $\vec{\tau}^k_i = \vec{\tau}^k_j$, consider
the following strategy $\vec{\sigma} + \sigma_d$.
Each player sends an
empty message to the mediator when it is first scheduled. Let $N$
be the number of times the mediator that has been scheduled before
receiving a message from some player. The mediator samples
an action profile $\vec{a}$ from $\vec{\tau}^N$ (or from $\vec{\tau}^1$
if $S$ is finite and has fewer than $N$ elements), and sends $a_i$ to
each player $i$. Players play whatever they receive from the
mediator. It is straightforward to check that for each strategy
$\vec{\tau}$ in $S$, there is a scheduler $\sigma_e$ that induces
$\vec{\tau}$ when the players and the mediator play $\vec{\sigma} +
\sigma_d$. To check that $\vec{\sigma} + \sigma_d$ is indeed a
$k$-resilient (resp., strongly $k$-resilient) correlated equilibrium,
note that all strategies in $S$ are $k$-resilient (resp., strongly
$k$-resilient) correlated equilibria of $\Gamma$ that give the same
utility to each player, and thus (a) 
after determining the number of times $t$ that the mediator is
scheduled before receiving the first message,  
no coalition of $k$ or less
players can benefit from playing an action profile  different from the
one suggested by the mediator, and (b) 
\commentout{
regardless of 
what the
coalition does,
the expected 
payoff at the end of the protocol is constant. 
}
the expected utility of each player is independent of $t$. 
This shows that $\vec{\sigma} +
\sigma_d$ is indeed a $k$-resilient (resp., strongly
$k$-resilient) correlated equilibrium.
\end{proof}

Theorem~\ref{thm:main2} and Proposition~\ref{prop:desc-asyn} together imply:
\begin{corollary}
  If $\Gamma = (P,A,U)$ is a normal-form game for $n$ players and $n >
  4k$, then
  $SE_k(\Gamma_{\ACT, asyn}) = \mathcal{P}_=(CE_k(\Gamma))$
and 
$SE_k^S(\Gamma_{\ACT, asyn}) = \mathcal{P}_=(CE_k^S(\Gamma))$
with both the DM and AH approaches. 
\end{corollary}
}

For Bayesian games,
we can show that the set of sequential equilibria in $\Gamma_d$ is equal to the set of communication equilibria of $\Gamma$. 

\begin{proposition}\label{prop:desc2}
  If $\Gamma = (P, T, q, A,U)$ is a Bayesian game for $n$ players and
  $n \ge k$, then
  $SE_k(\Gamma_{d}) = Com_k(\Gamma)$ 
and 
$SE_k^S(\Gamma_{d}) = Com_k^S(\Gamma)$ .
\end{proposition}

\begin{proof}
Suppose that $\vec{\sigma}$ is a $k$-resilient (resp., strongly
$k$-resilient) sequential equilibrium of $\Gamma_{d}$. Then the
correlated strategy profile $\mu$ induced by $\vec{\sigma}$ in
$\Gamma$ must be a $k$-resilient (resp., strongly $k$-resilient)
communication equilibrium. If $\mu$ is not a $k$-resilient (resp.,
strongly $k$-resilient) communication equilibrium, then 
there exists a coalition $K$ of players with $|K| \le k$ and two
functions $\psi: T_K \rightarrow T_K$ and $\varphi : A_K \rightarrow
A_K$ such that the inequality of Definition~\ref{def:k-comm} does not
hold for some (resp., for all) $i \in K$. It follows that if
agents in $\Gamma_{d}$ $K$ play $\vec{\sigma}$ as if they had
type profile $\psi(\vec{t}_K)$ instead of their true types, and then
play action $\varphi(\vec{a}_K)$ instead of the action profile
$\vec{a}_K$, the utility of all (resp., some) agents would strictly
increase, which contradicts the assumption that $\vec{\sigma}$ is a
$k$-resilient 
(resp., strongly $k$-resilient) sequential equilibrium of $\Gamma_{d}$. For the opposite inclusion, recall from the discussion after
  Definition~\ref{def:k-comm} that we can identify a correlated
strategy profile $\mu$ in $\Gamma$ with a canonical strategy
$\vec{\sigma}$ in $\Gamma_{d}$ in which players tell the mediator
their type and the mediator computes which action profile they should
play according to $\mu$. If $\mu$ is a $k$-resilient (resp., strongly
$k$-resilient) communication equilibrium, then the canonical strategy
$\vec{\sigma}$ is a $k$-resilient (resp., strongly $k$-resilient)
sequential equilibrium of $\Gamma_{d}$ that induces $\mu$ (see
the discussion after Definition~\ref{def:k-comm} for details).
\end{proof}

Theorem~\ref{thm:main} and Proposition~\ref{prop:desc2} together imply
the following:

\begin{corollary}
If $\Gamma = (P, T, q, A,U)$ be a Bayesian game for $n$ players and $n
> 3k$, then 
  $SE_k(\Gamma_{\ACT}) = Com_k(\Gamma)$ 
and 
$SE_k^S(\Gamma_{\ACT}) = Com_k^S(\Gamma)$ .
\end{corollary}

\commentout{
Unfortunately, we do not believe that there is a simple
description of $SE_k(\Gamma_{d, asyn})$.
To understand why, note that a key step in the proofs of
Theorems~\ref{thm:main} and \ref{thm:main2} is for players
to send their types to the mediator.
In asynchronous systems, the
mediator cannot  
distinguish between players that don't send their type to the mediator
and players
that send it 
but the receipt of the message is delayed by the scheduler.
Nevertheless, if players can somehow punish those that never send
their type to the mediator, 
\commentout{
the same construction used in the synchronous case would suffice for the asynchronous
case as well. More precisely, 
with the AH approach, all $\mu \in
Com_k(\Gamma)$ can be implemented in $\Gamma_{d, asyn}$ if
$\Gamma$ has a $k$-\emph{punishment equilibrium} with respect to
$\mu$. 
}
we can guarantee that it is optimal for all players to send a type to the mediator, and thus
we can extend the constructions used in the proof of
Propositions~\ref{prop:desc-asyn} and \ref{prop:desc2} to
describe $SE_k(\Gamma_{d, asyn})$ and $SE_k^S(\Gamma_{d, asyn})$.  

\begin{definition}[\cite{Bp03,ADGH06}]
If $\Gamma = (P, T, q, S,U)$ is a Bayesian game for $n$ players and
$\mu \in Com_k(\Gamma)$ 
or $\mu \in Com_k^S(\Gamma)$,
then a  \emph{$k$-punishment equilibrium}
  with respect 
  to $\mu$ is a 
  $k$-resilient Bayesian Nash equilibrium 
  $\vec{\tau}$
  such that $u_i(\mu) >
  u_i(
  \vec{\tau}
  )$ for all players $i$. 
\end{definition}

Intuitively, a $k$-punishment equilibrium
with respect to $\mu$
is a strategy profile that is a 
$k$-resilient
Bayesian
Nash
equilibrium which, if played by 
at least $n-k$ 
players, results in all
players being worse off than they would be with $\mu$.
If a $k$-punishment equilibrium exists, then 
we have the following.

\begin{proposition}\label{prop:desc3}
If $\Gamma = (P, T, q, S,U)$ is a Bayesian game,
$S$ is a subset of $k$-resilient (resp., strongly $k$-resilient)
communication equilibria of $\Gamma$ such that, for all players $i$
and all types $t_i$ of $i$, the expected utility of $i$ given $t_i$ is
the same for all $\mu$ in $S$
(i.e., $u_i(\mu \mid t_i) = u_i(\mu' \mid t_i)$ for all $\mu, \mu' \in
S$), 
and there exists a $k$-punishment
equilibrium  
$\vec{\tau}$
with respect to some $\mu$ in $S$, then $S \in
SE_k(\Gamma_{d, asyn})$ (resp., $S \in SE_k^S(\Gamma_{d, asyn})$) with
the AH approach or with the default 
move approach if, for each player $i$, $i$'s default move is 
$\tau_i$.

\end{proposition}

Note that if 
$\vec{\tau}$
is a $k$-punishment equilibrium
with respect to some $\mu$ in $S$, then it is a
$k$-punishment equilibrium with respect to all $\mu \in S$, since all $\mu$
give the same utility to the players. 

\begin{proof}
This proof follows the same lines as the proof of
Proposition~\ref{prop:desc-asyn}, with the only difference that the
players must tell the mediator their type at the beginning of the
game. The punishment strategy is required to force players to tell
their true type; if players do not tell their true type, the remaining
players will play their part of the punishment strategy (note that in
Proposition~\ref{prop:desc-asyn}, a punishment strategy is not
needed since players have no types
and thus the mediator requires no inputs).  
Since we are considering only distributions with rational
probabilities, the set $S$ is countable. Let $S = \{\mu^1, \mu^2,
\ldots\}$. Consider a strategy $\vec{\sigma} +  \sigma_d$ in
which each player sends its type to the mediator when it is first 
scheduled and the mediator waits until it receives the type
of each player. Let $N$ be the number of times that the
mediator is scheduled before receiving all the types. The
mediator samples an action profile $\vec{a}$ from $\mu^N$ (or from
$\mu^1$ if $S$ is finite and has fewer than $N$ elements), and sends
$a_i$ to each player $i$. Players play whatever they receive from the
mediator. If player $i$ never receives an action, it plays $\tau_i$
(either because it is the default action or because $i$ leaves it in
its will). As in the proof of Proposition~\ref{prop:desc-asyn}, it is
easy to check that $\vec{\sigma}$ induces $S$. 
\commentout{
To see that
$\vec{\sigma}$ is a $k$-resilient (resp., strongly $k$-resilient)
sequential equilibrium, 
}

We next show that $\vec{\sigma}$ is a $k$-resilient (resp., strongly
$k$-resilient) 
sequential equilibrium for some belief system $b$.
Consider a sequence of strategy profiles $\vec{\sigma}^1, \vec{\sigma}^2, \ldots$
such that, in $\vec{\sigma}^N$, each player $i$ acts as follows. The first time $i$ is scheduled, it sends 
its type $t_i$ to the mediator with probability $1 - 1/N$ (instead of
sending its type to the mediator with probability 1, as in
$\vec{\sigma}$). Moreover, each time it is scheduled, $i$ sends a
random message to each other player $j$ with probability
$1/N$. Whenever a player $i$ receives an action $a_i$  from the
mediator, it plays $a_i$. Otherwise, it plays $\vec{\tau}(t_i)$. By
construction, the strategy profiles $\vec{\sigma}^1 + \sigma_d,
\vec{\sigma}^2 + \sigma_d, 
\ldots$ are completely mixed and converge to $\vec{\sigma}$. Let $b$
be the belief system induced by $\vec{\sigma}^1 + \sigma_d,
\vec{\sigma}^2 + \sigma_d, \ldots$. According to $b$, since all
messages received from other players are randomly generated, they can
be ignored. Thus, $i$'s utility depends only on the message received
by the mediator. If player $i$ receives a message from the mediator,
then it must be that every player sent its type and therefore that it
is optimal for $i$ to play the action received. On the other hand, if $i$ never
receives a message from the mediator, then according to its beliefs
(as characterized by $b$), it must be
that no one else got a message and thus everyone else is going to play
their part of $\vec{\tau}$. Since $\vec{\tau}$ is a $k$-resilient Nash
equilibrium, 
playing $\tau_i$ is optimal. This shows that $(\vec{\sigma}, b)$ is a $k$-resilient (resp., strongly $k$-resilient)
sequential equilibrium.
\commentout{
note that a player $i$ cannot gain by not
sending its type to the mediator, since then the 
remaining players will play 
$\vec{\tau}$
which is strictly worse for
$i$ than
the equilibrium payoff. Moreover, given that all players send an input
to the mediator, since all correlated strategies $\mu \in S$ are
$k$-resilient (resp., strongly $k$-resilient) communication
equilibria, no coalition of $k$ or fewer players can gain
by lying to the mediator or by playing an action other than the
one the mediator sends. 

It remains to show that playing 
$\vec{\tau}(t_i)$
is optimal if player $i$
never gets a message from the mediator. If a player $i$
does not
receive an action to play,
then this must be because the mediator did not receive all inputs from
the players, 
since
the mediator does not deviate from its strategy.
Thus, $i$ will believe that no player received a message from the
mediator,  
and so all players will will play their part of $\vec{\tau}$. Since 
$\vec{\tau}$ is a $k$-resilient Nash equilibrium,
playing $\tau_i$ is optimal. 
}
\end{proof}

Proposition~\ref{prop:desc3} implies the following: 

\begin{corollary}
    If $\Gamma = (P, T, q, S,U)$ is a Bayesian game 
and $\mu$ is a $k$-resilient (resp., strongly $k$-resilient)
correlated equilibrium of $\Gamma$ such that  there exists a
$k$-punishment equilibrium  $\vec{\tau}$ w.r.t. $\mu$, then $\{\mu\}
\in SE_k(\Gamma_{d, asyn})$ (resp., $\{\mu\}  \in SE_k^S(\Gamma_{d,
  asyn})$) with the AH approach or with the default 
move approach if the default move is $\vec{\tau}$.
\end{corollary}
}

\fullv{
\section{The asynchronous setting}\label{sec:sync-async}

Up to this point, as
in the game-theory literature, 
we assumed that the communication between players proceeded in
synchronous rounds. 
That is, at each round $t$,
all the players send messages to each other player (we identify not
sending a message with sending a special message $\bot$)
and these messages are received by
their intended recipients before the beginning of round $t+1$.
The assumption that we identify not sending a message with sending
$\bot$ is made without loss of generality---player $j$ can tell if
player $i$ has not sent her a message.
Moreover, it is typically assumed
when analyzing games with mediators and communication games
that it is common knowledge when the
communication phase ends and that, after it ends, all the players
simultaneously move in the underlying game.  We  call this the
\emph{synchronous} setting.

In this section, we consider an \emph{asynchronous} setting that is quite
standard in the distributed computing literature.  In the asynchronous
setting, 
there is no global notion of time, and messages may take arbitrary
amounts of time to be delivered (although we do assume that all
messages are eventually delivered).
Thus, we can no longer identify not sending a message with sending
$\bot$; if $j$ has not received a message from player $i$, $j$ is not
sure if this is because $i$ did not send $j$ a message or if $i$ sent
a message that has not yet been delivered.
For ease of exposition, we assume
that message delivery is under the control of a
\emph{scheduler}, who 
also decides when
each player gets to move, with the guarantee that all players
eventually get to move if they want to move. 
More precisely, in an asynchronous setting, the players and the
scheduler alternate turns. During a player's turn, it receives all the
messages delivered to it by the scheduler since its last turn; it may
also perform additional computation and send messages to
other players (note that these messages are not received until they
are delivered by the scheduler to their recipients). When the player
finishes all necessary operations and sends all the necessary
messages, it ends its turn, at which point it become's the scheduler's
turn.  During its turn, the scheduler
can deliver messages that have been sent but not yet delivered, or it
may schedule a player. If the scheduler schedules player $i$,
the scheduler's turn ends and $i$'s turn begins.

What internal computations a player $i$ performs, what messages $i$
sends, and  to which players $i$ sends messages during its 
turn is dictated by $i$'s strategy, which is a function from $i$'s
local history to operations and messages (i.e., a strategy tells a
player what to do as a function of its local history). Note that 
$i$'s local history is simply an ordered list containing, for each
time that $i$ was scheduled, the messages that $i$ sent and received,
and the internal computations that $i$ performed.  (We can represent
this as a game, in which case we would identify each information set
of $i$ with a local history.)  Note that $i$ does not in general know
how many times another player $j$ 
has been scheduled since $i$'s last turn.

The scheduler's behavior is also encoded in its strategy, which
describes which messages the scheduler delivers and which players are
scheduled, as a function of the scheduler's local
history. The scheduler's local history is again an ordered list,
which includes which players have been scheduled up to that point
and which messages have been delivered. Note that the scheduler does not
know the contents of the messages sent by the players, but it can
still identify the messages by, for instance, labelling them by the
order they were sent. For example, an entry in the scheduler's local
history could be ``delivered player 3's fourth message''. The only
constraint on the scheduler's strategy is that
it must eventually deliver all messages that have been sent and,
starting from any local history, it must eventually schedule all
players that have not terminated.  (We assume that ``done'' is
a signal that a player can send the scheduler that indicates that it
has terminated.)

\commentout{
In both the synchronous
and the asynchronous setting, we assume that the messages are
\emph{authenticated}, so that each recipient knows who sent the message, and
messages are never corrupted or modified.
}

If $\Gamma_{\ACT}$ is asynchronous, the payoff of the players may
depend on the strategy $\sigma_e$ of the scheduler.
This means that the definition of implementation and
of the solution concepts must take the scheduler into account.
We extend the definitions of Nash equilibrium, correlated equilibrium,
sequential equilibrium, and communication equilibrium by requiring that the
relevant inequality 
holds for all choices of $\sigma_e$. For example
a strategy profile $\vec{\sigma}$ is a Nash
equilibrium in an asynchronous setting if, for all $i \in P$, all
strategies $\tau_i$ for 
player $i$, and all schedulers $\sigma_e$, $u_i(\vec{\sigma},
\sigma_e) \ge u_i(\tau_i, \vec{\sigma}_{-i}, \sigma_e)$.
Since
the action profile played might depend on the
scheduler, 
a strategy $\vec{\sigma}$ in $\Gamma_{\ACT}$ might induce
more than one strategy in $\Gamma$ (see Example~\ref{ex:race-game}
below). 
A strategy profile $\vec{\sigma}$ for
$\Gamma_{\ACT}$ \emph{implements a set $S$ of strategy profiles in
$\Gamma$} in an asynchronous setting 
if (a)
for every scheduler $\sigma_e$ in $\Gamma_{\ACT}$,
the outcome obtained when playing $\vec{\sigma}$ with scheduler
$\sigma_e$ is the same as that obtained when playing some strategy profile
in $S$, 
and (b) for every strategy profile $\vec{\tau}$ in $S$, there exists a
scheduler $\sigma_e$ such that
the outcome obtained when playing $\vec{\sigma}$ with $\sigma_e$
is the same as that obtained when playing $\vec{\tau}$. Intuitively, 
$\vec{\sigma}$ implements $S$ if the set of
outcomes that result from playing $\vec{\sigma}$ with different
choices of scheduler are precisely those that result from playing
strategy profiles in
$S$. As in the synchronous setting, if $\Gamma$ is a
normal-form game, we take the strategy profiles in $S$ to consist of
distributions over action profiles in $\Gamma$, while if $\Gamma$ is 
a Bayesian game, the strategy profiles in $S$ are 
correlated strategy profiles.

In cheap-talk games in synchronous systems, we assumed (as is
standard in the literature) that the game tree is 
finite and at the last move on a path, the players play an action in the
underlying game; we take the utility of a 
terminal node
to be the
utility of the corresponding action profile in the normal-form game.
We cannot assume
this in asynchronous systems. Indeed, in asynchronous systems, to
allow for players who deviate from an equilibrium, we must consider
game trees with infinite paths. For example, a player $i$ may wait forever
for a message that never arrives, because the sender deviated and
never sent it, so $i$ may never play an action in the underlying game.   We
must thus define what the payoffs are on infinite paths and at the 
terminal node
of
a finite path where some players do not play an action in the
underlying game.    This amounts to treating players who have not
played an action as, in fact, having played some action.
The two main approaches for doing this are the \emph{default-move
approach}
\cite{ADGH19}, 
where if player $i$ never plays an action in the
underlying game,  its action is sampled from a default distribution, and
the \emph{Aumann and Hart} (AH) approach \cite{AH03}, where the
action played by $i$ is some function of $i$'s local history.
The AH approach essentially 
assumes
that players 
can leave a ``will'' defining
the action in the underlying game that they should play if they never
actually play an action in the underlying game while playing the
cheap-talk game.  
In the asynchronous setting
(with both  the AH and the default move approach),
we say that a strategy profile $\vec{\sigma}$ 
in $\Gamma_{\ACT}$
\emph{implements} a strategy profile
$\vec{\tau}$ 
in $\Gamma$
if
the set of strategy profiles implemented by $\vec{\sigma}$ (for
different choices of scheduler)
is
$\{\vec{\tau}\}$ (thus, playing $\vec{\sigma}$ results in the 
same
outcome no matter what the scheduler does).

The following example illustrates some of the new subtleties that asynchrony introduces.

\begin{example}\label{ex:race-game}
Let $\Gamma$ be a normal-form game for $n$ players where the set
of actions is $\{1,2,\ldots, n\}$, and let $\Gamma_{d,asyn}$ be a cheap-talk
extension of $\Gamma$ in the asynchronous setting, where there is a trusted
mediator $d$ and the AH approach is used.   Suppose that
$\vec{\sigma}$ is the strategy profile that proceeds as follows:
each player sends an empty message to the mediator when
they are first scheduled. The mediator waits until
it receives  a message, and then sends a message to all players
with the index of the player whose message was received
first. When a player $j$ receives a message from the mediator with a
number $i \in \{1,2,\ldots, n\}$, player $j$ plays action $i$. If 
player $j$ never receives a message from the mediator, then according
to its will, $j$ plays action $j$.
The player whose index appears the most often in the resulting action
profile receives a payoff of 1 (if there are ties, they all get 1),
while the remaining players get 0. 

This game can be viewed as a
race. The player whose message gets to the mediator first receives a
payoff of 1 while the rest receive 0. Note that, regardless of the
scheduler, all players play the same action in $\Gamma$.
However, the scheduler decides which action is played. Thus, the
strategy $\vec{\sigma}$ in the example implements the set of
strategies $\{(1,1 \ldots, 1), (2,2,\ldots, 2), \ldots, (n, n, \ldots,
n)\}$. This strategy
is not a Nash equilibrium. To see this,
consider the scheduler $\sigma_e$ that schedules the players
sequentially (first player 1, then player 2, and so on).  
If player $i$ sends two messages to the
mediator and is the only player to do so, then the scheduler delivers
$i$'s first message before it delivers 
any other player's message. If more than one player sends two messages to the
mediator, the schedulers chooses one of these players at
random and delivers her message first.  The remaining messages 
are delivered in some random order.
With this scheduler, the players benefit by deviating and
sending two messages to the mediator instead of just
one. Thus, $\vec{\sigma}$ is not a Nash equilibrium. A similar
argument can be used to show that, for all games $\Gamma$, if
$\vec{\sigma}$ is a $k$-resilient Nash equilibrium in $\Gamma_{\ACT}$,
then the payoffs of the players when playing $\vec{\sigma}$ cannot
depend on the scheduler (see~\cite{ADGH19}). 
\end{example}

This example shows that in asynchronous systems, the set of possible
deviations is much larger than in synchronous
systems. However, there are cases where controlling the
scheduling of the messages and players does not give that much power
to an adversary. For instance, if we consider an asynchronous
cheap-talk extension of Example~\ref{ex:comm-games}, the strategy
profile in which both agents send their type to each other the first
time they are scheduled and play action $t_1t_2$ is a Nash
equilibrium (note that this is almost equivalent to the strategy
proposed in Example~\ref{ex:comm-games}, however there is no notion of
``rounds'' in an asynchronous setting). Many of the
results that hold in the synchronous case also hold in the
asynchronous case, but may require a larger proportion of the players
to be honest. For instance, asynchronous multiparty secure
computation tolerates deviations by at most a quarter of the agents,
while synchronous multiparty secure computation tolerates deviations
by at most a third \cite{BGW88,BCG93}. Our results hold for similar
bounds.  (An
intuitive explanation of why there is a difference between the
thresholds for the synchronous and asynchronous settings can be
found in Appendix~\ref{sec:tools2}.)

\subsection{Main results}\label{sec:results-asyn}

Given a normal-form or Bayesian game $\Gamma$, let $\Gamma_{\ACT, asyn}$ and $\Gamma_{d, asyn}$ be the asynchronous equivalents of $\Gamma_{\ACT}$ and $\Gamma_d$ respectively (i.e., the asynchronous cheap-talk and mediator extensions of $\Gamma$). As in Section~\ref{sec:results}, let $SE_k(\Gamma_{\ACT, asyn})$ (resp., $SE_k^S(\Gamma_{\ACT, asyn})$) denote the set of possible strategy profiles in $\Gamma$ induced by $k$-resilient sequential equilibria in $\Gamma_{\ACT, asyn}$. However, since in the asynchronous setting the outcome may depend on the 
scheduler
(see the discussion in
Section~\ref{sec:sync-async}), if $\Gamma$ is a normal-form game, then
$SE_k(\Gamma_{\ACT, asyn})$ is a set of sets of strategy profiles, and if
$\Gamma$ is a Bayesian game, then $SE_k(\Gamma_{\ACT, asyn})$ is a set
of sets of correlated strategy profiles. 

In the asynchronous setting we have equivalent results to those in the synchronous setting.
\begin{theorem}\label{thm:main2}
  If $\Gamma = (P,T, q, A,U)$ is a Bayesian game for $n$ players and
  $n > 4k$, then
  $SE_k(\Gamma_{\ACT, asyn}) = SE_k(\Gamma_{d, asyn})$
and 
$SE_k^S(\Gamma_{\ACT, asyn}) = SE_k^S(\Gamma_{d, asyn})$
with both the default move and AH approaches. 
\end{theorem}

Note that the only difference with respect to Theorem~\ref{thm:main}, is that in the asynchronous setting we require that $n > 4k$ instead of $n > 3k$. This difference comes down to the fact that, in the asynchronous setting, players cannot tell if a message has not been received because the player didn't send it or because the message is being delayed by the scheduler. In fact, this is the reason why most of the distributed computing primitives in asynchronous systems are resilient to deviations of up to a fourth of the players, as opposed to a third in the synchronous case (see Appendix~\ref{sec:tools2} for more details).

As in the synchronous case, $SE_k(\Gamma_{\ACT, asyn})$ and $SE_k^S(\Gamma_{\ACT, asyn})$ have a relatively simple characterization. This is described in Section~\ref{sec:desc-asyn}.

\subsection{Proof of Theorem~\ref{thm:main2} (outline)}

The proof of Theorem~\ref{thm:main2} is quite similar to the proof of
Theorem~\ref{thm:main}. Given a strategy profile $\vec{\sigma}' +
\sigma'_d \in SE_k(\Gamma_{d, asyn})$,
we first construct a $k$-resilient Nash equilibrium
$\vec{\sigma}$ that implements $\vec{\sigma}' + \sigma'_d$ and
guarantees that, if less than $k$ players deviate, all players will 
eventually be able to compute the action that they should play,
and then we construct a $k$-paranoid belief system $b$ that is
consistent with $\vec{\sigma}$ just as in the proof of
Theorem~\ref{thm:main}. An argument similar to the one used in the
synchronous case shows that $(\vec{\sigma}, b)$ is indeed a
$k$-resilient sequential equilibrium. The only major difference
between the synchronous and the asynchronous case is that, in an
asynchronous system, constructing a $k$-resilient Nash equilibrium
$\vec{\sigma}$ that implements $\vec{\sigma}' + \sigma'_d$ is far more
complicated than in a synchronous system (the reason for this can be
found below), and is beyond the scope of this work. In this paper, we
only provide a very high-level overview of the construction used in
\cite{GH19} and \cite{ADGH19} that satisfies all the properties we
need. 
As in the case of Theorem~\ref{thm:main}, we prove Theorem~\ref{thm:main2} for the case of $k$-resilient sequential equilibrium. The case of strongly $k$-resilient equilibrium is analogous.

As discussed in Section~\ref{sec:results}, in general, there
is no simple description of $SE_k(\Gamma_{d, asyn})$; thus,
the proof of Theorem~\ref{thm:main2} consists of showing that
arbitrary interactions 
$\vec{\sigma}' + \sigma'_d$
with a trusted mediator
in the asynchronous setting can be emulated by a strategy
$\vec{\sigma}$ without the mediator that tolerates deviations of
coalitions of size at most $k$. We proceed as in the proof of
Theorem~\ref{thm:main}: we first construct a $k$-resilient Nash
equilibrium $\vec{\sigma}$ that implements $\vec{\sigma}' +
\sigma'_d$, and then we construct a $k$-paranoid belief system to
extend the Nash equilibrium to a $k$-resilient sequential equilibrium
$\vec{\sigma}_{seq}$. The main problem is that the
construction of a $k$-resilient Nash equilibrium is not as simple as
in the synchronous setting. The communication between the
players and the mediator can be arbitrarily convoluted (as opposed to
the players just sending their types and receiving their suggested
action), and it may depend on the scheduler. This means that, to
be sure that all possible cases are considered, players in
$\vec{\sigma}$ must somehow be able to simulate all possible
scheduling orders for the players and the messages in $\vec{\sigma}' +
\sigma_d$. (In a given run, the particular order that they choose must
depend on the actual scheduler in $\vec{\sigma}$.) The simulation is
quite nontrivial, but the high-level idea is that each player $i$ in
$\vec{\sigma}$ simulates its counterpart in $\vec{\sigma}' +
\sigma'_d$, except that whenever it would send a message $msg$ in
$\vec{\sigma}' + \sigma'_d$, $i$ shares $msg$ using VSS. Moreover,
each player simultaneously plays its part simulating the mediator,
which means using a \emph{consensus protocol} 
(see Appendix~\ref{sec:consensus})
to guarantee that the
nonfaulty
players agree on when the
mediator is scheduled (there exist $t$-resilient consensus protocols
in asynchronous systems if $n > 3t$ \cite{ADH08})
and using CC to compute which messages it receives, which
messages it sends, and what the contents of those messages are. With
this procedure, each player is able to compute its share of every
message sent by the mediator. If that message was supposed to be sent
to player $j$, then each player $i$ sends its share of the message to
$j$ in order for $j$ to be able to reconstruct it. The full
construction, as well as more intuition, can be found in
\cite{GH19}. It satisfies a 
property called \emph{$k$-bisimulation}, where $\vec{\sigma}$
$k$-bisimulates $\vec{\sigma}'$ if
\begin{itemize}
      \item [(a)] For all coalitions $K$ of size at most $k$, all
      strategies $\vec{\tau}_K$ for players in $K$, and all schedulers
      $\sigma_e$, there exists a strategy $\vec{\tau}'_K$ and a
      scheduler $\sigma_e'$ such that 
      $(\vec{\sigma}_{-K},\tau_K,\sigma_e)$ and 
      $((\vec{\sigma}'_{-K}, \tau'_K) + \sigma'_d,\sigma'_e)$
      are identically distributed.
          \item [(b)] For all coalitions $K$ of size at most $k$, all
      strategies $\vec{\tau}'_K$ for players in $K$, and all
      schedulers $\sigma'_e$, there exists a strategy $\vec{\tau}_K$
      and a scheduler $\sigma_e$ such that 
      $(\vec{\sigma}_{-K}, \tau_K,\sigma_e)$ and 
      $((\vec{\sigma}'_{-K}, \tau'_K) + \sigma'_d,\sigma'_e)$
      are identically distributed. 
\end{itemize}

Intuitively, if $\vec{\sigma}$ $k$-bisimulates $\vec{\sigma}' + \sigma'_d$,
then all possible distributions over action profiles in $\vec{\sigma}$
when a coalition of at most $k$ players deviates also occur in
$\vec{\sigma}' + \sigma'_d$ when the same coalition deviates (possibly in a
different way). In particular, if a coalition $K$ of players with $|K|
\le k$ can increase their payoff in $\vec{\sigma}$ by deviating, they
could already do so in $\vec{\sigma}' + \sigma'_d$, which contradicts the
assumption that
$\vec{\sigma}' + \sigma'_d$ is a $k$-resilient sequential equilibrium. This
implies that $\vec{\sigma}$ is a $k$-resilient Nash 
equilibrium that implements $\vec{\sigma}' + \sigma'_d$ 
(a result shown by Abraham,
Dolev, Geffner, and Halpern~\nciteyear{ADGH19}).  

In the proof of Theorem~\ref{thm:main2} we use a slightly modified
variant of Geffner and Halpern's construction. In this variant, each
player $i$ terminates only if all 
messages
sent in the past are
correct
according to $\vec{\sigma}$ and $i$'s local history.
If $i$ ever 
lied
to some player $j$, $i$
sends the correct 
message
before terminating. The purpose of this is
that even if a player $j$ is not able to proceed with
$\vec{\sigma}$ (for instance, because it received incorrect messages),
$j$ would still believe that in the future it will still be able to
reconstruct the (simulated) mediator's messages as long as the
remaining players are able to 
continue.
Additionally, as
in the synchonous case, if players cannot
continue with the protocol (for instance, because they received
incorrect messages), they do nothing. 
It is easy to check that this variant still $k$-bisimulates $\vec{\sigma}' + \sigma'_d$ and, thus, $\vec{\sigma}$ is also a $k$-resilient Nash equilibrium that implements $\vec{\sigma}' + \sigma'_d$.

The rest of the proof of Theorem~\ref{thm:main2} consists of finding
an appropriate $k$-paranoid belief system consistent with
$\vec{\sigma}$. If these beliefs exist then, as in the synchronous
case, no coalition $K$ of up to $k$ players would have an incentive
to deviate from the proposed protocol since, according to their
beliefs, the remaining players would still be able to compute their
part of the action profile $\vec{a}$ (which is sampled from a
communication equilibrium), and thus it is optimal for players in $K$
to play their part of $\vec{a}$ as well. Note that, since 
the protocol states that players should
send the correct 
messages at the very end of
the protocol, with any belief system consistent with $\vec{\sigma}$,
players in $K$ always believe that they will be able to compute
their part of $\vec{a}$ if they wait long enough, regardless of their
local history. This means that proving Theorem~\ref{thm:main2} reduces
to proving an analogue of Proposition~\ref{prop:paranoid} for
asynchronous systems: 

\begin{proposition}\label{prop:paranoid-asyn}
For all strategies $\vec{\sigma}$ in $\Gamma_{\ACT, asyn}$, there
exists a $k$-paranoid belief system $b$ that is consistent with
$\vec{\sigma}$. 
\end{proposition}

\begin{proof}
Consider a sequence $\vec{\sigma}^1, \vec{\sigma}^2, \ldots$ of
strategy profiles defined just as in the proof of
Proposition~\ref{prop:paranoid}, and consider a randomized scheduler
$\sigma_e$ such that all valid \emph{schedule prefixes} (i.e., ordered
list of players scheduled and messages delivered so far) occur with
positive probability. For instance, $\sigma_e$ could be a scheduler
that at each point in time either delivers a message or chooses a
player to be scheduled uniformly at random. Consider the belief system
$b$ induced by the previous sequence of strategies with scheduler
$\sigma_e$. An argument similar to that given in
Lemmas~\ref{lemma:aux} and ~\ref{lemma:aux2} shows that $b$ is
a $k$-paranoid belief system, as desired. 
\end{proof}

\subsection{Characterization of $SE_k(\Gamma_{d, asyn})$ and $SE_k^S(\Gamma_{d, asyn})$}\label{sec:desc-asyn}

In the asynchronous setting, the description of $SE_k(\Gamma_{d,
  asyn})$ is a bit more convoluted than those of
Section~\ref{sec:descriptions}. Reasoning similar to that used in 
the proof of
  Proposition~\ref{prop:desc1} shows that for a fixed scheduler, 
  a $k$-resilient (resp., strongly $k$-resilient) sequential
  equilibrium in $SE_k(\Gamma_{d, asyn})$ induces a $k$-resilient (resp.,
strongly $k$-resilient) correlated equilibrium in $\Gamma$. Thus,
$SE_k(\Gamma_{d, asyn}) \subseteq \mathcal{P}(CE_k(\Gamma))$ and
$SE_k^S(\Gamma_{d, asyn}) \subseteq \mathcal{P}(CE_k^S(\Gamma))$,
where $\mathcal{P(S)}$ denotes the power set of $S$ (i.e., the set of
all subsets of $S$). The next proposition gives a precise description
of $SE_k(\Gamma_{d, asyn})$ and $SE_k^S(\Gamma_{d, asyn})$.

\begin{proposition}\label{prop:desc-asyn}
Given a set $S$ of strategies, let $\mathcal{P}_=(S)$
be
the set of 
  nonempty subsets $S'$ of $S$ such that every element of $S'$ gives the
same expected utility to all players. If $\Gamma = (P,A,U)$ is a
normal-form game for $n$ players and $n \ge k$, then 
$SE_k(\Gamma_{d, asyn}) = \mathcal{P}_=(CE_k(\Gamma))$
and 
$SE_k^S(\Gamma_{d, asyn}) = \mathcal{P}_=(CE_k^S(\Gamma))$.
\end{proposition}
\begin{proof}
In games with a mediator, we write $\sigma_d$ to denote a generic
strategy for the mediator, and $\vec{\sigma} + \sigma_d$ to denote a
strategy profile for the players and the mediator.
To show that $SE_k(\Gamma_{d, asyn}) \subseteq
\mathcal{P}_=(CE_k(\Gamma))$ and $SE_k^S(\Gamma_{d, asyn}) \subseteq
\mathcal{P}_=(CE_k^S(\Gamma))$, 
suppose by way of contradiction that some strategy $\vec{\sigma} +
\sigma_d$ in $SE_k(\Gamma_{d, asyn})$ (resp., $SE_k^S(\Gamma_{d,
  asyn})$) induces a set $S$ of strategies such that, for some player $i$,
there exist two strategies $\vec{\tau}, \vec{\tau}' \in S$ induced
by schedulers $\sigma_e$ and $\sigma_e'$, respectively, such that
$u_i(\vec{\tau}) < 
u_i(\vec{\tau}')$.  Consider a scheduler $\sigma_e''$ that does the
following: it first schedules player $i$. If $i$ sends a message to
itself when it is first scheduled, then $\sigma_e$ acts like
$\sigma'_e$, and otherwise it acts like 
$\sigma_e$.
By construction,
when the scheduler is $\sigma''_e$, 
$i$ gains by deviating from $\sigma_i$
and sending a message to itself when it is first scheduled.  (This
construction will not work if $\sigma_i$ already requires $i$ to send
a message to itself; in this case, we construct $\sigma''_e$
so that the signal from $i$ to the
scheduler is encoded differently, for example, by $i$ sending two
messages to itself.)  This
shows that $SE_k(\Gamma_{d, asyn}) \subseteq
\mathcal{P}_=(CE_k(\Gamma))$ and $SE_k^S(\Gamma_{d, asyn}) \subseteq
\mathcal{P}_=(CE_k^S(\Gamma))$. 

To prove the opposite inclusions, 
since we consider only equilibria with rational
probabilities, the set $CE_k(\Gamma)$ (resp., $CE_k^S(\Gamma)$) is
countable, as is any subset $S \subseteq CE_k(\Gamma)$ (resp., any
subset $S \subseteq CE_k^S(\Gamma)$). Given $S = \{\vec{\tau}^1,
\vec{\tau}^2, \ldots \}$ such that for all 
indices
$k, k'$
and
all
players $i$,
$u_i(\vec{\tau}^k) = u_i(\vec{\tau}^{k'})$,
consider
the following strategy $\vec{\sigma} + \sigma_d$.
Each player sends an
empty message to the mediator when it is first scheduled. Let $N$
be the number of times 
that the mediator
has been scheduled before
receiving a message from some player. The mediator samples
an action profile $\vec{a}$ from $\vec{\tau}^N$ (or from $\vec{\tau}^1$
if $S$ is finite and has fewer than $N$ elements), and sends $a_i$ to
each player $i$. Players play whatever they receive from the
mediator. It is straightforward to check that, for each strategy
$\vec{\tau}$ in $S$, there is a scheduler $\sigma_e$ that induces
$\vec{\tau}$ when the players and the mediator play $\vec{\sigma} +
\sigma_d$. To check that $\vec{\sigma} + \sigma_d$ is indeed a
$k$-resilient (resp., strongly $k$-resilient) correlated equilibrium,
note that all strategies in $S$ are $k$-resilient (resp., strongly
$k$-resilient) correlated equilibria of $\Gamma$ that give the same
utility to each player, and thus (a) 
after determining the number of times $t$ that the mediator is
scheduled before receiving the first message,  
no coalition of $k$ or less
players can benefit from playing an action profile  different from the
one suggested by the mediator, and (b) 
the expected utility of each player is independent of $t$. 
This shows that $\vec{\sigma} +
\sigma_d$ is indeed a $k$-resilient (resp., strongly
$k$-resilient) correlated equilibrium.
\end{proof}

Theorem~\ref{thm:main2} and Proposition~\ref{prop:desc-asyn} together imply:
\begin{corollary}
  If $\Gamma = (P,A,U)$ is a normal-form game for $n$ players and $n >
  4k$, then
  $SE_k(\Gamma_{\ACT, asyn}) = \mathcal{P}_=(CE_k(\Gamma))$
and 
$SE_k^S(\Gamma_{\ACT, asyn}) = \mathcal{P}_=(CE_k^S(\Gamma))$
with both the DM and AH approaches. 
\end{corollary}

For Bayesian games, we do not believe that there is a simple
description of $SE_k(\Gamma_{d, asyn})$.
To understand why, note that a key step in the proofs of
Theorems~\ref{thm:main} and \ref{thm:main2} is for players
to send their types to the mediator.
In asynchronous systems, the
mediator cannot  
distinguish between players that don't send their type to the mediator
and players
that send it 
but the receipt of the message is delayed by the scheduler.
Nevertheless, if players can somehow punish those that never send
their type to the mediator, 
we can guarantee that it is optimal for all players to send a type to the mediator, and thus
we can extend the constructions used in the proof of
Propositions~\ref{prop:desc-asyn} and \ref{prop:desc2} to
describe $SE_k(\Gamma_{d, asyn})$ and $SE_k^S(\Gamma_{d, asyn})$.  

\begin{definition}[\cite{Bp03,ADGH06}]
If $\Gamma = (P, T, q, S,U)$ is a Bayesian game for $n$ players and
$\mu \in Com_k(\Gamma)$ 
or $\mu \in Com_k^S(\Gamma)$,
then a  \emph{$k$-punishment equilibrium}
  with respect 
  to $\mu$ is a 
  $k$-resilient Bayesian Nash equilibrium 
  $\vec{\tau}$
  such that $u_i(\mu) >
  u_i(
  \vec{\tau}
  )$ for all players $i$. 
\end{definition}

Intuitively, a $k$-punishment equilibrium
with respect to $\mu$
is a strategy profile that is a 
$k$-resilient
Bayesian
Nash
equilibrium which, if played by 
at least $n-k$ 
players, results in all
players being worse off than they would be with $\mu$.
If a $k$-punishment equilibrium exists, then 
we have the following.

\begin{proposition}\label{prop:desc3}
If $\Gamma = (P, T, q, S,U)$ is a Bayesian game,
$S$ is a subset of $k$-resilient (resp., strongly $k$-resilient)
communication equilibria of $\Gamma$ such that, for all players $i$
and all types $t_i$ of $i$, the expected utility of $i$ given $t_i$ is
the same for all $\mu$ in $S$
(i.e., $u_i(\mu \mid t_i) = u_i(\mu' \mid t_i)$ for all $\mu, \mu' \in
S$), 
and there exists a $k$-punishment
equilibrium  
$\vec{\tau}$
with respect to some $\mu$ in $S$, then $S \in
SE_k(\Gamma_{d, asyn})$ (resp., $S \in SE_k^S(\Gamma_{d, asyn})$) with
the AH approach or with the default 
move approach if, for each player $i$, $i$'s default move is 
$\tau_i$.

\end{proposition}

Note that if 
$\vec{\tau}$
is a $k$-punishment equilibrium
with respect to some $\mu$ in $S$, then it is a
$k$-punishment equilibrium with respect to all $\mu \in S$, since all $\mu$
give the same utility to the players. 

\begin{proof}
This proof follows the same lines as the proof of
Proposition~\ref{prop:desc-asyn}, with the only difference that the
players must tell the mediator their type at the beginning of the
game. The punishment strategy is required to force players to tell
their true type; if players do not tell 
report a type,
the remaining
players will play their part of the punishment strategy (note that in
Proposition~\ref{prop:desc-asyn}, a punishment strategy is not
needed since players have no types
and thus the mediator requires no inputs).  
Moreover, this type must be their true type since all strategies in $S$ are $k$-resilient communication equilibria of $\Gamma$.
Since we are considering only distributions with rational
probabilities, the set $S$ is countable. Let $S = \{\mu^1, \mu^2,
\ldots\}$. Consider a strategy $\vec{\sigma} +  \sigma_d$ in
which each player sends its type to the mediator when it is first 
scheduled and the mediator waits until it receives the type
of each player. Let $N$ be the number of times that the
mediator is scheduled before receiving all the types. The
mediator samples an action profile $\vec{a}$ from $\mu^N$ (or from
$\mu^1$ if $S$ is finite and has fewer than $N$ elements), and sends
$a_i$ to each player $i$. 
If a player sent its true type to the mediator, it plays whatever it receives from the mediator. Otherwise, it plays the optimal action conditioned on what it sent and what it receives from the mediator.
If player $i$ never receives an action, it plays $\tau_i$
(either because it is the default action or because $i$ leaves it in
its will). As in the proof of Proposition~\ref{prop:desc-asyn}, it is
easy to check that $\vec{\sigma}$ induces $S$. 

We next show that $\vec{\sigma}$ is a $k$-resilient (resp., strongly
$k$-resilient) 
sequential equilibrium for some belief system $b$.
Consider a sequence of strategy profiles $\vec{\sigma}^1, \vec{\sigma}^2, \ldots$
such that, in $\vec{\sigma}^N$, each player $i$ acts as follows. The first time $i$ is scheduled, it sends 
its type $t_i$ to the mediator with probability $1 - 1/N$ (instead of
sending its type to the mediator with probability 1, as in
$\vec{\sigma}$). Moreover, each time it is scheduled, $i$ sends a
random message to each other player $j$ with probability
$1/N$. Whenever a player $i$ receives an action $a_i$  from the
mediator, it plays $a_i$. Otherwise, it plays $\vec{\tau}(t_i)$. By
construction, the strategy profiles $\vec{\sigma}^1 + \sigma_d,
\vec{\sigma}^2 + \sigma_d, 
\ldots$ are completely mixed and converge to $\vec{\sigma}$. Let $b$
be the belief system induced by $\vec{\sigma}^1 + \sigma_d,
\vec{\sigma}^2 + \sigma_d, \ldots$. According to $b$, since all
messages received from other players are randomly generated, they can
be ignored. Thus, $i$'s utility depends only on the message received
by the mediator. If player $i$ receives a message from the mediator,
then it must be that every player sent its type and therefore,
if $i$ sent its true type,
it
is optimal for $i$ to play the action received. 
If $i$ did not report its true type, by construction $i$ also plays its optimal action.
On the other hand, if $i$ never
receives a message from the mediator, then according to its beliefs
(as characterized by $b$), it must be
that no one else got a message and thus everyone else is going to play
their part of $\vec{\tau}$. Since $\vec{\tau}$ is a $k$-resilient Nash
equilibrium, 
playing $\tau_i$ is optimal. This shows that $(\vec{\sigma}, b)$ is a $k$-resilient (resp., strongly $k$-resilient)
sequential equilibrium.
\end{proof}

Proposition~\ref{prop:desc3} implies the following: 

\begin{corollary}
    If $\Gamma = (P, T, q, S,U)$ is a Bayesian game 
and $\mu$ is a $k$-resilient (resp., strongly $k$-resilient)
correlated equilibrium of $\Gamma$ such that  there exists a
$k$-punishment equilibrium  $\vec{\tau}$ w.r.t. $\mu$, then $\{\mu\}
\in SE_k(\Gamma_{d, asyn})$ (resp., $\{\mu\}  \in SE_k^S(\Gamma_{d,
  asyn})$) with the AH approach or with the default 
move approach if the default move is $\vec{\tau}$.
\end{corollary}
}

\section{Extending the set of equilibria}\label{sec:extension}

Up to now we have considered only distributions over action profiles
that involve rational probabilities.    We can extend our result to
some extent to distributions described by real numbers.  Specifically,
suppose that players can send messages and operate with real
numbers, and
an equilibrium $e = \sum_i \lambda_i e_i$ is a convex combination of
countable
rational equilibria $e_i$. If  players could
jointly sample a real number $r \in [0,1)$ uniformly at random,
they could easily implement $e$: after sampling $r \in [0,1]$
uniformly at random, they play their part of  
the rational equilibrium $e_m$ such that $\sum_{i = 1}^{m-1}
\lambda_i \le r < \sum_{i = 1}^{m} \lambda_i$.
(Our earlier results show that $e_m$ can be 
implemented.)
Note that this procedure guarantees that players 
implement $e_m$ with probability $\lambda_m$, as desired.

The following is  a $k$-resilient
strategy $\vec{\tau}$ that allows players to jointly sample $r$ in the
synchronous setting if $n > 3k$: 
\commentout{
In Section~\ref{sec:results} we discussed how players could implement
some strategies with non-rational probabilities given that they could
send real numbers in their messages, and that this reduced to the
problem of jointly sampling a real number $r \in [0,1)$ uniformly at
  random. In this section, we present a $k$-resilient strategy profile
  $\vec{\tau}$ that computes $r$ if $n > 3k$.  
In the synchronous setting, the description of $\vec{\tau}$ is as follows:
}
\begin{enumerate}
    \item Each player privately generates a random number $r_i \in [0, 1)$ and broadcasts it using Bracha's 
    $k$-resilient 
            broadcast protocol~\cite{Bracha87}. If $i$ doesn't broadcast a
    number, $r_i$ is assumed to be $0$. 
      \item Players take $r:= frac(r_1 + \ldots + r_n)$, where $frac(x)$
      is the fractional part of $x$. 
\end{enumerate}

\fullv{
In the asynchronous setting, this process is harder because players
cannot distinguish players that didn't broadcast a value from players
that are being delayed by the scheduler. Moreover, players cannot
agree on using only a subset $C$ of the values shared since players
might want to influence the computation of $C$ in such a way that the
number $r$ chosen is beneficial for them. However, if $n > 4k$, each
player can generate $r_i$ as in the synchronous setting, and then
compute $r$ using a $k$-resilient secure computation~\cite{BCG93}. 
}

\section{Conclusion}
We have shown that, for all Bayesian games $\Gamma$ for $n$ players,
all $k$-resilient sequential equilibria in $\Gamma_d$ can also be
implemented in $\Gamma_{\ACT}$ if $n > 3k$ in the synchronous setting
or $n > 4k$ in the asynchronous setting. These results are optimal
since it follows from \cite{ADH07} and \cite{GH21} that if $n \le 3k$
in the synchronous case or $n \le 4k$ in the asynchronous case, there
are $k$-resilient sequential equilibria in $\Gamma_d$ that cannot be
implemented even with a $k$-resilient Nash equilibrium. 

Our results also allow us to characterize the sets $SE_k(\Gamma_{\ACT})$
of all possible strategies in $\Gamma$ induced by $k$-resilient
sequential equilibria of $\Gamma_{ACT}$, generalizing
results of Gerardi \nciteyear{Gerardi04}. If players have perfect information
(i.e. $\Gamma$ is a normal-form game), then  $SE_k(\Gamma_{\ACT})$ is
the set of $k$-resilient correlated equilibria of $\Gamma$; if
$\Gamma$ is a Bayesian game, then $SE_k(\Gamma_{\ACT})$ is
the set of $k$-resilient communication equilibria. In the asynchronous
setting, we show that $SE_k(\Gamma_{\ACT}) = SE_k(\Gamma_{d})$, but
characterizing the set $SE_k(\Gamma_{d})$ of $k$-resilient sequential
equilibria with a trusted mediator is still an open problem. However, if $\Gamma$ has a $k$-punishment equilibrium, 
then $SE_k(\Gamma_{d})$ is also the set of $k$-resilient communication equilibria of $\Gamma$.

If $\Gamma$ has a $k$-punishment equilibrium, Abraham, Dolev, Gonen,
and Halpern~\nciteyear{ADGH06} and Abraham, Dolev, Geffner, and
Halpern~\nciteyear{ADGH19} showed that
all $k$-resilient Nash
equilibria in $\Gamma_d$ could be implemented in $\Gamma_{\ACT}$ if $n
> 2k$ in the synchronous setting or $n > 3k$ in the asynchronous
setting respectively. The question of whether these results can be
extended to $k$-resilient sequential equilibria is still open.

\shortv{
All of the results presented in this paper assume a \emph{synchronous}
setting: communication 
proceeds in atomic rounds, and all messages sent during round $r$
are received by round $r+1$.  In an \emph{asynchronous} setting,
there are no rounds and messages sent by the players may take
arbitrarily long to get to their 
recipients. 
Asynchrony is a standard assumption in 
the distributed computing
and cryptography
literature, precisely because many systems
that practitioners care about,
such as markets or the internet,
are asynchronous in practice.
Considering asynchronous systems can have significant implications for how
players will play a game.
For instance, in an online second-price auction, the seller can
benefit from inserting fake transactions whose value is between that
of the
highest and second-highest bid immediately after a new highest bid
is received, thus increasing the second-highest price at no 
cost~\cite{Roughgarden2020}.
This type of
attack can be carried out only in asynchronous or \emph{partially synchronous}
systems (where there is an upper bound on how long messages take to
arrive), since in a synchronous system, all bids are received at the
same time. (In a synchronous system, the seller would have to guess
what the players will bid in order to benefit from a fake transaction.) 
\commentout{
For instance, all blockchain implementations assume \emph{partial
    synchrony}, where there is an upper 
bound on how long messages take to arrive.
}

The simplicity of our approach allows us to generalize our main result
to the asynchronous setting with very little work. 
Given a $k$-resilient sequential equilibria $\vec{\sigma}$ in an
asynchronous game with $n$ players and a mediator, we extend 
implementation given by Abraham et al. \nciteyear{ADGH19} of asynchronous
$k$-resilient Nash equilibria in $\Gamma_{\ACT}$ for $n > 4k$
by constructing a consistent $k$-paranoid belief
system. Reasoning analogous to that used in the synchronous case
shows that the resulting strategy and belief system form a
$k$-resilient sequential equilibria.
Details can be found in \cite{GH23}.
In the asynchronous setting, this
result is also optimal, since it matches the $n > 4k$ lower
bound given by Geffner and Halpern~\nciteyear{GH21}. 

}


\newpage
\appendix

\section{Game theory: basic definitions}\label{sec:definitions}

In this section, we review the basic concepts and definitions that
are needed for this paper.
\commentout{
We begin by reviewing basic
game-theoretic definitions, and then review some distributed computing 
primitives that are used  in our protocols.
}
We begin by presenting
several types of games that are common in
the literature, 
each of them with its own settings. For each of these games, we also
introduce some of 
the \emph{solution concepts} that we are interested in. A solution
concept describes what rational agents should do in a game, according
to some theory  of rationality.  Typically, a solution concept
determines for each game a set of \emph{strategy profiles},  
where a strategy profile $\vec{s} = (s_1, \ldots, s_n)$ is a tuple
consisting of one strategy for each agent, and a strategy is 
a description of how each player should
move/act in all ``situations'' in the given game.  (As we shall see,
what counts as a situation depends 
on
the type of game we consider.) 

\subsection{Nash, Correlated, and Sequential Equilibrium}

A normal-form game $\Gamma$ is a tuple $(P, A, U)$, where
$P =
\{1, \ldots, n\}$ is the set of \emph{players}, $A = A_1 \times \cdots 
A_n$, where  $A_i$ is the set of possible 
actions
for
player $i \in P$, and $U = (u,
\ldots, u_n)$ is an $n$-tuple of \emph{utility}
functions $u_i: A \rightarrow \mathbb{R}$, again, one for each player
$i \in P$.  
A
\emph{pure strategy
  profile} 
  $\vec{a}$
  is
  an $n$-tuple of 
  actions
  $(a_1, \ldots,
  a_n)$, with $a_i \in A_i$.
  A \emph{mixed strategy} for player $i$ is
an element of $\Delta(A_i)$, the set of probability distributions
on $A_i$.  We extend $u_i$ to mixed strategy profiles
$\sigma = (\sigma_1, \ldots,
\sigma_n)$ by defining
$u_i(\sigma) = 
\sum_{(a_1,\ldots,a_n) \in A} \sigma_1(a_1)\ldots\sigma_n(a_n) u_i(a_1,
\ldots, a_n)$: that is, the sum over all pure strategy profiles $\vec{a}$ of the
probability of playing $\vec{a}$ according to
$\sigma$ times
the utility 
of $\vec{a}$ to $i$.

\begin{definition}\label{def:nash-eq}
In a normal-form game $\Gamma$, a (mixed) strategy profile
$\vec{\sigma} = (\sigma_1, \ldots, \sigma_n)$ is a \emph{Nash equilibrium}
if, for all $i \in P$ and strategies 
$\sigma_i' \in \Delta(A_i)$,
$$u_i(\sigma_i, \vec{\sigma}_{-i}) \ge u_i(\sigma'_i,
\vec{\sigma}_{-i}).$$ 
\end{definition}

\begin{definition}\label{def:correlated-eq}
  In a normal-form game $\Gamma = (P, A, U)$, a distribution
    $p \in \Delta(A)$ is a \emph{correlated equilibrium} (CE) if,
  for all  
  $i \in P$
and all $a'_i, a''_i \in A_i$ such that $p(a'_i) >
0$,  $$\sum_{\vec{a} \in \Delta(A) : a_i = a'_i} u_i(a'_i,
\vec{a}_{-i})p(\vec{a} \mid a_i = a'_i) \ge \sum_{\vec{a} \in \Delta(A)
    : a_i = a'_i} u_i(a_i'', 
\vec{a}_{-i})p(\vec{a} \mid a_i = a'_i).$$ 
\end{definition}

Intuitively, for a distribution $p$ over 
action profiles
to be a
correlated equilibrium, it cannot be worthwhile
for player $i$ to deviate from 
$a_i$ to $a_i'$
if $i$ knows only that 
$\vec{a}$ is sampled from distribution $p$ (and $a_i$).
Note that all Nash equilibria are correlated equilibria as well.

An \emph{extensive-form game} $\Gamma$ is a tuple $(P, G, M, U, R)$, where
\begin{itemize}
  \item
    $P = \{1, \ldots, n\}$ is the set of players.
\item $G$ is a
game 
tree
in which each node represents a possible state of the game and
each of the edges going out of a node an action that a player
can perform in the state corresponding to that node. (In the sequel,
we refer to the edges as actions.)
\item
  $M$ is a 
function that associates with each 
non-terminal
node of $G$ (we let 
$G^\circ$
denote the 
non-terminal
nodes) a player in $P$
that describes which player moves at each of these nodes; if $M(v) =
i$, then all the edges going out of node $v$ must correspond to
actions that player $i$ can play.
\item
$U = (u_1, \ldots, u_n)$ 
is an $n$-tuple of 
utility
functions from the leaves of $G$ to $\mathbb{R}$.
\item
$R = (\sim_1, \ldots, \sim_n)$ is an $n$-tuple of equivalence
relations on the nodes of  
$G^\circ$, one for each player.
Intuitively, if $v \sim_i v'$, then nodes $v$
and $v'$ are indistinguishable to player $i$.
Clearly, if $v
\sim_i v'$ for some $i \in [n]$, 
then the set of possible actions that can be performed at $v$ and $v'$
must be identical and $M(v) = M(v')$
(for otherwise, $i$ would have a basis for
distinguishing $v$ and $v'$).  
\end{itemize}
The equivalence relation $\sim_i$ induces a partition of $G^\circ$ into
equivalence classes called \emph{information sets} (of player $i$).
It is more standard in game theory to define $\sim_i$ as an
equivalence relation only on the nodes where $i$ moves.  We define it
on all 
non-terminal
nodes because if $K$ is a subset of players, we are
then able to define $\sim_K$ as the the intersection of $\sim_i$ for 
$i \in K$.  
Intuitively, $v \sim_K v'$ if the agents in 
$K$ cannot distinguish $v$ from $v'$, 
even if they pool all their information together.
We use the equivalence relation $\sim_K$ when defining $k$-resilient
sequential equilibrium.

A \emph{(pure) strategy}
$s_i$
for player $i$ in an extensive-form game
$\Gamma$ is a function
that maps each 
node $v$ where $i$ moves
to an action that can be performed at $v$,
with the constraint that if two nodes are indistinguishable to $i$, then $s_i$ must choose the same action on both. More precisely, if 
$v \sim_i v'$ then
$s_i(v) = s_i(v')$.
Each pure strategy profile $\vec{s}$
induces a unique path from the root of the tree to some
terminal node
$\ell \in G \setminus G^\circ$ where, given node 
$v$
on the path, the
next node in the path is 
the node reached from $v$ by the edge (action) $s_{M(v)}(v)$.
The payoff of player $i$
given strategy $\vec{\sigma}$ is given by $u_i(\ell)$, although we
also write $u_i(\vec{\sigma})$ for simplicity. 
A
\emph{behavioral strategy} for player $i$ 
allows randomization. More precisely, a behavioral strategy
is a function $s_i$ that maps each node 
$v$
such that 
$M(v) = i$
to a  distribution over
$i$'s possible actions at 
$v$ (again, with the requirement that $s_i(v) = s_i(v')$ if $v \sim_i v'$).
We denote by $u_i(\vec{\sigma})$ the
expected payoff of player $i$  if players play (behavioral) strategy 
profile $\vec{\sigma}$. 

Nash equilibrium is defined for extensive-form games just as it is in
strategic-form games.
\fullv{
Given a NE $\vec{\sigma}$, the \emph{equilibrium path} consists of
the nodes of $G$ that can be reached with positive probability
using $\vec{\sigma}$. 
}
It is well known that in extensive-form games, Nash equilibrium does
not always 
describe what intuitively would be a reasonable play.
\fullv{
For instance,
consider the following game for two players in which $(p_1,p_2)$
describes the payoff of players 1 and 2 respectively: 
\begin{center}
    \includegraphics[]{Example1.pdf}
\end{center}
In this case, the strategy profile in which
player 1 plays $B$ at node $a$ and player 2 plays $L$ at node $b$ and
$R$ at $c$ 
is a Nash equilibrium.
However, the reason that it is better for player 1 to play
$B$ is that if she plays $A$, then player 2 would play $L$ and they
would both get a utility of $-5$. This means that player 1 is being
influenced by an irrational threat, since if 1 plays $A$, it is in
player 2's best interest to play $R$ instead of $L$. In order 
to avoid these situations, we can extend the notion of Nash
equilibrium to \emph{subgame-perfect Nash equilibrium}, in which,
roughly speaking, the
strategy must be a Nash equilibrium, not only in $\Gamma$, but in all
the of subgames of $\Gamma$ as well. In
the example above, the only subgame-perfect equilibrium is the one
given by $\sigma_1(a) = A$, $\sigma_2(b) = R$ and $\sigma_2(c) = R$.

Unfortunately, subgame-perfect equilibrium may not be well-defined
if  players have nontrivial information sets.
Consider the following game, where $b$ and $b'$ are in the same
information set for player 2, as are $c$ and $c'$.
\begin{center}
    \includegraphics[]{Example2.pdf}
\end{center}
At node $b'$, player 2 would play $L$; 
at $b$, he would play $R$.
For player 2 to decide its
move, it must have a belief about the probability of being at each
node of the information set, and this belief must be consistent with
the strategy $\vec{\sigma}$ being played. For example, if
$\sigma_1(a) = B$, then
if player 2 is in information set $\{b, b'\}$,
then 2 would be sure that the true node is $b'$.
    }
      \shortv{We focus here on what is perhaps the most common
        refinement of Nash equilibrium considered in extensive-form
        games: \emph{sequential equilibrium}.  Roughly speaking, a
        strategy in a sequential equilibrium if it is a best response
        even to deviations off the equilibrium path.  To make this
        precise, we must consider player's beliefs off the equilibrium
        path.}
\fullv{        
  We capture the players' beliefs by using a \emph{belief system} $b$, }
\shortv{ We capture these  by using a \emph{belief system} $b$, }
which is
a function from information sets $I$ to a probability
distribution over nodes in $I$. Intuitively, if $I$ is an information
set for player $i$, then $b$ represents
how likely it is for $i$ to be at each of the nodes of $I$.
Note that if $\vec{\sigma}$
is \emph{completely mixed}, which means that all actions are taken with positive probability, then
$b$ is uniquely determined by Bayes' rule.
More generally, we 
 say that a belief 
 system $b$ is \emph{consistent}
 with a strategy $\vec{\sigma}$ if there exists a
 sequence  $\vec{\sigma}^1, \vec{\sigma}^2,  \ldots$ of
 completely-mixed strategies that converges to $\vec{\sigma}$ such that the
 beliefs 
induced by Bayes' rule converge to $b$.
Given an extensive-form game $\Gamma = (P,G, M, U, R)$, a strategy
profile $\vec{\sigma}$, a belief system $b$, and an information set
$I$ for player $i$, let $u_i(\vec{\sigma}, I, b)$  denote $i$'s
expected utility conditional on being in information set $I$ and having
belief system $b$, if $\vec{\sigma}$ is played.

\begin{definition}[\cite{KW82}]\label{def:seq-eq}
A pair $(\vec{\sigma},b)$ consisting of a strategy profile
$\vec{\sigma}$ and a belief system $b$ consistent with $\vec{\sigma}$
is a \emph{sequential equilibrium}
if, for all players
$i$, all information sets $I$ of player $i$, and all strategies
$\tau_i$ for player $i$, $$u_i(\vec{\sigma}, I, b) \ge u_i((\tau_i,
\vec{\sigma}_{-i}), I, b).$$
\end{definition}

Even though it is standard to define a sequential equilibrium as a
pair $(\vec{\sigma}, b)$ consisting of strategy profile and a belief
system, for convenience we also say that a strategy profile
$\vec{\sigma}$ is a sequential equilibrium if there exists a belief
system $b$ such that $(\vec{\sigma}, b)$ is a sequential
equilibrium. 

\subsection{Bayesian games}\label{sec:bayesian}

In all of the previous definitions, the utility of each player is
assumed to be common knowledge (\emph{perfect information}). However,
this is not always the case. Bayesian games capture this idea by
assuming that each player $i$ has a type 
$t_i \in T_i$ sampled from a
distribution $q \in \Delta (T)$, where $T = (T_1, \ldots, T_n)$, and
that the utility $u_i$ of $i$ is not only a function of the action
profile being played, but also of its type $t_i$.
Formally, a \emph{Bayesian games} is a tuple $(P, T, q, A, U)$, where,
as in  normal-form games, 
$P$, $A$, and $U$ 
are  the set of players, their 
actions, 
and their utility
functions, respectively;
$T$ is the set  
of possible type profiles, and $q$ is a distribution in $\Delta(T)$.

A \emph{strategy} in a Bayesian game for player $i$ is a map $\mu_i :
T_i \rightarrow \Delta(A_i)$. Intuitively, 
a strategy in a Bayesian game tells player $i$ how to choose its
action given its type. Since the distribution $q$ is common knowledge, given a  strategy profile $\vec{\mu}$ in $\Gamma$, the expected
utility of a player $i$ 
is

\begin{equation}\label{eq:expected-payoff}
u_i(\vec{\mu}) = \sum_{t_i \in T_i} q(t_i) \sum_{\vec{t}}
q(\vec{t} \mid t_i) u_i(\vec{\mu}(\vec{t})),
\end{equation}
where 
$u_i(\vec{\mu}(\vec{t}))$ denotes the expected utility of player
$i$ when an action profile 
is chosen
according to 
$(\mu(\vec{t}))$.
This allows us to define Bayesian Nash equilibrium as follows:

\begin{definition}\label{def:bayesian-nash}
In a Bayesian Game $\Gamma = (P,T, q, A, U)$, a strategy profile
$\vec{\mu} := (\mu_1, \ldots, \mu_n)$ 
is a \emph{(Bayesian) Nash equilibrium} if, 
for
all players $i$ and all strategies $u'_i$ for $i$,
$$u_i(\vec{\mu}) \ge u_i(\vec{\mu}_{-i}, \mu'_i).$$
\end{definition}

We can generalize correlated equilibrium to
Bayesian games as
follows. We can view 
a correlated equilibrium as a distribution $p \in \Delta(A)$ such that
if a trusted mediator samples an action profile $\vec{a} \in p$ and
sends action $a_i$ to each player $i$, it is always better for $i$ to
play $a_i$ rather than something else. In a Bayesian game,
the mediator instead samples the action
profile from a distribution that depends on the type profile.
More precisely, suppose that players send their types to a
trusted mediator, the mediator samples an action profile
$\vec{a}$ from a distribution 
$\mu(\vec{t})$
that depends on the type
profile $\vec{t}$ received,
and then sends action $a_i$ to each player $i$. We
say that 
$\mu$ 
is a
\emph{communication equilibrium} if it is optimal 
for the players to (a) tell their true type to the mediator, and (b)
play the action sent by the mediator.
The following definitions make this precise.

\begin{definition}\label{def:comm-eq}
\cite{F86,Myerson86}
  A 
  correlated strategy profile
  $\mu : T \rightarrow \Delta(A)$ is a \emph{communication equilibrium} of
    Bayesian Game $\Gamma = (P,T, q, A, U)$ if, for all $i \in P$, all
    $t_i \in T_i$, all
    $\psi: T_i \rightarrow T_i$
    and all $\varphi : A_i \rightarrow
        A_i$, we have that $$\sum_{\vec{t}_{-i} \in T_{-i}} \sum_{\vec{a} \in A} q(t_{-i},
t_i)\mu(\vec{a} \mid \vec{t}_{-i}, t_i)u_i(\vec{t}_{-i}, t_i,\vec{a})
\ge $$ $$\sum_{\vec{t}_{-i} \in T_{-i}} \sum_{\vec{a} \in A} q(t_{-i},
t_i)\mu(\vec{a} \mid \vec{t}_{-i}, \psi(t_i))u_i(\vec{t}_{-i},
t_i,\vec{a}_{-i}, \varphi(a_i)).$$ 
\end{definition}

We can combine the notions of extensive-form game and Bayesian game in
the obvious way to get \emph{extensive-form Bayesian games}: we start
with an extensive form game, add a type space $T$ and a commonly known
distribution $q$ on $T$, and then 
have the utility function  depend on the type profile as well the 
terminal node
reached. We leave formal details to the reader.

\section{VSS and CC: a review}\label{sec:tools2}

In this section, we present the two main primitives used to construct
the strategies required for the 
proof of Theorem~\ref{thm:main}.
The purpose of these primitives is that agents are
able to distribute information between them (verifiable secret
sharing), and compute any relevant data from the information shared
without learning anything besides the data itself (circuit computation).

Intuitively, verifiable secret sharing is designed to distribute
between the players a given piece of information $y$ in some finite
field $\mathbb{F}_q$ in such a way that each agent $i$ knows only a
part $y_i$ of the secret that does not reveal any information about
the secret by itself. In fact, with $k$-resilient verifiable secret
sharing, even if $k$ agents collude and put their information
together, they still cannot deduce anything about the shared secret,
but any subset of $k+1$ agents can compute the secret with no error
given their pieces of information. To achieve these properties, the
idea is that the agent $i$ that knows the secret $y$ shares it
using Shamir's secret-sharing scheme \cite{S79}: $i$ gives each
agent $j$ a value $y_j$ such that there exists a polynomial $p$ of
degree $k$ such that $p(j) = y_j$ for each $j$ and $p(0) =
y$. It is easy to check that, if $p$ is chosen uniformly at random,
then the values $y_i$ received by any subset of $k$ agents is
uniformly distributed in $\mathbb{F}_q^k$, and thus convey no
information about $y$. However, any subset of $k+1$ agents can
reconstruct the secret $y$ by interpolating their $k+1$ points. Note
that,
since up to $k$ agents may misreport their values,
players cannot just simply reconstruct the secret using the first 
$k+1$ points they receive.
In fact, it can be shown that at least $3k+1$ players are required in
order to reconstruct the secrets with no error. Suppose that there
exists a mechanism that splits the information about a certain secret
between the players in such a way that (a) if no player defects, every
subset of $k+1$ players can combine their information shares in order to
reconstruct the secret with no error, and (b) no subset of $k$ or less
players can learn anything about the secret, even if they combine
their information shares. If at most $k$ players defect, honest
players can do the following to guarantee that the reconstructed
secret is indeed correct. They first wait until they find a
subset $X$ of $2k+1$ players such that the information shares of all
subsets of $k+1$ players in $X$ define the same secret $s$. If this
occurs, they can guarantee that the original secret is $s$. To see
this, note that at most $k$ players can misreport their
shares. Therefore, at least one of the subsets of $k+1$
players in $X$ 
consists of only  honest players, which means that $s$ must be
the correct secret. For this strategy to work, it is necessary that
honest players can somehow guarantee that they will always find such a
subset $X$, and this is why at least $3k+1$ players are
necessary.\footnote{Note that this is not a complete proof since it
shows only that $3k+1$ players are necessary for this particular
strategy to work. For a complete proof, see~\cite{ADH07, GH21}.}
Since at most $k$ players defect, the subset $X$ defined by the
players that do not defect satisfies all the desired
properties. Applying this idea to Shamir's secret-sharing scheme, it
follows that players have to wait until they find a polynomial that
interpolates $2k+1$ of the received shares.  
\commentout{
In fact, they must wait until receiving at
least $2k+1$ points that lie on the same polynomial $p$ of degree
$k$. In that case, even if $k$ players misreport their values, at least
$k+1$ of the points in $p$ are guaranteed to be reported by honest
players, which ensures that $p$ defines the correct secret. At a high
level, this is why this primitive requires that the number of players
$n$ satisfies $n > 3k$: even if $k$ players misreport their values, at
least $2k+1$ of those values are shared by honest players, and thus
all players are guaranteed to eventually find $2k+1$ points that lie
on the same polynomial of degree $k$,
which allows them to
reconstruct the secret correctly.
}
\commentout{
if a secret is distributed this way, then there is a way to
reconstruct the original secret $y$ even if $k$ of the agents report
wrong values: if $2k+1$ points lie on the same polynomial $p$, then
$p$ is guaranteed to be the correct polynomial since at least $k+1$ of
the points reported are correct. This means that if the number of
players $n$ satisfies $n > 3k$, they just have to find a subset
of $2k+1$ points that lie on the same polynomial to 
reconstruct the secret with no error. This subset of points is
guaranteed to exist since at most $k$ points are reported incorrectly.  
}

In order to share a secret using Shamir's secret-sharing scheme, it is
unfortunately not enough to  have the sender choose a polynomial $p$ uniformly
at random such that $p(0) = y$ and send each agent $i$ the value of
$p(i)$. This protocol is vulnerable when the sender is malicious: it
could generate $n$ points that do not lie on the same polynomial, and
agents could never reconstruct the secret. In order to guarantee that
the shared points define some secret, players have to follow a
non-trivial protocol that checks the validity of the points without
leaking information about the underlying secret.  In the synchronous
setting,  Ben-Or,  Goldwasser, and Wigderson \nciteyear{BGW88} provide
a $k$-resilient VSS protocol if $n > 3k$; in the asynchronous
setting, Ben-Or, Canetti, and Goldreich \nciteyear{BCG93} provide a 
$k$-resilient VSS protocol if $n > 4k$. Intuitively, in asynchronous systems,
$n > 4k$ is necessary since $k$ of the $4k+1$ players might
be delayed arbitrarily (and are indistinguishable from deviating
players that didn't send a message), and $k$ might send the wrong values,
which leaves $2k+1$ points to interpolate without error. A more
rigorous description of the properties of VSS can be found in
Appendix~\ref{sec:VSS}. 

In circuit computation, players generate the shares of either the sum
or product of two shared secrets without leaking information about the
secrets themselves. More precisely, suppose that secrets $y,z \in
\mathbb{F}_q$ are shared among the players.
Then $k$-resilient circuit          %
computation that allows each player to compute its share
of $y+z$ or $yz$ without leaking any information about $y$ or $z$ to
any coalition of at most $k$ players. By successively using this
primitive, players can compute the shares of any secret that is the
output of a circuit with addition and multiplication gates given the
inputs. This, in turn, allows agents to compute the shares of the
output of any function whose domain and range are finite, since,
without loss of generality, we can 
then take the domain and range to be $\mathbb{F}_q$ for some $q$
sufficiently large, and any function
$f: \mathbb{F}_q \rightarrow \mathbb{F}_q$ can be viewed as a
polynomial on $\mathbb{F}_q$, and so 
involves only addition and multiplication.
An implementation of a $k$-resilient circuit computation protocol is
given by Ben-Or, Goldwasser, and Wigderson \nciteyear{BGW88} 
in the case of synchronous systems, and by Ben-Or, Canetti, and
Goldreich \nciteyear{BCG93} in the case of
asynchronous systems. A more detailed description of the
properties of CC can be found in Appendix~\ref{sec:CC}. 

\subsection{Verifiable secret sharing}\label{sec:VSS}

Verifiable secret sharing allows a player (the sender) to securely
distribute shares of a private value $v$ among all players.
A protocol $\vec{\sigma}$ is a \emph{$k$-resilient implementation of VSS in
synchronous systems} if, for all senders $s \in [n]$, all coalitions $K
\subseteq [n]$ of at most $k$ players, and all strategies
$\vec{\tau}_K$ for players in $K$, the following holds of
$(\vec{\sigma}_{-K}, \vec{\tau}_K)$: 
\begin{itemize}
    \item [(a)] All players not in $K$ terminate.
          \item [(b)] If $s^i$ is $i$'s output, then there exists a
      polynomial $r$ of degree $k$ such that $r(i) = s^i$ for all $i
      \not \in K$. Moreover, if the sender $s$ is not in $K$, $r(0)$
      is its input $v$. 
    \item [(c)] If the sender $s$ is not in $K$, the distribution over
      histories $h_K^v$ of players in $K$ when $s$ shares $v$ is
the same for all $v \in \mathbb{F}_q$.  
\end{itemize}

Properties (a) and (b) guarantee that, regardless of what a coalition of at
most $k$ players does, all honest players terminate and their
outputs are consistent even when the sender is deviating from the
protocol, which means that even if the sender decides not to share an
input or not send messages at all, all honest players eventually
agree on some shares. Property (c) guarantees that no coalition of $k$
players can learn anything about the secret being shared if the sender
is not in the coalition. 

In asynchronous systems, properties (a) and (b) cannot be guaranteed
simultaneously, since a player that sends no messages is
indistinguishable from a player that is being delayed by the
scheduler. In this case we have that: 
\begin{itemize}
    \item [(a1)] If the sender is not in $K$, all players not in $K$ terminate.
    \item [(a2)] If a player $i \not \in K$ terminates, all players not in $K$ terminate.
    \item [(b)] If $S$ is the subset of players not in $K$ that
      terminate and $s^i$ is the output of player $i \in
      S$, then  there exists a polynomial $r$ of degree $k$ such that
      $r(i) = s^i$ for all $i \in S$. Moreover, if the sender $s$ is
      not in $K$, $r(0)$ is its input $v$. 
          \item [(c)] If the sender $s$ is not in $K$, then for all schedulers,
      the distribution over histories $h_K^v$ of players in $K$ when
            $s$ shares $v$ is the same for all $v \in
      \mathbb{F}_q$. 
\end{itemize}

Note that, in the asynchronous setting, a player might not terminate if
the sender 
deviates from the protocol. However, the protocol must guarantee that
the outputs of the honest players are consistent, which means that
either they all terminate or they all don't, and, if they do, their
outputs should all lie on the same polynomial of degree $k$ (which
encodes the secret shared if the sender is honest). Finally,
note that (c) is quantified over all schedulers, since histories
depend on the scheduling of the players and messages.


Since the shares of a given secret $s$ are just evaluations of a
polynomial $p$ of degree $k$ at $1, 2, \ldots, n$, if a player $i$
learns at least $k+1$ of the shares of $s$, it can compute $p$ by
interpolating the given shares and then compute the secret $s :=
p(0)$. Note, 
however, that reconstructing a secret this way is not $k$-resilient,
since if only one player lies about its share to $i$, then $i$ would
not reconstruct the right secret. To deal with this concern, $i$
waits for $2k+1$ shares that lie on the same polynomial $p$.
This allows $i$ to be sure that, if at most $k$ players lie, then at
least $k+1$ of 
those shares are truthful, and thus define the correct polynomial $p$.

\subsection{Circuit computation}\label{sec:CC}

Circuit computation allows players to compute the shares of the sum or
product of two inputs without revealing any information about
them. More precisely, given shares $\{s_v^i\}_{i \in [n]}$ of $v$ and
$\{s_{v'}^i\}_{i \in [n]}$ of $v'$, we say that $\vec{\sigma}$ is a
$k$-resilient implementation of a \emph{multiplication gate} in both
synchronous and asynchronous systems if, for all coalitions $K
\subseteq [n]$ of at most $k$ players, and all strategies
$\vec{\tau}_K$ for players in $K$ the following holds for
$(\vec{\sigma}_{-K}, \vec{\tau}_K)$: 
\begin{itemize}
    \item [(a)] All players not in $K$ terminate.
          \item [(b)] If $s^i$ is $i$'s output, then there exists a
      polynomial $r$ of degree $k$ such that $r(0) = vv'$ and $r(i) =
      s^i$ for all $i \not \in K$. 
    \item [(c)] The distribution over histories $h_K^v$ of players in
      $K$ is the same for all sets 
      $\{s_v^i\}_{i \in [n]}$ and $\{s_{v'}^i\}_{i \in [n]}$ of shares.  
\end{itemize}

As in the case of VSS, properties (a) and (b) guarantee termination
and consistency among the outputs of honest players, while property
(c) guarantees that no coalition of at most $k$ players can deduce
anything from $v$, $v'$, or $vv'$ by running $\vec{\sigma}$. The
definition of a $k$-resilient \emph{addition gate} is analogous. A
$k$-resilient circuit computation protocol consists of a $k$-resilient
implementation of an addition gate and a $k$-resilient implementation
of a multiplication gate. 

Note that, in the case of the sum, it is easy to check that if
$\{s_v^i\}_{i \in [n]}$ are shares of some secret $v$ and
$\{s_{v'}^i\}_{i \in [n]}$ are shares of $v'$, then $\{s_v^i +
s_{v'}^i\}_{i \in [n]}$ are shares of $v + v'$. Thus, implementing a
$k$-resilient addition gate is trivial.
Implementing a $k$-resilient multiplication gate requires more
work.
(See \cite{BCG93,BGW88} for details.)


\section{Consensus protocols}\label{sec:consensus}

The construction of the strategy used for the proof of
Theorem~\ref{thm:main2} requires the use of a \emph{consensus protocol}.
In a $k$-resilient consensus protocol, each player has an input which
is a single bit (i.e., either 0 or 1), which intuitively represents
their preference, and players must 
agree on a bit with the following
guarantees if $k$ or fewer players deviate: 
\begin{itemize}
    \item All honest players eventually terminate with the same output.
    \item If an honest player terminates with output $o$, at least one
            honest player preferred $o$ to $1-o$.
\end{itemize}

More precisely, suppose that there $n$ players and each player $i$ has
a preference $x_i \in \{0,1\}$. A consensus protocol should satisfy
the following properties for all subsets $K \subseteq [n]$ such that
$|K| \le k$ and all strategies $\vec{\tau}_K$ for players in $K$: 

\begin{itemize}
    \item All players $i \not \in K$ eventually terminate. Moreover, if $i$ terminates with output $o$, all players not in $K$ are guaranteed to eventually terminate with output $o$.
    \item If $x_i = x_j$ for all players $i, j \not \in K$, then, if a player $i \not \in K$ terminates, it terminates with output $x_i$ (i.e., if all players not in $K$ have the same preference, then this preference is their output).
\end{itemize}

Note that by running several consensus protocols in succession, players can agree on an element of a set of arbitrary (finite) size. Bracha~\nciteyear{Br84} provided a consensus protocol that is $k$-resilient in both synchronous and asynchronous systems if $n > 3k$.

\commentout{
\subsection{Reconstructing a secret}

\commentout{
Since shares are just evaluations of a polynomial at different values of $x$, a set $S = \{s_{s_1}, \ldots, s_{j_t}\}$ of shares is $k$-consistent if there exists a polynomial $p$ of degree $k$ such that $p(j_\ell) = s_{j_\ell}$ for all $\ell \le t$. If $|S| > k$, this polynomial is uniquely defined and the value defined by $S$ is $p(0)$.
}
Since the shares of a given secret $s$ are just evaluations of a
polynomial $p$ of degree $k$ at $1, 2, \ldots, n$, if a player $i$
learns at least $k+1$ of the shares of $s$, it can compute $p$ by
interpolating the given shares and then compute $s := p(0)$. Note,
however, that reconstructing a secret this way is not $k$-resilient,
since if only one player lies about its share to $i$, then $i$ would
reconstruct an erroneous secret. To deal with this concern, $i$
waits for $2k+1$ shares that lie in the same polynomial $p$.
This allows $i$ to be sure that, if at most $k$ players lie, then at
least $k+1$ of 
those shares are truthful, and thus define the correct polynomial $p$.
}

\bibliographystyle{ACM-Reference-Format}
\bibliography{joe,game1}

\end{document}